\documentclass[aos,preprint]{imsart}
\usepackage{xcite}
\usepackage{xr}
\makeatletter
\newcommand*{\addFileDependency}[1]{
  \typeout{(#1)}
  \@addtofilelist{#1}
  \IfFileExists{#1}{}{\typeout{No file #1.}}
}
\makeatother

\newcommand*{\myexternaldocument}[1]{%
    \externaldocument{#1}%
    \addFileDependency{#1.tex}%
    \addFileDependency{#1.aux}%
}
\myexternaldocument{supplement}

\usepackage{pdfpages}

\RequirePackage{amsthm,amsmath,amsfonts,amssymb}
\RequirePackage[authoryear]{natbib} %% uncomment this for author-year bibliography
\RequirePackage[colorlinks,citecolor=blue,urlcolor=blue]{hyperref}
\RequirePackage{graphicx}% uncomment this for including figures

\RequirePackage{multirow}
\usepackage{tikz}
\usepackage{xargs}
\usepackage{mathrsfs}

\usepackage{enumitem}
\usepackage{MnSymbol}
\usepackage{xcolor}
\usepackage{todonotes}

\usepackage[textwidth=405bp,textheight=646bp]{geometry}

\startlocaldefs

\theoremstyle{plain}
\newtheorem{theorem}{Theorem}[section]
\newtheorem{proposition}{Proposition}[section]
\newtheorem{lemma}{Lemma}[section]

\newtheorem{corollary}{Corollary}[section]
\theoremstyle{remark}
\newtheorem{assumption}{Assumption}[section]
\theoremstyle{remark}
\newtheorem{remark}{Remark}[section]
\global\long\def\expect{\mathbf{E}}%
\global\long\def\prob{\mathrm{Pr}}%
\global\long\def\real{\mathbb{R}}%
\global\long\def\Op{O_{p}}%
\global\long\def\manifold{\mathcal{M}}%
\global\long\def\covarop{\mathcal{C}}%
\global\long\def\asympgt{\gtrsim}%
\global\long\def\asympeq{\asymp}%
\global\long\def\diffop{\mathrm{d}}%
\newcommandx\tangentspace[2][usedefault, addprefix=\global, 1=\manifold]{T_{#2}#1}%
\newcommandx\lpnorm[3][usedefault, addprefix=\global, 1=r, 2=]{\|#3\|_{\mathcal{L}^{#1}}^{#2}}%
\newcommandx\lp[1][usedefault, addprefix=\global, 1=p]{\mathcal{L}^{#1}}%

\global\long\def\innerprod#1#2{\langle#1,#2\rangle}%
\global\long\def\Log{\mathrm{Log}}%
\global\long\def\Exp{\mathrm{Exp}}%
\newcommandx\vfnorm[3][usedefault, addprefix=\global, 1=$\mu$, 2=]{\|#3\|_{#1}^{#2}}%
\newcommandx\vfinnerprod[2][usedefault, addprefix=\global, 1=$\mu$]{\llangle#2\rrangle_{#1}}%
\global\long\def\tdomain{\mathcal{T}}%
\newcommandx\opnorm[3][usedefault, addprefix=\global, 1=$\mu$\textbf{}, 2=]{\vertiii{#3}_{#1}^{#2}}%
\newcommandx\fronorm[2][usedefault, addprefix=\global, 1=]{|#2|_{F}^{#1}}%
\global\long\def\dtispd{\mathrm{Sym}_{\star}^+(3)}
\def\pt{\mathcal P}
\def\vpt{\mathscr P} % parallel transport on the vector bundle

\usepackage{graphicx}
\graphicspath{{./fig/}}

\endlocaldefs

\begin{document}
	
\newgeometry{textwidth=405bp,textheight=646bp,footskip=2cm}

\begin{frontmatter}
%%%%%%%%%%%%%%%%%%%%%%%%%%%%%%%%%%%%%%%%%%%%%%
%%                                          %%
%% Enter the title of your article here     %%
%%                                          %%
%%%%%%%%%%%%%%%%%%%%%%%%%%%%%%%%%%%%%%%%%%%%%%
\title{Intrinsic Riemannian Functional Data Analysis for Sparse Longitudinal Observations}
%\title{A sample article title with some additional note\thanksref{T1}}
\runtitle{iRFDA for Sparse Data}
%\thankstext{T1}{A sample of additional note to the title.}

\begin{aug}
%%%%%%%%%%%%%%%%%%%%%%%%%%%%%%%%%%%%%%%%%%%%%%
%%Only one address is permitted per author. %%
%%Only division, organization and e-mail is %%
%%included in the address.                  %%
%%Additional information can be included in %%
%%the Acknowledgments section if necessary. %%
%%%%%%%%%%%%%%%%%%%%%%%%%%%%%%%%%%%%%%%%%%%%%%
\author[A]{\fnms{Lingxuan} \snm{Shao} \ead[label=e1]{shao-14@pku.edu.cn}},
\author[B]{\fnms{Zhenhua} \snm{Lin} \ead[label=e2]{linz@nus.edu.sg}}
\and
\author[C]{\fnms{Fang} \snm{Yao} \ead[label=e3]{fyao@math.pku.edu.cn}}
%%%%%%%%%%%%%%%%%%%%%%%%%%%%%%%%%%%%%%%%%%%%%%
%% Addresses                                %%
%%%%%%%%%%%%%%%%%%%%%%%%%%%%%%%%%%%%%%%%%%%%%%

\renewcommand{\thefootnote}{\fnsymbol{footnote}}
\address{\small for the Alzheimer's Disease Neuroimaging Initiative\footnotemark[1]}
\address[A]{School of Mathematical Sciences, Center for Statistical Science, Peking University }
\address[B]{Department of Statistics and Data Science, National University of Singapore }
\address[C]{School of Mathematical Sciences, Center for Statistical Science Peking University}
\end{aug}

\begin{abstract}
	A new framework is developed to intrinsically analyze sparsely observed Riemannian functional data. It features four innovative components: a frame-independent covariance function, a smooth vector bundle termed \emph{covariance vector bundle}, a parallel transport and a smooth bundle metric on the covariance vector bundle. The introduced intrinsic covariance function links estimation of covariance structure to smoothing problems that involve raw covariance observations  derived from sparsely observed Riemannian functional data, while the covariance vector bundle provides a rigorous mathematical foundation for formulating such smoothing problems. The parallel transport and the bundle metric together make it possible to measure fidelity of fit to the covariance function. They also play a critical role in quantifying the quality of estimators for the covariance function. As an illustration, based on the proposed framework, we develop a local linear smoothing estimator for the covariance function, analyze its theoretical properties, and provide numerical demonstration via simulated and real datasets. The intrinsic feature of the framework makes it applicable to not only Euclidean submanifolds but also manifolds without a canonical ambient space.
\end{abstract}



\renewcommand{\thefootnote}{\fnsymbol{footnote}}
\footnotetext[1]{\tiny Data used in preparation of this article were obtained from the Alzheimer's Disease Neuroimaging Initiative (ADNI) database (\url{adni.loni.usc.edu}). As such, the investigators within the ADNI contributed to the design and implementation of ADNI and/or provided data but did not participate in analysis or writing of this report. A complete listing of ADNI investigators can be found at: \url{http://adni.loni.usc.edu/wp-content/uploads/how_to_apply/ADNI_Acknowledgement_List.pdf}}
\renewcommand{\thefootnote}{\arabic{footnote}}

\end{frontmatter}
%%%%%%%%%%%%%%%%%%%%%%%%%%%%%%%%%%%%%%%%%%%%%%
%% Please use \tableofcontents for articles %%
%% with 50 pages and more                   %%
%%%%%%%%%%%%%%%%%%%%%%%%%%%%%%%%%%%%%%%%%%%%%%
%\tableofcontents

%%%%%%%%%%%%%%%%%%%%%%%%%%%%%%%%%%%%%%%%%%%%%%
%%%% Main text entry area:

\section{Introduction}
Functional data are nowadays commonly encountered in practice and have been extensively studied in the literature; for instance, see the monographs \cite{Ramsay2005,Ferraty2006,Hsing2015,Kokoszka2017}, as well as the survey papers \cite{Wang2016} and \cite{Aneiros2019}, for a comprehensive treatment on  functional data analysis. These classic endeavors study functional data in which functions are real- or vector-valued, and thus are challenged by data of functions  that do not take values in a vector space. Such data emerge increasingly often, partially due to the rapid development of modern technologies. For example, in the longitudinal study of diffusion tensors, as the tensor measured at a time point is represented by a $3\times 3$ symmetric positive-definite matrix (SPD), the study results in a collection of SPD-valued functions. The space of SPD matrices is not a vector space, and in particular, the usual Euclidean distance on it suffers from the ``swelling effect'' which introduces artificial and undesirable inflation of variability in data analysis \citep{Arsigny2007}. 
Specialized distance functions \citep{Pennec2006,Dryden2009} or metrics \citep{Moakher2005,Arsigny2007,Lin2019Riemannian} are required to alleviate or completely eliminate the swelling effect. These metrics turn the space of SPD matrices of a fixed dimension into a nonlinear Riemannian manifold. Data in the form of Riemannian manifold valued functions are termed Riemannian functional data and  modeled by  Riemannian random processes which are random processes taking values in Riemannian manifolds \citep{Lin2019}.

\newgeometry{textwidth=405bp,textheight=646bp}

Since the mean and covariance functions are two of the most fundamental concepts in functional data analysis, as many downstream analyses depend on them, it is of particular importance to generalize them to Riemannian functional data. For the mean function, the generalized counterpart is the well established  Fr\'echet mean function that is adopted in \cite{Dai2018,Dai2019,Lin2019} and is an extension of Fr\'echet mean. The concept of Fr\'echet mean in turn generalizes the usual mean of random vectors to manifold-valued random elements, and has been studied in depth by  \cite{Bhattacharya2003,Bhattacharya2005,Afsari2011,Schoetz2019,Pennec2019}. Related to estimation of Fr\'echet mean function is regression on manifold-valued non-functional data that was investigated by \cite{Pelletier2006, Shi2009,Steinke2010,Fletcher2013,Hinkle2014,Cornea2017}, and more broadly, on metric-space valued data by \cite{Hein2009,Faraway2014,Petersen2019,Lin2019b}, among others.

The genuine challenge comes from modeling and estimating the covariance structure. To tackle nonlinearity of the Riemannian manifold, a strategy commonly employed in the literature is to transform data from the manifold into tangent spaces via Riemannian logarithmic maps, and then to model the covariance via the transformed data. Specifically, at each time point, the associated observations are transformed into the tangent space at the Fr\'echet mean at that time point. Although tangent spaces of a manifold are linear spaces and thus provide the desired vector structure, there is one  issue to resolve: \textit{Different tangent spaces are distinct vector spaces and thus their tangent vectors are incomparable, but the covariance involves random tangent vectors from different tangent spaces.}
More specifically, the value of the covariance function at a time pair $(s,t)$  involves observations at both $s$ and $t$, and in the manifold setting, the observations at these time points are often transformed into tangent vectors of distinct tangent spaces. %; see Figure \ref{fig:cov-illustration} for an illustration.

%\begin{figure}[t]
%	\begin{center}
%		\includegraphics[scale=1.2]{covariance.pdf}
%	\end{center}
%	\caption{Illustration of random tangent vectors in different tangent spaces. The curve across the manifold $\manifold$ represents the Fr\'echet mean function $\mu$ defined in \eqref{eq:mu-def}. The two parallelograms represent the tangent spaces at $\mu(s)$ and $\mu(t)$, respectively. An (random) observation associated with the time $s$ is often transformed into the tangent space $T_{\mu(s)}\manifold$ at $\mu(s)$; the transformed observation is represented by the tangent vector $U$ in the figure. Similarly, the tangent vector $V$ in the tangent space $T_{\mu(t)}\manifold$ at $\mu(t)$ represents a transformed observation associated with the time point $t$.  The tangent vectors $U$ and $V$ are  incomparable, since they reside in distinct tangent spaces. The vectors $e_1$ and $e_2$ represent an orthonormal basis in  $T_{\mu(s)}\manifold$ while $e_1^\prime$ and $e_2^\prime$ form an orthonormal basis of  $T_{\mu(t)}\manifold$. The tangent vectors $U$ and $V$ can be respectively represented by their (random) coefficients with reference to the bases $\{e_1,e_2\}$ and $\{e_1^\prime,e_2^\prime\}$. The coefficients are real-valued random vectors and directly comparable, e.g., their covariance can be defined in the usual way.\label{fig:cov-illustration}}
%\end{figure}

The above issue is especially pronounced for sparsely observed Riemannian functional data.  A common strategy well established in the Euclidean setting for sparse functional data is to smooth the discrete and noisy raw covariance function \citep{Yao2005a,Cai2010,Li2010,Zhang2016}. However, there are fundamental difficulties in extending this {seemingly simple} strategy to the manifold setting. First, as previously mentioned, the covariance function involves  tangent vectors from different tangent spaces, so that an appropriate definition of covariance between two incomparable random tangent vectors is in order.  Second, for the smoothing strategy to work, the underlying covariance function shall possess certain regularity of smoothness, such as continuity or differentiability. However, it is challenging to define and quantify such regularity for covariance of Riemannian functional data. This problem is unique to sparsely observed data; when data are fully observed or sufficiently dense so that each trajectory can be individually recovered, the sample covariance operator serves as an estimate for the covariance structure \citep{Lin2019}, which does not require smoothing.

To overcome the above difficulties, in this paper we develop a novel framework to model and estimate the covariance when Riemannian functional data are sparsely and noisily recorded. The proposed framework features four innovative components. 
\begin{itemize}
	\item First, an intrinsic covariance function is developed to characterize covariance between random tangent vectors from distinct tangent spaces. Such covariance function is invariant to manifold parameterization, frame selection and embedding, and is made possible by considering the covariance of two random tangent vectors in different tangent spaces as a linear operator that maps one tangent space into the other. This covariance function does not require reference to a frame and thus is fundamentally different from the covariance function of coefficients with respect to a frame in \cite{Lin2019}.
	\item Second, we construct a novel smooth vector bundle from the manifold, termed \emph{covariance vector bundle}, to provide an appropriate mathematical foundation for intrinsic  quantification of the regularity of the proposed covariance function, such as continuity, differentiability and smoothness. For example, it makes  statements like ``find a \emph{smooth} covariance function that minimizes the mean squared error'' sensible. In addition, covariance function estimation amounts to smoothing data located in a smooth vector bundle. Although there is a rich literature on smoothing  Riemannian manifold-valued data, the study on smoothing data in a vector bundle is still in its infancy.
	\item Third, a parallel transport on the covariance vector bundle is developed from the intrinsic geometry of the manifold, which also induces a covariant derivative on the bundle. The covariant derivative allows intrinsic definition of derivatives of a function taking values in the covariance vector bundle. Such derivatives are often needed when one analyzes theoretical properties of a smoothing procedure. The parallel transport also enables one to move the raw observations into a common vector space in which classic smoothing methods may apply. 
	\item Fourth, a smooth bundle metric is constructed and plays an essential role in measuring the fidelity of fit to data during estimation and quantifying the quality of an estimator. It is derived from the intrinsic geometry of the underlying Riemannian manifold and utilizes the Hilbert--Schmidt inner product of linear operators between two potentially different Hilbert spaces. Such inner product, mathematically well established \citep[e.g., Definition 2.3.3 and Proposition B.0.7 by][]{Prevot2007}, is less seen in statistics; the common one is usually for operators that map a Hilbert space into itself.	
\end{itemize}

The intrinsic covariance function and the covariance vector bundle together pave the way for intrinsically smoothing the observed raw covariance function, while the parallel transport and the bundle metric are critical for developing an estimation procedure for sparsely observed Riemannian functional data. As an illustration, we propose an estimator for the covariance function based on local linear smoothing and establish the point-wise and uniform convergence rates of the estimator under various designs, while emphasize that other smoothing techniques such as spline smoothing are also applicable. Other contributions include extending the invariance principle of \cite{Lin2019} to the sparse design and connecting holonomy theory to statistics via Lemma \ref{lem:c_tildec} that might be of independent interest.

Our work is clearly set apart from existing endeavors in the scarce literature that emerge only in recent years. \cite{Su2014} first represented each trajectory by its normalized velocity curve  and then transported the velocity vectors into a common tangent space.  \cite{Zhang2018b} specifically considered spherical trajectories and developed a data transformation geared to the spherical geometry,  while \cite{Dai2018}, \cite{Lin2019} and \cite{Dubey2020} studied trajectories on a more general manifold or metric space. All of these works assume fully observed functions and thus require no smoothing. \cite{Dai2019} proposed to smooth the sparsely observed raw covariance by embedding the manifold into a Euclidean space. This approach of using an embedding, although making adaption of classic smoothing techniques to the manifold setting straightforward, does not readily apply to manifolds without a canonical embedding. Moreover, the results and their interpretations may be tied to the chosen embedding; see Section~\ref{sec:invariance} of the supplement. In contrast, our framework does not require an embedding and thus circumvents these drawbacks, though it needs to overcome drastically elevated technical challenges.
 
We shall emphasize that, the proposed framework is not to replace explicit parameterization in practice, but to make the statistical outcomes invariant to the parameterization and/or frame adopted in computation. For instance, in Section \ref{sec:simulation} we demonstrate that the proposed method produces identical results under different parameterizations. In particular, the intrinsicality featured by our method refers to requiring no embedding, rather than no parameterization. This makes our framework immediately applicable to manifolds without a canonical embedding. We demonstrate this feature via the manifold of SPD matrices endowed with the affine-invariant metric in our simulation studies; see Section \ref{sec:simulation}.

The rest of the paper is organized as follows. In Section \ref{sec:cov-bundle}, we construct the covariance vector bundle, a parallel transport and a smooth metric on th bundle. In addition, we formulate the intrinsic concept of covariance function for Riemannian functional data. An estimator for the covariance function from sparsely observed Riemannian functional data is described in Section \ref{sec:estimation}, and its theoretical properties are given in Section \ref{sec:theory}. Simulation studies are placed in Section \ref{sec:simulation}, followed by an application to longitudinal diffusion tensors in Section \ref{sec:application}. 
%Proofs for results in Section \ref{sec:cov-bundle} are collected in the Appendix, while proofs for other results are deferred to an online Supplementary Material for space economy.
All the proofs are deferred to an online Supplementary Material for space economy.

\section{Covariance vector bundle}\label{sec:cov-bundle}
\subsection{Preliminaries}\label{subsec:preliminaries}
We  briefly review concepts from Riemannian manifolds that are essential for our development at a high level, while relegate all formal definitions to the supplement and refer readers to the introductory text by \citet{Lee1997} for further exposition. 

Let $\manifold$ be a $d$-dimensional smooth manifold, roughly speaking, a space that locally resembles $\real^d$ and is endowed with a smooth structure. A smooth structure is formally described by a (maximal) smooth atlas on $\manifold$, specifically, a collection of pairs $(U_{\alpha},\phi_{\alpha})$ that are indexed by an index set $J$ and satisfy the following conditions:

\begin{itemize}
	\item Each $U_{\alpha}$ is an open subset of $\manifold$ and $\bigcup_{\alpha\in J}U_{\alpha}=\manifold$;
	\item Each $\phi_{\alpha}$ is a bijective continuous map between $U_{\alpha}$ and an open set of $\real^d$;
	\item If $U_{\alpha}\cap U_{\beta}\neq\emptyset$,
	then the {transition map} $\phi_{\alpha}\circ\phi_{\beta}^{-1}:\phi_{\beta}(U_{\alpha}\cap U_{\beta})\rightarrow\phi_{\alpha}(U_{\alpha}\cap U_{\beta})$ 
	is smooth, i.e., infinitely differentiable; we say $\phi_{\alpha}$ and $\phi_{\beta}$
	are {compatible}.
\end{itemize}
The pair $(U_{\alpha},\phi_{\alpha})$ or sometimes $\phi_{\alpha}$
itself is called a {chart} (or {coordinate map}). Intuitively, $\phi_\alpha$ assigns a local coordinate to each point in $U_\alpha$.  Two atlases are {compatible} if their union is
again an atlas (satisfying the above conditions). An atlas
is {maximal} if it contains any other atlas compatible
with it. 

Every point in a $d$-dimensional manifold is associated with a distinct $d$-dimensional vector space, called the tangent space at the point. In addition, any chart $(U_\alpha,\phi_\alpha)$ gives rises to a basis for the tangent space at each point in $U_\alpha$, and the basis smoothly varies with the point with $U_\alpha$. More generally, one can assign to each tangent space a basis. Such an assignment is called a frame.  Tangent spaces at different points of a manifold are conceptually distinct spaces, so that their elements, called tangent vectors, are incomparable; only tangent vectors from the same tangent space are comparable.

A Riemannian manifold is a smooth manifold equipped with a Riemannian metric $\innerprod{\cdot}{\cdot}$ which defines an inner 
product $\innerprod{\cdot}{\cdot}_{p}$ on the tangent space $\tangentspace p$ at each point $p\in\manifold$, with the associated norm denoted by $\|v\|_p=\sqrt{\langle v,v\rangle_p }$ for $v\in \tangentspace p$. The metric, which smoothly varies with $p$, induces a distance function $d_{\manifold}$ on $\manifold$ and turns the manifold into a metric space. A geodesic in a Riemannian manifold is a constant-speed curve of which every sufficiently small segment is the shortest path connecting the endpoints of the segment. At each point  $p\in\manifold$ there is an exponential map $\Exp_{p}$ that maps tangent vectors at $p$ onto the manifold $\manifold$. In particular, for each $v\in\tangentspace p$, $\gamma_v(t):=\Exp_p(tv)$ defines a geodesic. The inverse of $\Exp_p$, when it exists, is called the Riemannian logarithmic map at $p$ and denoted by $\Log_p$. %; see the left panel of Figure \ref{fig:parallel-transport} for a graphical illustration.

In statistical analysis, it is desirable to compare the tangent vectors from different tangent spaces. To this end, one may transport the tangent vectors into a common tangent space in which tangent vectors can be directly compared by vector subtraction. For a Riemannian manifold, there is a unique (parallel) transport associated with the Riemannian metric and realized by Levi--Civita connection. In this paper, unless otherwise stated, parallel transport is performed along shortest geodesics between two points $y$ and $z$, denoted by $\pt_{y}^z$, which moves tangent vectors from the tangent space $T_y\manifold$ to $T_z\manifold$ in a smooth way and meanwhile preserves the inner product.%; see the right panel of Figure \ref{fig:parallel-transport} for a graphical illustration. %In addition, inner products are also preserved. A canonical choice for the path is the shortest path connecting the two points. Such parallel transport, {denoted by $\pt_{y}^z$}, is uniquely determined by the points $y,z$ and the Riemannian metric tensor.

A smooth vector bundle, denoted by $\pi:\mathcal E\rightarrow \manifold$ or simply $\mathcal E$, consists of a base smooth manifold $\manifold$, a smooth manifold  $\mathcal E$ called total space, and a smooth bundle projection $\pi$, such that for every $p\in\manifold$, the fiber $\pi^{-1}(p)$ is a $k$-dimensional real vector space, and there is an open neighborhood $U\subset \manifold$ of $p$ and a diffeomorphism $\Phi:\pi^{-1}(U)\rightarrow U\times\real^k$ satisfying the property that for all $z\in U$, $(\pi\circ\Phi^{-1})(z,v)=z$ for all $v\in\real^k$ and the map $v\mapsto\Phi^{-1}(z,v)$ is a linear isomorphism between $\real^k$ and $\pi^{-1}(z)$. The map $\Phi$ is called a local trivialization. %As graphically illustrated in Figure \ref{fig:vector-bundle}, a vector bundle locally resembles the product space $U\times \real^k$ for some integer $k$. 
A prominent example of vector bundle is the space composed by the union of all tangent spaces of a manifold, which is called the tangent bundle of the manifold, where the tangent space at each point is a fiber. To identify different fibers, one can introduce a parallel transport $\vpt$ on a vector bundle along a curve $\gamma$ on the base manifold. Such parallel transport must satisfy the following axioms: 1) $\vpt_{p}^{p}$ is the identity map on $\pi^{-1}(p)$ for all $p\in\manifold$, 2) $\vpt_{\gamma(u)}^{\gamma(t)}\circ\vpt_{\gamma(s)}^{\gamma(u)}=\vpt_{\gamma(s)}^{\gamma(t)}$, and 3) the dependence of $\vpt$ on $\gamma$, $s$ and $t$ are smooth. The parallel transport $\pt$ introduced previously for a Riemannian manifold is indeed a parallel transport on the tangent bundle. In Section \ref{subsec:bundle} we shall construct a new type of vector bundle and a parallel transport on it. If for each fiber in a smooth vector bundle there is an inner product and the inner product smoothly varies from fiber to fiber, then the inner products are collectively referred to as a smooth bundle metric. The aforementioned Riemannian metric is indeed a smooth bundle metric on the tangent bundle.

\subsection{Riemannian functional data}
Functional data in which each function takes values in a Riemannian manifold are termed Riemannian functional data and modeled by the Riemannian random process \citep{Lin2019}. Specifically, let $\manifold$ be a $d$-dimensional Riemannian manifold and $X$  a $\manifold$-valued random process indexed by a compact domain $\tdomain\in\real$, i.e., $X:\tdomain\times \Omega \rightarrow \manifold$, where $\Omega$ is the sample space of the underlying probability space. In reality, measurements of $X$ are often corrupted by noise. To accommodate this common practice, we assume that the actual observable process is $Y$ which is indexed by the same domain $\tdomain$.

The process $X$ is said to be of second order, if for each $t\in\tdomain$, $F(p,t)=\expect d_{\manifold}^2(X(t),p) < \infty$ for some $p\in\manifold$ and hence for all $p\in\manifold$ due to the triangle inequality. The minimizer of $F(p,t)$, if it exists, is called the Fr\'echet mean of $X(t)$ and denoted by $\mu(t)$, i.e., 
\begin{equation}\label{eq:mu-def}
\mu(t):=\underset{p\in\manifold}{\arg\min}\,F(p,t).
\end{equation}
%Similarly, we define the Fr\'echet mean of $Y(t)$
%\begin{align*}
%\mu^{\ast}(t):=\underset{p\in\manifold}{\arg\min}\,F^{\ast}(p,t),\quad\text{where } F^{\ast}(p,t):=\expect d_{\manifold}^{2}(Y(t),p).
%\end{align*}
The concept of the Fr\'echet mean generalizes the mean from the Euclidean space to the Riemannian manifold and plays an important role in analysis of data residing in a Riemannian manifold. Under fairly general conditions, the Fr\'echet mean exists  and is unique \citep{Bhattacharya2003,Sturm2003,Afsari2011}, for instance, when the manifold is of nonpositive sectional curvature \citep[p.146,][]{Lee1997} or data are located in a small subspace of the manifold. Formally, we make the following assumption.
\begin{assumption}
	\label{ass:muexist}  
	%The Fr\'{e}chet means $\mu(t)$ and $\mu^{\ast}(t)$ exist and are unique for each $t\in\tdomain$.
	The Fr\'{e}chet mean functions of $X$ and $Y$  exist and are unique.
\end{assumption}

As the manifold $\manifold$ is not a vector space, it is challenging to directly study the processes $X$ and $Y$. A common strategy is to transform them into tangent spaces, in which the vector structure can facilitate the analysis, via Riemannian logarithmic maps. This requires an additional assumption to ensure the well-posedness of the Riemannian logarithmic maps. For simplicity, we assume the following sufficient condition, which can be relaxed by a delicate formulation via cut locus\footnote{See Section \ref{sec:rm} of the online supplementary material for a precise definition.}.  %To state the assumption, define $\mathfrak c(p)\define\inf_{v\in\tangentspace p,\|v\|_p=1}\mathrm{cut}(p,v)$ and $\mathfrak{D}_p^\epsilon\define\{q:d_{\manifold}(p,q)<\mathfrak{c}(p)-\epsilon\}$, where $\mathrm{cut}(p,v)$ is the cut time defined in Section \ref{subsec:preliminaries}. It is clear from the definition that $\mathfrak{D}_p^{\epsilon}\subset \mathscr D_p$ for each $p\in\manifold$ and $\epsilon\geq 0$. %The following assumption imposes Lipschitz continuity on the function $\mathfrak c$, as well as the range of $X$ and $Y$. 
%\begin{assumption}
%	\label{ass:Xnearmu}  %The function $\mathfrak c$ is locally Lipschitz continuous when it is finite. [LIN: this can be proved.]
%	There exists a constant $\epsilon_0>0$ such that
%	\begin{enumerate}[label=\textup{(\alph*)}]
%		\item\label{ass:Xnearmu:a} $\prob\{X(t)\in  \mathfrak{D}_{\mu(t)}^{\epsilon_0}$ for all $t\in\tdomain\}=1$, and
%		\item $\prob\{Y(t)\in \mathfrak{D}_{\mu(t)}^{\epsilon_0}$ for all $t\in\tdomain\}=1$.
%	\end{enumerate}
%	%where $A^{\epsilon_0}$ denotes the set $\bigcup_{p\in A}\{q\in \manifold:d_{\manifold}(p,q)<\epsilon_0\}$, and by convention $\emptyset^{\epsilon_0}=\emptyset$.
%\end{assumption}
\begin{assumption}\label{ass:Xnearmu} 
	There exists a geodesically convex\footnote{A subset in a Riemannian manifold is geodesically convex if for any two points in the subset there is a unique shortest geodesic that is contained in the subset and connects the points.} subset $\mathcal Q\subset\manifold$ such that $X(t),Y(t)\in\mathcal Q$ for all $t\in\tdomain$.
\end{assumption}
If the manifold is of nonpositive sectional curvature, $\mathcal Q$ can be taken to be $\manifold$ and thus the above assumption becomes superfluous.  %since $\mathfrak D_{\mu(t)}=\manifold$ and $\mathfrak c(p)=\infty$ for all $p\in\manifold$. 
Examples of manifolds of this kind include hyperbolic manifolds, tori and the space of symmetric positive-definite matrices endowed with the affine-invariant metric \citep{Moakher2005}, Log-Euclidean metric \citep{Arsigny2007} or Log-Cholesky metric \citep{Lin2019Riemannian}. {An example $\mathcal Q$ for Riemannian manifolds of positive sectional curvature is the hypersphere $\mathbb S^k=\{(x_0,\ldots,x_k)\in\real^{k+1}:x_0^2+\cdots+x_k^2=1\}$ or the positive orthant $\mathcal Q=\{(x_0,\ldots,x_k)\in\mathbb S^{k}: x_j\geq 0\text{ for all }j=0,\ldots,k\}$, which has applications in compositional data analysis \citep{Dai2018}, where $k$ is a positive integer.}

Under Assumptions~\ref{ass:muexist} and~\ref{ass:Xnearmu}, the Riemannian logarithmic maps $\Log_{\mu(t)}\{X(t)\}$ and $\Log_{\mu(t)}\{Y(t)\}$ are well defined. In addition, we can further model the observed process by $$Y(t)=\Exp_{\mu(t)}(\Log_{\mu(t)}\{X(t)\}+\varepsilon(t)),$$ where $\varepsilon(t)\in \tangentspace{\mu(t)}$ represents the random noise in the tangent space, is independent of $X$, and satisfies $\expect \varepsilon(t)=0$ and $\Exp_{\mu(t)}\varepsilon(t)\in \mathcal Q$. With this setup, the mean functions of $X$ and $Y$ are the same, in analogy to the Euclidean case; see Lemma \ref{lem:mu=muast} below. % i.e. $\mu(t)=\mu^{\ast}(t)$. It can be also shown that the Fr\'echet mean function of $Y$ defined in this way is $\mu$.

\begin{lemma}
	\label{lem:mu=muast}
	If Assumptions~\ref{ass:muexist} and \ref{ass:Xnearmu} hold, and $\manifold$ is complete and simply connected, then $\expect\{\Log_{\mu(\cdot)} X(\cdot)\}=0$. In addition, if $Y(t)=\Exp_{\mu(t)}(\Log_{\mu(t)}X(t)+\varepsilon(t))$, where $\varepsilon(t)\in \tangentspace{\mu(t)}$ is independent of $X$ and satisfies $\expect \varepsilon(\cdot)=0$,  then $\mu$ is also  the Fr\'echet mean function of $Y$.
\end{lemma}

Now we are ready to model sparsely observed Riemannian functional data. First, the sample functions $X_{1},\ldots,X_{n}$ are considered as i.i.d. realizations of $X$. However,  accessible are their noisy copies $Y_1,\ldots,Y_n$, rather than $X_1,\ldots,X_n$. To further accommodate the practice that functions are recorded at discrete points, we assume each $Y_{i}$ is only observed at $m_{i}$ time points $T_{i,1},\ldots,T_{i,m_{i}}\in\tdomain$. Specifically, the observed data are $\{(T_{ij},Y_{ij})\in \tdomain\times \manifold:1\leq i\leq n,1\leq j\leq m_{i}\}$ with  $Y_{ij}=\Exp_{\mu(T_{ij})}(\Log_{\mu(T_{ij})}\{X_i(T_{ij})\}+\varepsilon_{ij})$, where the centered random elements  $\varepsilon_{ij}\in\tangentspace{\mu(T_{ij})}$ are independent of each other and also independent of $\{X_{i}:1\leq i\leq n\}$.

\subsection{Covariance function of Riemannian functional data}
In addition to the Fr\'echet mean function, the covariance structure of Riemannian functional data is essential for downstream analysis, for instance, functional principal component analysis. In \cite{Lin2019} the covariance structure is modeled by the covariance operator of  $\Log_{\mu(\cdot)}X(\cdot)$ from the random element perspective \citep[Chapter 7,][]{Hsing2015} and also by the covariance function of $\Log_{\mu(\cdot)}X(\cdot)$ with respect to a frame\footnote{See Section \ref{sec:rm} in the online supplementary material for a precise definition.}. The covariance operator is not computationally friendly to sparse data, while the frame-dependent covariance function is not compatible with most smoothing methods; see Section~\ref{sec:invariance} of the supplement for more details. 

To develop a frame-independent intrinsic concept of the covariance function from the perspective of stochastic processes, we first revisit the covariance between two centered random vectors $U$ and $V$. When they are in a common Euclidean space, it is classically defined as the matrix $\expect (UV^\top)$. When $U$ and $V$ are in different general inner product spaces $\mathbb{U}$ and $\mathbb{V}$, a matrix representation of the covariance is definable if one picks an orthonormal basis for each of $\mathbb{U}$ and $\mathbb{V}$. To eliminate the dependence on the orthonormal bases, we take an operator perspective to treat the covariance $C$ of $U$ and $V$ as a linear operator between $\mathbb U$ and $\mathbb V$ characterized by 
$$
\innerprod{Cu}v_{\mathbb{V}}:=\expect(\innerprod{U}{u}_{\mathbb U}\innerprod{V}{v}_{\mathbb V}),\quad \forall u\in\mathbb{U},v\in\mathbb{V},
$$
where $\innerprod{\cdot}{\cdot}_{\mathbb{U}}$ and $\innerprod{\cdot}{\cdot}_{\mathbb{V}}$ denote the inner products of $\mathbb{U}$ and $\mathbb V$, respectively. 
To simplify the notation, we write $C=\expect(U\otimes V)$.

Observe that $\Log_{\mu(\cdot)}X(\cdot)$ ($\Log_{\mu}X$ for short) is a random vector field along the curve $\mu$ with $\expect(\Log_{\mu}X)=0$ {according to Lemma~\ref{lem:mu=muast}} (also Theorem 2.1 of \cite{Bhattacharya2003}).  
Given that $\Log_{\mu(s)}X(s)\in{}\tangentspace{\mu(s)}$ and $\Log_{\mu(t)}X(t)\in\tangentspace{\mu(t)}$, and both $\tangentspace{\mu(s)}$ and $\tangentspace{\mu(t)}$ are Hilbert spaces, we define the covariance function for $X$ by % at time $(s,t)$ by %$\covarop(s,t):\tangentspace{\mu(s)}\rightarrow \tangentspace{\mu(t)}$
\begin{equation}\label{eq:cov-func}
\covarop(s,t):=\expect\{\Log_{\mu(s)}X(s) \otimes \Log_{\mu(t)}X(t)\},\quad\text{for }(s,t)\in \tdomain^2.
\end{equation}
This covariance function is clearly \emph{independent} of any frame or coordinate system. This feature fundamentally and distinctly separates  \eqref{eq:cov-func} from the frame-dependent covariance function (5) defined in \cite{Lin2019} for the coordinate of $\Log_{\mu}X$ with respect to a frame along the mean function. Moreover, \eqref{eq:cov-func} can be viewed as the intrinsic covariance function of the covariance operator $\mathbf C$ proposed in \cite{Lin2019}. Specifically, under some measurability or continuity assumption on $X$ and the condition that $\expect \int_{\tdomain}\|\Log_{\mu(t)}X(t)\|_{\mu(t)}^2<\infty$, the process $\Log_{\mu}X$ can be regarded as a random element in the Hilbert space 
$$
\mathscr{T}(\mu):=\{Z:Z(\cdot)\in \tangentspace{\mu(\cdot)}, \int_{\tdomain}\langle Z(t),Z(t)\rangle^{2}_{\mu(t)}\diffop t<\infty\}
$$  
endowed with the inner product 
$\llangle Z_{1},Z_{2}\rrangle_{\mu}:=\int_{\tdomain}\innerprod{Z_{1}(t)}{Z_{2}(t)}_{{\mu(t)}}\diffop t$ for $Z_{1},Z_{2}\in\mathscr{T}(\mu)$. The covariance operator $\mathbf C:\mathscr{T}(\mu)\rightarrow \mathscr{T}(\mu)$ for $X$ can be defined by 
\begin{equation}\label{eq:cov-op}
\llangle \mathbf C u,v\rrangle_{\mu}:=\expect(\llangle \Log_\mu X,u\rrangle_{\mu}\llangle \Log_\mu X,v\rrangle_{\mu}) \quad\text{ for } u,v\in\mathscr{T}(\mu). 
\end{equation}
The following theorem, which generalizes Theorem 7.4.3 of \cite{Hsing2015} to Riemannian random processes, shows that the proposed covariance function induces the covariance operator $\mathbf C$.
\begin{theorem}
	\label{pro:C=Cst}
	Let  $\covarop(\cdot,\cdot)$ and $\mathbf C$ be defined in \eqref{eq:cov-func} and \eqref{eq:cov-op}, respectively. Suppose that $X$ is
	mean-square continuous, i.e., $\lim_{k\rightarrow\infty}\expect d^2(X(t_k),X(t))=0$ for any $t\in\tdomain$ and any sequence $\{t_k\}$ in $\tdomain$ converging to $t$. Also assume that $X$ is jointly measurable, i.e, $X:\tdomain\times \Omega\rightarrow\manifold$ is measurable with respect to the product $\sigma$-field on $\tdomain\times\Omega$, where $\Omega$ is the sample space of the underlying probability space. Then under Assumptions \ref{ass:muexist} and \ref{ass:Xnearmu}, for all $t\in\tdomain$ and $u\in\mathscr{T}(\mu)$, we have
	$$
	(\mathbf C u)(t)=\int_{\tdomain}\covarop(s,t)u(s)\diffop s.
	$$ 
\end{theorem}
In light of this result, in the sequel we often use the same notation $\covarop$ to denote both the covariance operator and the covariance function in \eqref{eq:cov-func}. %In fact, the covariance function in \eqref{eq:cov-func} is the \emph{coordinate-independent} covariance function of the covariance operator of \cite{Lin2019}. 
The proposed covariance function enables estimating the covariance operator $\mathbf C$ through estimating $\covarop(s,t)$ for each $(s,t)\in \tdomain\times\tdomain$ in a frame-independent fashion. The frame-independent feature is of particular importance to deriving a frame-invariant estimate in the more practical scenario that only discrete and noisy observations are available so that  smoothing is desirable; see Section \ref{sec:estimation} for more detail.

\subsection{The vector bundle of covariance and parallel transport}\label{subsec:bundle}
To estimate the covariance function in \eqref{eq:cov-func}, it seems rather intuitive to perform smoothing over the raw covariance  
\begin{equation}\label{eq:raw-cov}
	\hat{\covarop}_{i,jk}:=\Log_{\hat{\mu}(T_{ij})}Y_{ij}\otimes \Log_{\hat{\mu}(T_{ik})}Y_{ik} \in \mathbb{L}(\hat{\mu}(T_{ij}),\hat{\mu}(T_{ik})),
\end{equation}
 where $\hat\mu$ is an estimate of $\mu$ to be detailed in Section~\ref{sec:estimation}. The first challenge encountered is that these raw observations $\hat{\covarop}_{i,jk}$ do not reside in a common vector space. This also gives rise to the second challenge in defining the key concept of \emph{smoothness} of the function $\covarop$ and its estimate. 
To circumvent these difficulties, we consider the spaces $\mathbb{L}(p,q)$ consisting of all linear maps from $\tangentspace p$ to $\tangentspace q$, and their disjoint union $\mathbb{L}=\bigcup_{(p,q)\in\manifold^2}\mathbb L(p,q)$. Then $\hat{\covarop}_{i,jk}$ are encompassed by the space $\mathbb L$, and in addition, the covariance function $\covarop$ is now viewed as an $\mathbb L$-valued function. Although the space $\mathbb L$ is not a vector space so that the smoothness is not definable in the classic sense, we observe that $\mathbb L$ comes with a canonical smooth structure induced by the manifold $\manifold$, and continuity, differentiability and smoothness relevant to statistics can be defined with reference to this smooth structure as follows.

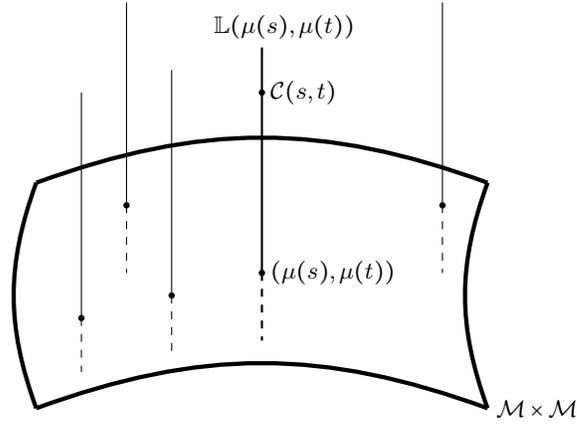
\begin{figure}[t]
	\centering
	\begin{tikzpicture}[scale=0.6]
	\draw[ultra thick][-] (0,0) to [out=20,in=160] (10,0);
	\draw[ultra thick][-] (0,5) to [out=20,in=160] (10,5);
	\draw[ultra thick][-] (0,0) to [out=110,in=250] (0,5);
	\draw[ultra thick][-] (10,0) to [out=110,in=250] (10,5);
	\node[right] at (10,0) {$\manifold\times\manifold$};
	%\draw[thick][-] (2,1.5) to [out=10,in=160] (7,1.5);
	%\draw[thick][-] (2,4) to [out=10,in=160] (7,4);
	%\draw[thick][-] (2,1.5) to [out=100,in=250] (2,4);
	%\draw[thick][-] (7,1.5) to [out=100,in=250] (7,4);
	%\node[right] at (7,4) {$\mu\times\mu$};
	\draw[][-] (2,4.5) to (2,9);
	\draw[][dashed][-] (2,4.5) to (2,3);
	\fill (2,4.5) circle (2pt);
	\draw[thick][-] (5,3) to (5,8);
	\draw[thick][dashed][-] (5,3) to (5,1.5);
	\draw[][-] (3,2.5) to (3,7.5);
	\draw[][dashed][-] (3,2.5) to (3,1.2);
	\fill (3,2.5) circle (2pt);
	\draw[][-] (1,2) to (1,7);
	\draw[][dashed][-] (1,2) to (1,0.8);
	\fill (1,2) circle (2pt);
	\draw[][-] (9,4.5) to (9,9);
	\draw[][dashed][-] (9,4.5) to (9,3);
	\fill (9,4.5) circle (2pt);
	\fill (5,3) circle (2pt);
	\node[right] at (5,3) {$(\mu(s),\mu(t))$};
	\fill (5,7) circle (2pt);
	\node[right] at (5,7) {$\covarop(s,t)$};
	\node[above] at (5.5,8) {$\mathbb{L}(\mu(s),\mu(t))$};
	\end{tikzpicture}
	\caption{Illustration of the vector bundle $\mathbb L$. 
		The thick bending parallelogram presents the product manifold $\manifold\times\manifold$ and the vertical lines represent fibers. The value of $\covarop(s,t)$ is located within the fiber $\mathbb L(\mu(s),\mu(s))$ at the point $(\mu(s),\mu(t))\in\manifold\times\manifold$.}
	\label{fig:section}
\end{figure}

We first observe that $\mathbb{L}$ is a vector bundle on $\manifold\times \manifold$, with $\pi:\mathbb{L}\rightarrow\manifold\times\manifold$ defined by  $\pi(\mathbb{L}(p,q))=(p,q)$ being the bundle projection and $\mathbb{L}(p,q)$ being the fiber attached to the point $(p,q)\in\manifold\times\manifold$; see Figure~\ref{fig:section} for a graphical illustration. To define the smoothness structure on $\mathbb{L}$ induced by the manifold $\manifold$, let $\{(U_\alpha,\phi_\alpha):\alpha\in J\}$ for an index set $J$ be an atlas of $\manifold$. Recall that each chart $(U_\alpha,\phi_\alpha)$ gives rise to a smoothly varying basis of $\tangentspace p$ for each $p\in U_\alpha$. Such basis is denoted by $B_{\alpha,1}(p),\ldots,B_{\alpha,d}(p)$. %With respect to the basis, each tangent vector $v\in \tangentspace p$ is represented by its coefficients $(v^1,\ldots,v^d)\in \real^d$, i.e., $v=v^jB_{\alpha,j}(p)$, where the Einstein summation convention\footnote{In Einstein summation convention, when an index variable appears in a single term of an equation/formula exactly twice (one as an upper index and the other as a lower index) and is not otherwise defined in the equation, it implies summation of that term over all the values of the index.} is used to avoid clutter of notations. Here,  the superscript in each $v^j$ is understood as an upper index rather power, with the purpose to use Einstein summation convention. 
For $(p,q)\in U_\alpha\times U_\beta$, the tensor products $B_{\alpha,j}(p)\otimes B_{\beta,k}(q)$, $j,k=1,\ldots,d$, form a basis for the space $\mathbb L(p,q)$. % the tensor product space $\tangentspace p\otimes \tangentspace q=\{u\otimes v:u\in\tangentspace p,v\in\tangentspace q\}=\mathbb L(p,q)$. 
Each element $v\in \mathbb L(p,q)$ is then  identified with its coefficients $v_{jk}$ with respect to this basis, i.e., $v=\sum_{j,k=1}^dv_{jk}B_{\alpha,j}(p)\otimes B_{\beta,k}(q)$. % and the Einstein summation convention\footnote{In Einstein summation convention, when an index variable appears in a single term of an equation/formula exactly twice (one as an upper index and the other as a lower index) and is not otherwise defined in the equation, it implies summation of that term over all the values of the index.} is used to avoid clutter of notations. Here,  the superscript in each coefficient $v^{jk}$ is understood as an upper index rather power, with the purpose to use Einstein summation convention. 
For each $U_\alpha\times U_\beta$, we define the map $\varphi_{\alpha,\beta}(p,q,\sum_{j,k=1}^dv_{jk}B_{\alpha,j}(p)\otimes B_{\beta,k}(q))=(\phi_\alpha(p),\phi_\beta(q),v_{11},v_{12},\ldots,v_{dd})\in\real^{2d+d^2}$, for $(p,q)\in U_\alpha\times U_\beta$. 
The collection $\{(\pi^{-1}(U_\alpha\times U_\beta),\varphi_{\alpha,\beta}):(\alpha,\beta)\in J^2\}$ indeed is a smooth atlas that turns  $\mathbb L$ into a smooth manifold. Moreover, $\mathbb L$ is a \emph{smooth} vector bundle with the projection map $\pi$ and the  local trivializations $\Phi_{\alpha,\beta}:\pi^{-1}(U_\alpha\times U_\beta)\rightarrow U_\alpha\times U_\beta\times \real^{d^2}$ defined as $\Phi_{\alpha,\beta}(p,q,\sum_{j,k=1}^dv_{jk}B_{\alpha,j}(p)\otimes B_{\beta,k}(q))=(p,q,v_{11},v_{12},\ldots,v_{dd})$.
\begin{theorem}\label{thm:L}
	The collection $\{(\pi^{-1}(U_\alpha\times U_\beta),\varphi_{\alpha,\beta}):(\alpha,\beta)\in J^2\}$ is a smooth atlas on $\mathbb L$. With this atlas, $\mathbb L$ is a smooth vector bundle with %the atlas $\{(\pi^{-1}(U_\alpha\times U_\beta),\varphi_{\alpha,\beta}):(\alpha,\beta)\in I^2\}$, 
	the smooth  projection map $\pi$ and smooth local trivializations $\Phi_{\alpha,\beta}$. In addition, any compatible atlas of  the manifold $\manifold$ gives rise to the same smooth vector bundle $\mathbb L$.
	%smooth manifold. In addition, the bundle projection $\pi$ is smooth and the local trivializations are diffeomorphisms, and consequently, the vector bundle is a smooth vector bundle. Moreover, any compatible atlas of  the manifold $\manifold$ gives rise to the same smooth vector bundle $\mathbb L$.
\end{theorem}

With the above smooth structure, the covariance function $\covarop$ in \eqref{eq:cov-func}, viewed as an $\mathbb L$-valued function, %can be considered as a section of the vector bundle along   $\mu\times\mu:=\{\mu(s)\times\mu(t):(s,t)\in\tdomain\times\tdomain\}\subset \manifold\times\manifold$,  i.e., $\covarop(s,t)\in \mathbb L(\mu(s),\mu(t))$ for each $(s,t)\in\tdomain\times\tdomain$, illustrated graphically in Figure~\ref{fig:section}. In addition, the regularity of $\covarop$, such as the continuity, differentiability or smoothness, is equivalent to the regularity of the section. For instance, we say $\covarop$ 
is said to be $\kappa$-times continuously differentiable in $(s,t)$, if $(\mu(s),\mu(t))\in U_\alpha\times U_\beta$ implies that  $\varphi_{\alpha,\beta}(\mu(s),\mu(t),\covarop(s,t))$ is $\kappa$-times continuously differentiable in $(s,t)$, where we recall that $\{(\pi^{-1}(U_\alpha\times U_\beta),\varphi_{\alpha,\beta}):(\alpha,\beta)\in J^2\}$ is a smooth atlas on $\mathbb L$.  From this perspective, the constructed vector bundle $\mathbb L$ provides a framework to rigorously define the regularity of $\covarop$. In this framework, estimating the covariance function $\covarop$ amounts to smoothing the discrete raw observations $\hat\covarop_{i,jk}$ in the vector bundle $\mathbb{L}$.

Although the vector bundle $\mathbb L$ provides a qualitative framework for defining differentiability or other smoothness regularity, it does not provide a quantitative characterization. Roughly speaking, the smooth vector bundle $\mathbb L$ allows one to check whether $\covarop$ is differentiable or smooth, but not to  measure how rapidly $\covarop$ changes relative to $(s,t)$. In other words, derivatives that quantify the rate of change of the function $\covarop$ at a given pair $(s,t)$ and that are consistent across all compatible atlases for $\mathbb L$ require an additional structure as follows. %To introduce such an additional structure that is naturally induced by the manifold $\manifold$, 
We first introduce the parallel transport on the covariance vector bundle $\mathbb L$ to identify different fibers and to compare the elements from the fibers. 
Suppose that $(p_{1},q_{1}),(p_{2},q_{2})\in \manifold\times\manifold$ and $\gamma(t)=(\gamma_{p}(t),\gamma_{q}(t))$ is the shortest geodesic connecting $(p_{1},q_{1})$ to $(p_{2},q_{2})$. The parallel transport $\vpt_{(p_{1},q_{1})}^{(p_{2},q_{2})}$ from a fiber $\mathbb{L}(p_{1},q_{1})$ to another fiber $\mathbb{L}(p_{2},q_{2})$ is naturally constructed from the parallel transport operators $\pt_{p_{2}}^{p_{1}}$ and $\pt_{q_{1}}^{q_{2}}$ on $\manifold$ by 
\begin{equation}\label{eq:parallel-transport}
(\vpt_{(p_{1},q_{1})}^{(p_{2},q_{2})}C)(u):=\pt_{q_{1}}^{q_{2}}(C(\pt_{p_{2}}^{p_{1}}u)),
\end{equation}
where $C\in \mathbb{L}(p_{1},q_{1})$ and $u\in T_{p_{2}}\manifold$. To distinguish between the parallel transport on the manifold and the one on the vector bundle $\mathbb L$, notationally we use the caliligraphic symbol $\mathcal P$ for the manifold while the script symbol $\mathscr P$ for the bundle. The parallel transport $\mathscr P$ further determines a covariant derivative\footnote{For a definition of the covariant derivative, see Chapter 4 (specifically, Page 50) of \cite{Lee1997} or Section \ref{sec:rm} in the online supplementary material.} on the bundle. 
\begin{theorem}\label{thm:cov-deriv}
	For a tangent vector $V$ of $\manifold\times\manifold$ at $(p,q)$, the map $\nabla_V$ defined by 
	\begin{equation}\label{eq:covariant-deriv}
	\nabla_V W:=\underset{h\rightarrow0}{\lim}\frac{\vpt_{\gamma(h)}^{\gamma(0)}W(\gamma(h))-W(\gamma(0))}{h}:=\frac{\diffop}{\diffop t}\vpt_{\gamma(t)}^{\gamma(0)}W(\gamma(t))\big|_{t=0}
	\end{equation}
	for all differentiable section $W$ is a covariant derivative in the direction of $V$, where $\gamma$ is a smooth curve\footnote{It can be shown that the value $\nabla_V W$ depends on $V$, but not on $\gamma$.} in $\manifold\times\manifold$ with initial point $\gamma(0)=(p,q)$ and initial velocity $\gamma^\prime(0)=V$, and a section is any function $W:\manifold\times\manifold\rightarrow\mathbb L$ satisfying $W(p,q)\in \mathbb L(p,q)$ for all $(p,q)\in\manifold\times\manifold$.
\end{theorem}

The covariant derivative of a section $W$ can be viewed as the first derivative of the section. It quantifies the rate and direction of change of $W$ at each point in $\manifold\times\manifold$. This applies to the covariance function $\covarop$ since it can be viewed as a section along the surface  $\mu\times\mu:\tdomain\times\tdomain\rightarrow\manifold\times\manifold$. In addition, the ``partial derivative'' $\frac{\partial}{\partial s}\covarop(s,t)|_{s=s_0}$ of $\covarop(s,t)$ with respect to $s$ at $s_0$ can be understood as the limit 
$$\underset{h\rightarrow0}{\lim}\frac{\vpt_{\gamma(s_0+h)}^{\gamma(s_0)}\covarop(s_0+h,t)-\covarop(s_0,t)}{h}\in T_{(\mu(s_{0}),\mu(t))}\manifold^2$$
with $\gamma(s)=(\mu(s),\mu(t))$. Furthermore, since the derivative $\frac{\partial}{\partial s}\covarop(s,t)$ is again a section of the vector bundle, one can define the partial derivatives of $\frac{\partial}{\partial s}\covarop(s,t)$, which can be regarded as the second derivatives of $\covarop$. Higher-order  derivatives can be defined in a recursive way.

To further illustrate the parallel transport and the induced covariant derivative on the vector bundle $\pi:\mathbb L \rightarrow\manifold\times\manifold$, consider a simple example in which $\manifold=\real$ and the bundle $\mathbb{L}$ is then parameterized by $(x,y,z)\in\real^2\times\real$. Let $g:\manifold\times\manifold\rightarrow\mathbb L$ be  a smooth section. 
For visualization, we fix $y=0$ and write $f(x)=g(x,0)$. For the smooth function  $f(x)$  shown in Figure~\ref{fig:derivative}, the classic definition of the derivative of $f(x)$ at $x_{1}$ is
$$
\frac{\partial}{\partial x}f(x)\Big|_{x=x_{1}}:=\lim_{x_{2}\rightarrow x_{1}}\frac{f(x_{2})-f(x_{1})}{x_{2}-x_{1}}.
$$
From the perspective of the vector bundle, each point $x$ in the $x$-axis is attached with a fiber $\real_{x}$ which is simply a copy of the $z\text{-axis}=\real$. Since $f(x_{1})\in \real_{x_{1}}$ while $f(x_{2})\in \real_{x_{2}}$, the operation $f(x_{2})-f(x_{1})$ would not be well defined if we did not identify $\real_{x_1}$ with $\real_{x_2}$. The identification between $\real_{x_1}$ and $\real_{x_2}$ is canonical, and nothing else but parallelly transporting $\real_{x_2}$ to $\real_{x_1}$. This inspiring observation applies to general manifolds and covariant derivatives. Specifically, the covariant derivative is defined by parallel transporting  $f(x_{2})$ from the fiber $\mathbb F_{x_2}$ into the fiber $\mathbb F_{x_1}$ and then performing differentiation therein, i.e., 
$$
\frac{\partial}{\partial x}f(x)\Big|_{x=x_{1}}
:=\lim_{x_{2}\rightarrow x_{1}}\frac{\vpt_{x_{2}}^{x_{1}}(f(x_{2}))-f(x_{1})}{x_{2}-x_{1}} \in \mathbb{F}_{x_{1}}.
$$

\begin{figure}[t]
	\centering
	\begin{minipage}[c]{0.45\linewidth}
		\centering
		\begin{tikzpicture}[scale=0.55]
		\draw[ultra thick][<->] (0,5)--(0,0)--(10,0);
		\draw[ultra thick][->] (0,0)--(-1,-1);
		\draw[thick] (1,3) to [out=-50,in=230] (9,5);
		\fill (4,1.91) circle (2pt);
		\fill (7,3.18) circle (2pt);
		\draw[dashed] (4,1.91)--(0,1.91);
		\draw[dashed] (4,0)--(4,5);
		\draw[dashed] (7,3.18)--(0,3.18);
		\draw[dashed] (7,0)--(7,5);
		\draw (4,1.91)--(7,3.18);
		\node[below] at (4,0) {$x_{1}$};
		\node[below] at (7,0) {$x_{2}$};
		\node[left] at (0,1.91) {$z_{1}$};
		\node[left] at (0,3.18) {$z_{2}$};
		\node[right] at (4,1.71) {$f(x_{1})$};
		\node[right] at (7,2.98) {$f(x_{2})$};
		\node[right] at (9,5) {$f(x)$};
		\node[below] at (10,0) {$x$-axis};
		\node[right,rotate=45] at (-1.7,-0.8) {$y$-axis};
		\node[above] at (0,5) {$z\text{-axis}=\real$};
		\node[above] at (4,5) {$\real_{x_{1}}$};
		\node[above] at (7,5) {$\real_{x_{2}}$};
		\end{tikzpicture}
	\end{minipage}
	\qquad\quad\quad
	\begin{minipage}[c]{0.45\linewidth}
		\centering
		\begin{tikzpicture}[scale=0.55]
		%\draw[ultra thick][->] (15,-0.05)--(15,5);
		\draw[ultra thick][->] (15,0) to [out=-10,in=170] (25,0);
		\draw[thick] (16,3) to [out=-50,in=230] (24,5);
		\fill (19,1.91) circle (2pt);
		\fill (22,3.18) circle (2pt);
		\draw[dashed] (19,0)--(19,5);
		\draw[dashed] (19,2.88) to [out=-10,in=185] (22,3.18);
		\draw[dashed] (22,0.3)--(22,5);
		\node[below] at (19,0) {$x_{1}$};
		\node[below] at (22,0) {$x_{2}$};
		\node[right] at (19,1.71) {$f(x_{1})$};
		\node[right] at (22,2.98) {$f(x_{2})$};
		\node[above] at (19,2.9) {$\vpt_{x_{2}}^{x_{1}}(f(x_{2}))$};
		\node[right] at (24,5) {$f(x)$};
		%\node[right] at (25,0) {$x$-axis};
		%\node[above] at (15,5) {$y$-axis=$\real_{0}$};
		\node[above] at (19,5) {$\mathbb F_{x_{1}}$};
		\node[above] at (22,5) {$\mathbb F_{x_{2}}$};
		\end{tikzpicture}
	\end{minipage}
	\caption{Illustration of classic differentiation (left) and general covariant derivative (right).\label{fig:derivative}}
\end{figure}
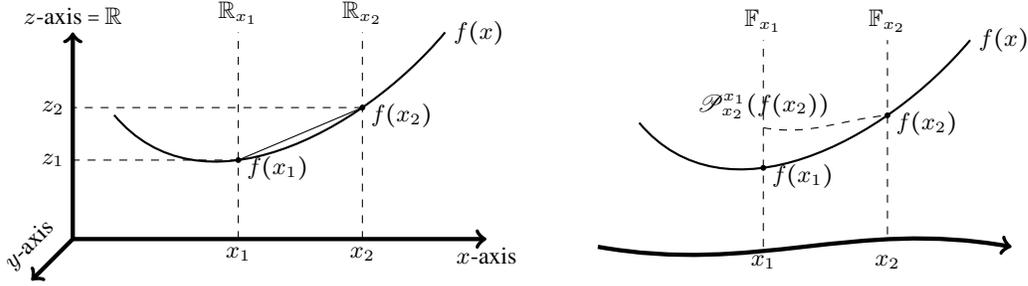

Finally, when smoothing the raw covariance function $\hat\covarop_{i,jk}$, one needs to quantify the discrepancy between the data and the fit. Such discrepancy is often measured by a distance function or inner product on the data space. Fortunately,  the vector bundle $\mathbb L$ comes with a natural bundle metric. Specifically, for any $(p,q)\in\manifold\times\manifold$, the metric $G_{(p,q)}:\mathbb{L}(p,q)\times \mathbb{L}(p,q)\rightarrow \mathbb{R}$ is defined as the Hilbert--Schmidt inner product, i.e., 
\begin{equation}\label{eq:bundle-metric}
G_{(p,q)}(L_{1},L_{2})=\sum_{k=1}^d \langle L_1e_k,L_2e_k\rangle_q \quad\text{ for } L_1,L_2\in \mathbb L(p,q),
\end{equation}
where $e_1,\ldots,e_d$ denotes an orthonormal basis of $\tangentspace{p}$. One can show that the definition \eqref{eq:bundle-metric} does not depend on the choice of the orthonormal basis. In fact, $G$ is a smooth bundle metric and the parallel transport \eqref{eq:parallel-transport} defines an isometry between any two fibers,  as asserted by the following result.
\begin{theorem}\label{thm:bundle-metric}
	The metric defined in \eqref{eq:bundle-metric} is a vector bundle metric that smoothly varies with $(p,q)\in\manifold\times\manifold$ and is preserved by the parallel transport in \eqref{eq:parallel-transport}.
\end{theorem}

The inner product \eqref{eq:bundle-metric} is defined for linear operators that map a Hilbert space, such as $\tangentspace p$, to another potentially different Hilbert space, such as $\tangentspace q$. The inner product of this type, although mathematically well established \citep[e.g., Definition 2.3.3 and Proposition B.0.7 by][]{Prevot2007}, is less seen in statistics; the commonly used one is usually linear operators that map a Hilbert space into the same Hilbert space. %In particular, it coincides with the usual Hilbert--Schmidt inner product  when $p=q$. 
The metric $G$ also induces a norm, denoted by $\|\cdot\|_{G(p,q)}$ or simply $\|\cdot\|_{G}$,  on each fiber $\mathbb L(p,q)$. This norm in turn defines a distance on each fiber $\mathbb L(p,q)$ by $\|A-B\|_{G(p,q)}$ for $A,B\in\mathbb L(p,q)$, which is an integrated part of the loss function in \eqref{eq:cov-estimate} for estimating the covariance function.

The smooth vector bundle $\mathbb L$ together with the covariant derivative \eqref{eq:covariant-deriv} and the bundle metric \eqref{eq:bundle-metric}, termed \emph{covariance vector bundle} in this paper, paves the way for estimation of the covariance function \eqref{eq:cov-func} from sparsely observed Riemannian functional data. The smooth structure and the covariant derivative together provide an intrinsic mechanism to quantify the regularity of $\covarop$. For example, it makes meaningful the statement that the second derivatives of $\covarop(s,t)$ are continuous. In the Euclidean case, statements of this kind are often adopted as assumptions that are fundamental to theoretical analysis of estimators derived from a smoothing method. The developed vector bundle and covariant derivative now enable us to extend such assumptions to the manifold setting, as demonstrated in Section \ref{sec:theory:cov} where we analyze the theoretical properties of the proposed estimator in Section \ref{sec:estimation} for the covariance function  $\covarop$. The parallel transport \eqref{eq:parallel-transport} and the bundle metric \eqref{eq:bundle-metric} allow an intrinsic measure of the discrepancy of objects in the covariance vector bundle. Such measure is critical for finding an estimator for $\covarop$ and quantifying estimation quality, as illustrated in the following sections.

\section{Estimation}\label{sec:estimation}
%\subsection{Estimation of mean}
The first step is to estimate the mean function, for which we adopt the local linear regression method proposed by \cite{Petersen2019} and also employed by \cite{Dai2019}. Define the local weight function
$$
\hat{w}(T_{ij},t,h_{\mu})=\frac{1}{\hat{\sigma}_{0}^{2}(t)}K_{h_{\mu}}(T_{ij}-t)\{\hat{u}_{2}(t)-\hat{u}_{1}(t)(T_{ij}-t)\},
$$
where $\hat{u}_{k}(t)=\sum_{i}\lambda_{i}\sum_{j}K_{h_{\mu}}(T_{ij}-t)(T_{ij}-t)^{k}$,  $\hat{\sigma}_{0}^{2}(t)=\hat{u}_{0}(t)\hat{u}_{2}(t)-\hat{u}_{1}^{2}(t)$ and $K_{h_{\mu}}(\cdot)=K(\cdot/h_\mu)/h_\mu$ for a kernel function $K$ with bandwidth $h_{\mu}>0$. The estimate $\hat{\mu}$ is defined as the minimizer of the weighted function 
$$
\hat{Q}_{n}(y,t)=\sum_{1\leq i\leq n}\lambda_{i}\sum_{1\leq j\leq m_{i}}\hat{w}(T_{ij},t,h)d_{\manifold}^{2}(Y_{ij},y),
$$
i.e., 
$$
\hat{\mu}(t)=\underset{y\in\manifold}{\arg\min}~\hat{Q}_{n}(y,t),
$$
where the weights $\{\lambda_{i}\}_{1\leq i\leq n}$ are subject-specific and satisfy $\sum_{i=1}^{n}\lambda_{i}m_{i}=1$.
For the Euclidean case $\manifold=\real$, the objective function $\hat{Q}_{n}(y,t)$ coincides with the sum of squared error loss used in~\cite{Zhang2016}. Two popular choices for $\lambda_i$ are $\lambda_{i}=({\sum_{i=1}^{n}m_{i}})^{-1}$ \citep{Yao2005a} that assigns equal weight to each observation, and   $\lambda_{i}=({nm_{i}})^{-1}$ \citep{Li2010} that  assigns equal weight to each subject. Other choices are discussed in \cite{Zhang2018}.
%As for the Euclidean case the former scheme was found to work better for non-dense and the latter for ultra-dense functional data, we will adopt the former weight scheme in our implementations for sparsely sampled data.

Given the parallel transport introduced in Section \ref{subsec:bundle}, we are allowed to move {the raw covariance}   $\hat\covarop_{i,jk}$ defined in \eqref{eq:raw-cov} from different fibers into the same fiber and employ the classic local linear smoothing on the transported observations. %Recall that the raw covariance is defined by 
%$$
%\hat{\covarop}_{i,jk}:=\Log_{\hat{\mu}(T_{ij})}Y_{ij}\otimes \Log_{\hat{\mu}(T_{ik})}Y_{ik} \in \mathbb{L}(\hat{\mu}(T_{ij}),\hat{\mu}(T_{ik})).
%$$
For $\hat\covarop_{i,jk}$  to be well defined, similar to Assumption \ref{ass:muexist}, we assume the existence and uniqueness of the empirical mean function $\hat\mu$ in the following Assumption~\ref{ass:muhatexist}. Such assumption always holds for manifolds of nonpositive sectional curvature, or may be replaced by a convexity condition on the distance function when Assumption \ref{ass:Xnearmu} holds.
In particular, Assumption~\ref{ass:muhatexist} is satisfied by the manifolds $\mathrm{Sym}_{AF}^+$ and $\mathrm{Sym}_{LC}^+$ of SPD matrices adopted in the simulation studies in Section~\ref{sec:simulation} and data application in Section \ref{sec:application} without additional conditions, as these manifolds have nonpositive sectional curvature. The assumption also holds for the sphere $\mathbb S^2$ used in the simulation studies when Assumption \ref{ass:Xnearmu} is fulfilled.

\begin{assumption}
	\label{ass:muhatexist}  
	The estimated mean function  $\hat\mu(t)$  exists and is unique for each $t\in\tdomain$.
\end{assumption}

%which is guaranteed to be well defined for all $i,j,k$ with probability tending to one by Lemma~\ref{lem:loghatmu}. 
To estimate $\covarop(s,t)$, the nearby raw observations $\hat{\covarop}_{i,jk}$ are parallelly transported into the fiber $\mathbb{L}(\hat{\mu}(s),\hat{\mu}(t))$, %When the bandwidth $h_{\covarop}$ is sufficiently small, the geodesic exists and the parallel transport is well defined. 
and the estimate $\hat{\covarop}(s,t)$ is set by $\hat{\covarop}(s,t)=\hat\beta_0$ with
\begin{equation}\label{eq:cov-estimate}
\begin{aligned}
(\hat\beta_0,\hat\beta_1,\hat\beta_2)=&\underset{\beta_{0},\beta_{1},\beta_{2}\in \mathbb{L}(\hat{\mu}(s),\hat{\mu}(t))}{\arg\min} \bigg\{ \sum_{i}\nu_{i}\sum_{j\neq k}
\|\vpt_{(\hat{\mu}(T_{ij}),\hat{\mu}(T_{ik}))}^{(\hat{\mu}(s),\hat{\mu}(t))}\hat{\covarop}_{i,jk}-\beta_{0}-\beta_{1}(T_{ij}-s) \\
&\qquad\quad -\beta_{2}(T_{ik}-t)\|_{G(\hat{\mu}(s),\hat{\mu}(t))}^{2}
K_{h_{\covarop}}(s-T_{ij})K_{h_{\covarop}}(t-T_{ik})\bigg\},
\end{aligned}
\end{equation}
where $h_{\covarop}>0$ is a bandwidth, $\vpt_{(\hat{\mu}(T_{ij}),\hat\mu(T_{ik}))}^{(\hat\mu(s),\hat\mu(t))}$ is the parallel transport along minimizing geodesics defined in \eqref{eq:parallel-transport}, and the weights $\{\nu_{i}\}_{1\leq i\leq n}$ are subject-specific and satisfy $\sum_{i=1}^{n}\nu_{i}m_{i}(m_{i}-1)=1$.
Similar to the estimation of the mean function, two popular choices for the weights are $\nu_{i}=({\sum_{i=1}^{n}m_{i}(m_{i}-1)})^{-1}$ \citep{Yao2005a} that assign equal weight to each observation, and $\nu_{i}=({{n}m_{i}(m_{i}-1)})^{-1}$ \citep{Li2010} that assign equal weight to each subject, while more options are studied in \cite{Zhang2018}.

The objective function in \eqref{eq:cov-estimate} involves only intrinsic concepts and thus is fundamentally different from the objective function in (5) of \cite{Dai2019} in which the raw observations $\hat\covarop_{i,jk}$ are computed in an ambient space. In addition, the quantities $\hat\covarop_{i,jk}$ in \eqref{eq:cov-estimate} are frame-independent and thus the resulting estimator is invariant to the frame\footnote{For the computational purpose, a frame might be adopted, but the resulting estimator is independent of the choice of the frame, since the objective function in \eqref{eq:cov-estimate} does not depend on any frame.}. This frame-independent feature makes  our estimator distinct from the non-invariant estimators discussed Section~\ref{sec:invariance} of the supplement.% in the following Remark \ref{rem:attempt:2} in which a frame is essential and the produced estimator is not invariant to the frame.

\begin{remark}One might attempt to endow $\mathbb{L}$ with a distance $\rho$ so that the estimation is turned into a regression problem with a metric-space-valued response and the local linear method of \cite{Petersen2019} can be adopted. Such distance is expected to have the following properties: 
	\begin{itemize}
		\item The  distance $\rho$ on $\mathbb{L}$ coincides with the fiber metric $G$ for any two points on the same fiber. Specifically, for $L_{1},L_{2}\in\mathbb{L}(p,q)$, $\rho^2(L_{1},L_{2})=G_{(p,q)}(L_{1}-L_{2},L_{1}-L_{2})$.
		\item The distance $\rho$ on the zero section $W_0(p,q)=0\in \mathbb L(p,q)$ coincides with the geodesic distance on $\manifold\times\manifold$. Specifically, for $(p_{1},q_{1}),(p_{2},q_{2})\in\manifold\times\manifold$, $\rho(W_0(p_1,q_1),W_0(p_2,q_2))=d_{\manifold^2}((p_{1},q_{1}),(p_{2},q_{2}))$.
		\item When $\manifold$ is a Euclidean space, especially when $\manifold=\real$, the estimate derived from \cite{Petersen2019} under the distance $\rho$ coincides with the classic estimate, i.e., the estimate derived from the same method but applied to the observations $\hat\covarop_{i,jk}\in\real$ that are treated as real-valued responses.
	\end{itemize}
	However, such distance $\rho$ does not exist. On one hand, the positive-definiteness of the distance suggests that $\rho(\hat\covarop_{i_1,j_1k_1},\hat\covarop_{i_2,j_2k_2})\neq 0$ as long as  $\hat\covarop_{i_1,j_1k_1},\hat\covarop_{i_2,j_2k_2}\in \mathbb{L}$ reside in different fibers, i.e., when $\hat\mu(T_{i_1j_1})\neq \hat\mu(T_{i_2j_2})$ or $\hat\mu(T_{i_1k_1})\neq \hat\mu(T_{i_2k_2})$. On the other hand, when $\manifold=\real$, the quantities %although $\mathbb L$ is identified with $\manifold^2\times\real=\real^3$, 
	$\hat\covarop_{i_1,j_1k_1}$ and $\hat\covarop_{i_2,j_2k_2}$ are treated as real numbers and thus their distance could be zero even when  $\hat\mu(T_{i_1j_1})\neq \hat\mu(T_{i_2j_2})$ or $\hat\mu(T_{i_1k_1})\neq \hat\mu(T_{i_2k_2})$. 
\end{remark}	

Once an estimate $\hat\covarop$ of the covariance function $\covarop$ is obtained, according to Theorem \ref{pro:C=Cst}, the intrinsic Riemannian functional principal component proposed in \cite{Lin2019} can be adopted. Specifically, the eigenvalues $\hat\lambda_k$ and eigenfunctions $\hat \psi_k$ of $\hat\covarop$ can be obtained by eigen-decomposition of $\hat\covarop$, e.g., via the method described in Section 2.3 of \cite{Lin2019}. For estimation of the scores $\xi_{ik}=\llangle \Log_{\mu}X_i,\psi_k\rrangle$ in the intrinsic Karhunen--Lo\'eve expansion $\Log_\mu X_i=\sum_{k=1}^\infty \xi_{ik}\psi_k$ proposed in \cite{Lin2019}, numerical approximation  to the integral $\llangle \Log_\mu X_i,\psi_j\rrangle$ is infeasible when the data are sparse. In the Euclidean setting, this issue is addressed by the technique of principal analysis through conditional expectation \citep[PACE,][]{Yao2005a}. The technique was also adopted by  \cite{Dai2019} for their ambient approach to Riemannian functional data analysis on sparsely observed data. To adapt this technique in our intrinsic framework, for each $\tangentspace{\hat\mu(T_{ij})}$, we fix an orthonormal basis $B_{ij,1},\ldots,B_{ij,d}$; in Proposition~\ref{prop:score} we will show that the computed scores do not depend on the choice of the basis.  Then, the observations $\Log_{\hat\mu(T_{ij})}Y_{ij}$ and the estimated eigenfunctions $\hat\psi_k(T_{ij})$ can be represented by their respective coordinate vectors $z_{ij}$ and $g_{k,ij}$  with respect to the basis. Similarly, the estimated covariance function $\hat\covarop(T_{ij},T_{il})$ at $(T_{ij},T_{il})$ can be represented by a matrix $C_{i,jl}$ of coefficients. By treating the vectors $z_{ij}$ as $\real^d$-valued observations, the best linear unbiased predictor (BLUP) of $\xi_{ik}$ is given by 
\begin{equation}\label{eq:fpc-score}
\hat\xi_{ik}=\hat\lambda_k g_{k,i}^\top \Sigma_{i}^{-1}z_i,
\end{equation}
where $g_{k,i}=(g_{k,i1}^\top,\ldots,g_{k,im_i}^\top)^\top$, $z_i=(z_{i1}^\top,\ldots,z_{im_i}^\top)^\top$ and 
$$\Sigma_i=\hat\sigma^2 \mathbf I + \begin{pmatrix}C_{i,11} & C_{i,12} & \cdots & C_{i,1m_{i}}\\
C_{i,21} & C_{i,22} & \cdots & C_{i,2m_{i}}\\
\vdots & \vdots & \ddots & \vdots\\
C_{i,m_{i}1} & C_{i,m_{i}2} & \cdots & C_{i,m_{i}m_{i}}
\end{pmatrix}$$
with $\hat\sigma^2=\sum_{i=1}^n\sum_{j=1}^{m_i}(ndm_i)^{-1}\mathrm{tr}\{z_{ij}z_{ij}^\top-\hat\covarop(T_{ij},T_{ij})\}$. The following invariance principle shows that the scores $\hat\xi_{ik}$ in \eqref{eq:fpc-score} are invariant to the choice of bases $B_{ij,1},\ldots,B_{ij,d}$. This extends the invariance principle of \cite{Lin2019} from the fully observed and/or dense design to the sparse case.
\begin{proposition}\label{prop:score}
	The principal component scores $\hat\xi_{ik}$ in \eqref{eq:fpc-score} do not depend on the choice of the orthonormal bases $\{(B_{ij,1},\ldots,B_{ij,d}):i=1,\ldots,n,j=1,\ldots,m_i\}$.
\end{proposition}

It remains to choose the bandwidths $h_\mu$ and $h_{\covarop}$. Although the theoretical analysis in the next section sheds light on how to choose them when the sample size is large, to determine appropriate values for them when the sample is limited, we propose the following $k$-fold cross-validation procedure. For an integer $k\geq 2$, divide the subjects into $k$ partitions, denoted by $P_1,\ldots,P_k\subset\{1,\ldots,n\}$, of roughly even size. Let $P_{-l}=\{1,\ldots,n\}\backslash P_{l}$ for $l=1,\ldots,k$. For a candidate value $h$ of $h_\mu$, its cross-validation error is computed by 
\[\textsc{cv}(h)=\sum_{l=1}^k\sum_{i\in P_j}\sum_{j=1}^{m_i}d^2_{\manifold}\big(\hat\mu_{-l}^h(T_{ij}),Y_{ij}\big),\]
where $\hat\mu_{-l}^h$ is the estimated mean function by using the bandwidth $h$ and the data $P_{-l}$. Among a set of candidate values of $h_\mu$, the one with the minimal cross-validation error is selected. A value for $h_{\covarop}$ can be selected by the similar procedure. As demonstrated in Section \ref{sec:simulation}, this $k$-fold cross-validation procedure is numerically effective.

\section{Asymptotic properties}\label{sec:theory}
In the sequel we assume $m_1=\cdots=m_n=m$ for a clear exposition; extension to more general cases is technically straightforward  \citep{Zhang2016}.
There are two popular types of designs, namely, the random design in which the design points $T_{ij}$ are i.i.d. sampled from a distribution, and the deterministic design in which $T_{ij}$ are predetermined and thus nonrandom. For the random design, the following assumption is commonly adopted \citep{Yao2005a,Li2010,Zhang2016}.
\begin{assumption}[Random Design]
	\label{ass:random:design} The design points $T_{ij}$, independent of other random quantities, are i.i.d. sampled from a distribution on $\tdomain$ with a probability density that is bounded away from zero and infinity.
\end{assumption}
In contrast, the deterministic design is less studied, especially the irregular deterministic design; for instance, \cite{Cai2011} considers only a regular deterministic design. In this paper, we consider a deterministic design with the following condition that basically states that the design points are sufficiently irregular. To focus on longitudinal observations, the regular design of a common grid case is studied in Section \ref{sec:regulardesign} of the supplement, and is 
not included in the following condition.
\begin{assumption}[Deterministic Design]\label{ass:deterministic-design} The design points $T_{ij}$ are nonrandom, and there exist constants $c_2\geq c_1>0$, such that for any interval $A,B\subset \tdomain$ and all $n\geq 1$,
	\begin{enumerate}[label=\textup{(\alph*)}]
		\item $\sup_{1\leq i\leq n}\sum_{j=1}^m 1_{T_{ij}\in A}\leq \max\{c_2m|A|,1\}$,
		\item $c_1nm|A|-1\leq \sum_{i,j}1_{T_{ij}\in A}\leq \max\{c_2nm|A|,1\}$, and 
		\item $c_1nm^2|A||B|-1\leq \sum_{i,j,k}1_{T_{ij}\in A}1_{T_{ik}\in B}\leq \max\{c_2nm^2|A||B|,1\}$,
	\end{enumerate}
	where $|A|$ denotes the length of $A$.
\end{assumption}

In many applications, the design points are neither completely random nor completely predetermined. For example, in longitudinal studies the visit of a patient may take place at a time that randomly deviates from the scheduled time. Such design includes both a deterministic part and a random component, which is termed \emph{hybrid design} in this paper.
{Specifically, suppose all measurements are scheduled to take place in some of the $L$ predetermined points of $\tdomain$. Without loss of generality, we assume these predetermined points are $\mathcal A_L:=\{s_{k}:1\leq k\leq L\}$ with $s_k=k/(L+1)$ and the set of all $m$-element subsets of $\mathcal A_L$ is $\mathcal S_m=\{\{t_1,\ldots,t_m\}: t_1,\ldots,t_m\in \mathcal A_L\text{ are distinct}\}$. 
There are $m\leq L$ measurements scheduled at distinct time points $\mathbb{S}_{i}:=\{S_{i1},\ldots,S_{im}\}\in \mathcal S_m$ for each curve $i$.}
Instead of $S_{ij}$, the actual measurement takes place at $T_{ij}=S_{ij}+\zeta_{ij}$ for some random variable $\zeta_{ij}\in(-1/(2L+2),1/(2L+2))$. We postulate the following condition in which we 
emphasize that $S_{i1},\ldots,S_{im}$ are \emph{not} independent and thus neither are $T_{i1},\ldots,T_{im}$. In addition, note that the condition also includes the special case that $S_{i1},\ldots,S_{im}$ are deterministic when $m=L$.
\begin{assumption}[Hybrid Design]
	\label{ass:hybrid}
~
	\begin{enumerate}[label=\textup{(\alph*)}]
		\item\label{ass:hybrid:design} For each $i=1,\ldots,n$, $S_{i1},\ldots,S_{im}$ are $m$ distinct elements randomly sampled (without replacement) from $\mathcal A_L$. In addition, $\mathbb S_1,\ldots,\mathbb S_n$ are i.i.d. random subsets of $\mathcal A_L$ and there are positive constants $c_1,c_2$ such that  $c_1|\mathcal S_m|^{-1}\leq \prob(\mathbb S_1\in\mathbf s)\leq c_2|\mathcal S_m|^{-1}$ for all $\mathbf s\in\mathcal S_m$.
		\item\label{ass:hybrid:design:random:error} $\zeta_{ij}$ are i.i.d. centered random variables taking values in $(-1/(2L+2),1/(2L+2))$, and are independent of other random quantities. In addition, there exist universal positive constants $c_3$ and $c_4$ such that the probability density $f_\zeta$ of $\zeta_{11}$ satisfies $c_3 L\leq \inf_s f_\zeta(s)\leq \sup_s f_\zeta(s)\leq c_4 L$.
	\end{enumerate}
\end{assumption}

Although these designs differ in nature, in the next two sections, we show that the estimator with either of these designs, respectively for the mean function and for covariance function, achieves the same convergence rate under suitable regularity conditions.

\subsection{Mean function}
The pointwise convergence rate of the estimate $\hat{\mu}(t)$ is established in \cite{Petersen2019}, while the uniform convergence rate is derived by  \cite{Dai2019}. For completeness, we include them here, and establish a new local uniform result that is needed in the theoretical analysis of the covariance estimator. First, we require the following assumptions, where the condition \ref{ass:basic:data-range} may be replaced with tail and moment conditions on the distributions of $Y$ and $X$ at the cost of heavier technicalities. In addition,  by modifying our proofs, the compactness in the condition  \ref{ass:basic:kernel} can be replaced with a condition on the decay rate of the kernel function when it moves away from zero, so that noncompact kernels such as Gaussian kernel can be accommodated.

\begin{assumption}
	\label{ass:basic}
	\quad
	\begin{enumerate}[label=\textup{(\alph*)}]
		\item\label{ass:basic:M} The Riemannian manifold $\manifold$ is complete and simply connected\footnote{See Section \ref{sec:rm} in the supplementary material for a precise definition of simple connectedness.}. 
		\item\label{ass:basic:data-range} There exists a compact 
		subset of $\mathcal K\subset\manifold$ such that $\prob\{X(t),Y(t)\in\mathcal{K}$ for all $t\in\tdomain\}=1$. 
		\item\label{ass:basic:tdomain} The domain $\tdomain$ is a compact interval.
		\item\label{ass:basic:kernel} The kernel function $K$ is Lipschitz continuous, symmetric, positive on $(-1,1)$, compactly supported on $[-1,1]$, and monotonically decreasing on $[0,1]$.
	\end{enumerate}
\end{assumption}

The following  regularity on the mean function or related quantities is adapted from \cite{Petersen2019} and is specialized to the Riemannian manifold. Part \ref{ass:mu:b:0} states that the Fr\'echet mean is well separated from the other points in terms of the Fr\'echet function $F^\ast(y,t):=\expect d_{\manifold}^2(X(t),y)$, while part \ref{ass:mu:c}  basically amounts to convexity of $F^\ast(\cdot,t)$ around $\mu(t)$; they hold, for example, when the manifold has nonpositive curvature or the data sufficiently concentrate on a geodesically convex region.

\begin{assumption}\label{ass:mu}~
	\begin{enumerate}[label=\textup{(\alph*)}]
		\item\label{ass:mu:a} The second partial derivative $\partial ^2_tF^\ast(y,t)$ is bounded on $\mathcal K\times \tdomain$.
		\item\label{ass:mu:b:0} For any $\delta>0$,
		\begin{align*}
		&\inf_{\stackrel{d_{\manifold}(y,{\mu}(t))>\delta}{t\in \tdomain}}\{F^{\ast}(y,t)-F^{\ast}(\mu(t),t)\}>0,
		\end{align*}
		\item\label{ass:mu:c} There exist $\eta_{1}>0$ and $C_{1}>0$ such that for all $t\in\tdomain$ and  all $y$ with $d_{\manifold}(y,\mu(t))<\eta_{1}$, 
		\begin{align*}
		& F^{\ast}(y,t)-F^{\ast}(\mu(t),t) - C_{1}d_{\manifold}(y,\mu(t))^{2}\geq 0.
		\end{align*}
		
	\end{enumerate}
\end{assumption}

The following proposition, whose proof, as well as proofs for other results in this section, is deferred to the Supplementary Material, states the point-wise and uniform convergence rates of the estimated mean function, where the point-wise rate is an immediate consequence of the local uniform rate stated in Proposition \ref{prop:random:mu:rate:local}. The condition  $nh_\mu\asympgt 1$ in the following is only needed for the deterministic design. 

\begin{proposition}
	\label{prop:random:mu:sup}
	Suppose that Assumptions~\ref{ass:muexist},~\ref{ass:Xnearmu},~\ref{ass:muhatexist},~\ref{ass:basic} and~\ref{ass:mu}. Under either of Assumptions \ref{ass:random:design},  \ref{ass:deterministic-design} and \ref{ass:hybrid}, if $h_{\mu}\rightarrow 0$ and $nmh_{\mu}\rightarrow \infty$, then for any fixed $t\in\tdomain$,
	\[
	d_{\manifold}^2(\mu(t),\hat{\mu}(t))=\Op\left(h_{\mu}^{4}+\frac{1}{n}+\frac{1}{nmh_{\mu}}\right),
	%=\Op(h_{\mu}^{2}+n^{-\frac{1}{2}}m^{-\frac{1}{2}}h_{\mu}^{-\frac{1}{2}}+n^{-\frac{1}{2}})
	\]
	and if $h_{\mu}\rightarrow 0$, $nh_\mu\asympgt 1$ %indeed, $n^ah_\mu\asympgt 1$ for any fixed (large) a
	and $nmh_{\mu}/\log n\rightarrow \infty$, then
	\[
	\sup_{t\in\tdomain}d_{\manifold}^2(\mu(t),\hat{\mu}(t))=\Op\left(h_{\mu}^{4}+\frac{\log n}{n}+\frac{\log n}{nmh_{\mu}}\right).
	\]
\end{proposition}

To derive the point-wise convergence rate of the estimator $\hat{\covarop}(s,t)$ in the next subsection, we require a  local convergence property of the estimator $\hat{\mu}$. The following Proposition \ref{prop:random:mu:rate:local}, which is new in the literature, shows that the local uniform convergence rate is the same as the point-wise rate in Proposition~\ref{prop:random:mu:sup}, and differs from the global uniform convergence rate that has an additional $\log n$ factor. The reason for this phenomenon is that $\expect \{K_{h_{\mu}}(T-t)\}=1$ at a fixed point $t$ but $\expect\{\sup_{t\in\tdomain} K_{h_{\mu}}(T-t)\}={1}/{h_{\mu}}\rightarrow \infty$. Therefore, the additional  $\log n$ factor is needed to offset this explosion in the case of global uniform convergence. In the local case, if $h=O(h_{\mu})$ and thus $\expect\{\sup_{\tau:|\tau-t|\leq h} K_{h_{\mu}}(T-\tau)\}=O({h}/{h_{\mu}})=O(1)$,  then no offset is required. The proposition also directly implies the point-wise rate in Proposition~\ref{prop:random:mu:sup}. 

\begin{proposition}
	\label{prop:random:mu:rate:local}
	Suppose that Assumptions~\ref{ass:muexist},~\ref{ass:Xnearmu},~\ref{ass:muhatexist},~\ref{ass:basic} and \ref{ass:mu} hold. Under  either of Assumptions \ref{ass:random:design}, \ref{ass:deterministic-design} and \ref{ass:hybrid}, if $h_{\mu}\rightarrow 0$ and  $nmh_{\mu}\rightarrow \infty$, then for any fixed $t$ and $h=O(h_{\mu})$,
	\[
	\sup_{\tau:|\tau-t|\leq h}d_{\manifold}^2(\mu(\tau),\hat{\mu}(\tau))=\Op\left(h_{\mu}^{4}+\frac{1}{n}+\frac{1}{nmh_{\mu}}\right).
	%=\Op(h_{\mu}^{2}+n^{-\frac{1}{2}}m^{-\frac{1}{2}}h_{\mu}^{-\frac{1}{2}}+n^{-\frac{1}{2}})
	\]
\end{proposition}

\subsection{Covariance function}\label{sec:theory:cov}
We start with the following assumption on the regularity of the covariance function $\covarop$.  As discussed in Section \ref{subsec:bundle}, such regularity condition in the manifold setting is made precise and meaningful by the constructed covariance vector bundle $\mathbb L$ and the covariant derivative $\nabla$ in \eqref{eq:covariant-deriv}.
\begin{assumption}
	\label{ass:covariance}
	The covariance function $\covarop$  is twice differentiable and its second derivatives are continuous.
\end{assumption}

To study the asymptotic properties of the estimator $\hat\covarop$, one of the major challenges that are not encountered in the Euclidean setting of \cite{Zhang2016} or the ambient case of \cite{Dai2019} is to deal with the parallel transport in \eqref{eq:cov-estimate}. %, \eqref{eq:cov-pointwise} and \eqref{eq:cov-uniform}. 
It turns out that we need to quantify the discrepancy between a tangent vector and the parallelly transported one along a geodesic quadrilateral. We address this issue by the following lemma which may be of independent interest. In particular, the proof of the lemma given in the Supplementary Material utilizes holonomy theory that appears new in statistical literature.
\begin{lemma}\label{lem:c_tildec}
	For a compact subset $\mathcal G\subset\manifold$, there exists a constant $c>0$ depending only on $\mathcal G$, such that for all $p_1,p_2,q_1,q_2,y\in\mathcal G$, \[
	\|\pt_{q_1}^{p_1}\pt_{q_2}^{q_1}\Log_{q_2}y-\pt_{p_2}^{p_1}\Log_{p_2}y\|_{p_1}\leq c(d_{\manifold}(p_1,q_1)+d_{\manifold}(p_2,q_2)).
	\]
\end{lemma}

With the above regularity condition and lemma, the following theorem establishes the point-wise convergence rate of $\hat\covarop$.
\begin{theorem}
	\label{the:covariance}
	Suppose that Assumptions~\ref{ass:muexist},~\ref{ass:Xnearmu},~\ref{ass:muhatexist},~\ref{ass:basic},~\ref{ass:mu} and~\ref{ass:covariance} hold. Under either of Assumptions~\ref{ass:random:design}, \ref{ass:deterministic-design} and \ref{ass:hybrid}, if  $h_{\mu}\rightarrow 0$, $h_{\covarop}=O(h_{\mu})$, and $\min\{nmh_{\mu},nm^{2}h_{\covarop}^{2}\}\rightarrow \infty$, then for any fixed $s,t\in\tdomain$,
	\begin{equation}\label{eq:cov-pointwise}
	\left\|\vpt_{(\hat{\mu}(s),\hat{\mu}(t))}^{(\mu(s),\mu(t))}\hat{\covarop}(s,t)-\covarop(s,t)\right\|_{G(\mu(s),\mu(t))}^2=\Op\left(h_{\mu}^{4}+h_{\covarop}^{4}+\frac{1}{n}+\frac{1}{nmh_{\mu}}+\frac{1}{nm^{2}h_{\covarop}^{2}}\right).
	\end{equation}
\end{theorem}
The rate in the above theorem matches the point-wise rate in the Euclidean setting of \cite{Zhang2016} in the case of $m_i=m$. Unlike  \cite{Zhang2016} which assumes that the mean function is known in their analysis, we do not need such assumption thanks to the local uniform rate of the mean function stated in Proposition \ref{prop:random:mu:rate:local}. In our analysis, the local uniform rate can not be replaced with the global uniform rate in Proposition \ref{prop:random:mu:sup} without introducing an additional $\log n$ factor. 
Although the condition $h_{\covarop}=O(h_{\mu})$ is required in order to utilize  Proposition \ref{prop:random:mu:rate:local}, it does not limit the convergence rate, as a proper choice of $h_\mu$ and $h_{\covarop}$ leads to the following rates that still match the rates of \cite{Zhang2016}.
\begin{corollary}
	\label{cor:covariance}
	Assume the conditions of Theorem \ref{the:covariance}.
	\begin{enumerate}[label=\textup{(\alph*)}]
		\item When $m\asympeq n^{1/4}$ or $m\gg n^{1/4}$, with $h_{\mu}\asymp h_{\covarop} \asymp n^{-1/4}$, one has 
		$$
		\left\|\vpt_{(\hat{\mu}(s),\hat{\mu}(t))}^{(\mu(s),\mu(t))}\hat{\covarop}(s,t)-\covarop(s,t)\right\|^2_{G_{(\mu(s),\mu(t))}}=\Op\left(\frac 1 n\right).
		$$
		\item When $m\ll n^{1/4}$, with $h_{\mu}\asymp h_{\covarop}\asymp n^{-1/6}m^{-1/3}$, one has
		$$
		\left\|\vpt_{(\hat{\mu}(s),\hat{\mu}(t))}^{(\mu(s),\mu(t))}\hat{\covarop}(s,t)-\covarop(s,t)\right\|^2_{G_{(\mu(s),\mu(t))}}=\Op\left(\frac{1}{n^{2/3}m^{4/3}}\right).
		$$ 
	\end{enumerate}
\end{corollary}
Like the Euclidean case, a phase transition is observed at $m\asymp n^{1/4}$. With a proper choice of $h_\mu$ and $h_{\covarop}$, if $m$ grows at least as fast as $n^{1/4}$, it does not impact the convergence rate that is at a parametric order of magnitude, i.e., $n^{-1/2}$. Otherwise, the sampling rate $m$ becomes an integrated part of the convergence rate of $\hat\covarop$. In particular, when $m\ll n^{1/4}$, the choice  $h_\mu\asympeq n^{-1/6}m^{-1/3}$ is required to respect the condition $h_{\covarop}=O(h_\mu)$. This choice is strictly larger than the optimal choice $h_\mu\asympeq (nm)^{-1/5}$ that is implied by  Proposition \ref{prop:random:mu:sup} in the case of $m\ll n^{1/4}$. This suggests that oversmoothing in the mean function estimation may be needed in order to reach the optimal point-wise rate of the covariance estimator when $m\ll n^{1/4}$. It is interesting to note that, in the literature if a lower-dimensional estimate depends on a higher-dimensional one, undersmoothing the latter often helps the former to attain a better rate. For instance, undersmoothing the two-dimensional covariance surface estimate leads to a better rate of the  one-dimensional eigenfunction estimates \citep{Hall2006a}. In contrast, the phenomenon in our case is reversed: the two-dimensional covariance function estimate depends on the one-dimensional mean estimate, thus requires oversmoothing the latter instead.

The following results establish the uniform convergence rate of the estimator $\hat\covarop$, where the condition $\min\{nh_\mu,nh_{\covarop}^2\}\asympgt 1$ is only needed for the deterministic design.
\begin{theorem}
	\label{the:covariance:uniform}
	Suppose that Assumptions~\ref{ass:muexist},~\ref{ass:Xnearmu},~\ref{ass:muhatexist},~\ref{ass:basic},~\ref{ass:mu} and~\ref{ass:covariance} hold. Under either of Assumptions~\ref{ass:random:design}, \ref{ass:deterministic-design} and \ref{ass:hybrid}, if $\max\{h_{\mu},h_{\covarop}\}\rightarrow 0$, $\min\{nh_\mu,nh_{\covarop}^2\}\asympgt 1$ and $\min\{nmh_{\mu},nm^{2}h_{\covarop}^{2}\}/\log n\rightarrow \infty$, then
	\begin{equation}\label{eq:cov-uniform}
	\sup_{(s,t)\in\tdomain^2}\left\|\vpt_{(\hat{\mu}(s),\hat{\mu}(t))}^{(\mu(s),\mu(t))}\hat{\covarop}(s,t)-\covarop(s,t)\right\|^2_{G_{(\mu(s),\mu(t))}}=\Op\left(
	h_{\mu}^{4}+h_{\covarop}^{4}+
	\frac{\log n}{n}+\frac{\log n}{nmh_{\mu}}+\frac{\log n}{nm^{2}h_{\covarop}^{2}}\right).
	\end{equation}
\end{theorem}
\begin{corollary}
	\label{cor:covariance_uniform}
	Assume the conditions of Theorem \ref{the:covariance:uniform}.
	\begin{enumerate}[label=\textup{(\alph*)}]
		\item  When $m\asympeq n^{1/4}$ or $m\gg n^{1/4}$, with $h_{\mu}\asymp h_{\covarop} \asymp n^{-1/4}$, one has 
		$$
		\sup_{(s,t)\in\tdomain^2}\left\|\vpt_{(\hat{\mu}(s),\hat{\mu}(t))}^{(\mu(s),\mu(t))}\hat{\covarop}(s,t)-\covarop(s,t)\right\|^2_{G_{(\mu(s),\mu(t))}}=\Op\left(\frac{\log n }{n}\right).
		$$
		\item When $m\ll n^{1/4}$, with $h_{\mu}\asymp {n}^{-1/5}m^{-1/5}(\log n)^{1/5}$ and $h_{\covarop}\asymp n^{-1/6}m^{-1/3}(\log n)^{1/6}$, one has  
		$$
		\sup_{(s,t)\in\tdomain^2}\left\|\vpt_{(\hat{\mu}(s),\hat{\mu}(t))}^{(\mu(s),\mu(t))}\hat{\covarop}(s,t)-\covarop(s,t)\right\|^2_{G_{(\mu(s),\mu(t))}}=\Op\left(\frac{(\log n)^{2/3}}{n^{2/3}m^{4/3}}\right).
		$$ 
	\end{enumerate}
\end{corollary}
These rates  again match the uniform rates in \cite{Zhang2016}. They also coincide with the rates\footnote{Note that the extra term $1/(nmh_{\covarop})$ in \cite{Dai2018} is dominated by $1/n+1/(nm^2h^2_{\covarop})$ due to the inequality of arithmetic and geometric means, i.e., $\sqrt{ab}\leq (a+b)/2$.}  in \cite{Dai2018}. It is interesting to see that, when $m\ll n^{1/4}$, the choice of $h_\mu$ in the corollary is the same as the optimal choice implied by Proposition \ref{prop:random:mu:sup}, which suggests that no oversmoothing is needed in order to reach the optimal uniform rate for the covariance estimator $\hat\covarop$. This is because, the local uniform result of Proposition \ref{prop:random:mu:rate:local} and thus the condition $h_{\covarop}=O(h_\mu)$ are not required, as the role of Proposition \ref{prop:random:mu:rate:local} in the analysis is now played by Proposition \ref{prop:random:mu:sup}.

\section{Simulation studies}\label{sec:simulation}
We consider three different manifolds for illustrating the numerical properties of the proposed covariance estimator \eqref{eq:cov-estimate} in Section \ref{sec:estimation}; the numerical performance of the mean estimator can be found in \cite{Dai2019}. Namely, they are the two-dimensional unit sphere $\mathbb S^{2}$, the manifold $\mathrm{Sym}_{LC}^+$ of symmetric positive-definite $2\times 2$ matrices with the Log-Cholesky metric \citep{Lin2019Riemannian}, and the manifold $\mathrm{Sym}_{AF}^+$ of symmetric positive-definite $2\times 2$ matrices with the affine-invariant metric \citep{Moakher2005}, representing manifolds of positive, zero and negative sectional curvature, respectively. Note that although  $\mathrm{Sym}_{LC}^+$ and $\mathrm{Sym}_{AF}^+$ share the same collection of matrices, they are endowed with different Riemannian metric tensors and thus have fundamentally different Riemannian geometry. 
We set $\tdomain=[0,1]$. The sampling rate $m_i$ is randomly sampled from $\text{Poisson}(m)+2$, where $\text{Poisson}(m)$ is a Poisson distribution with parameter $m$. Conditional on $m_i$, the time points $T_{i1},\ldots,T_{im_i}$ are i.i.d. sampled from the uniform distribution  $\text{Uniform}(0,1)$. The random process $X$ and its mean and covariance functions are described below.

\paragraph*{\textbf{Sphere $\mathbb S^2$}}  We parameterize $\mathbb 
S^2=\{(x,y,z)\in\real^3:x^2+y^2+z^2=1\}$ by the polar coordinate system 
\begin{equation}
\label{equ:simu:polar}
x(u,v) = \cos(u)\sin(v),\quad 
y(u,v) = \cos(u)\cos(v), \quad 
z(u,v) = \sin(u)
\end{equation}
for the latitude $u\in (-\pi/2,\pi/2)$ and longitude $v\in [0,2\pi)$.
This coordinate system also gives rise to a local chart $\phi:U\rightarrow (-\pi/2,\pi/2)\times [0,2\pi)$ on $V=\mathbb S^2\backslash\{(0,0,-1),(0,0,1)\}$. Let $B_1(t)=\frac{\partial\phi}{\partial u}$ and $B_2(t)=\frac{\partial\phi}{\partial v}$. The random process $X$ is then given by 
$$
X(t)=\Exp_{\mu(t)}\big(tZ_{1}B_{1}(t)+tZ_{2}B_{2}(t)\big)
$$
with $Z_{1},Z_{2}\stackrel{i.i.d.}{\sim}  \mathrm{Uniform}(-0.1,0.1)$.  The mean curve $\mu$ of $X$ is $\mu(t)=\phi(0,\pi t/2)=(\sin(\pi t/2),\cos(\pi t/2),0)$, which lies on the equator. The covariance function is $\covarop(s,t)=\frac{st}{300} \mathbf I_2$ under the frame $(B_1,B_2)$, where $\mathbf I_2$ denotes the $2\times 2$ identity matrix. The contaminated observations are 
\[
Y_{ij}=\Exp_{\mu(T_{ij})}\{(T_{ij}Z_{1i}+\upsilon_{ij1})B_1(T_{ij})+(T_{ij}Z_{2i}+\upsilon_{ij2})B_2(T_{ij})\},
\]
where $Z_{1i},Z_{2i}\stackrel{i.i.d.}{\sim}  \mathrm{Uniform}(-0.1,0.1)$ and $\upsilon_{ij1},\upsilon_{ij2}\stackrel{i.i.d.}{\sim} \text{Uniform}(-a,a)$ with $a>0$ chosen to make $\text{SNR}=5$ defined by
\begin{equation}
\label{equ:snr}
\text{SNR}:=\frac{\expect\int_{\tdomain}\|\Log_{\mu(t)}X(t)\|^{2}_{\mu(t)}\diffop t}{\expect\int_{\tdomain}\|\varepsilon(t)\|^{2}_{\mu(t)}\diffop t}.
\end{equation}

%R_{122}^{1}=-sin^{2}(\iota),R_{121}^{2}=1

\paragraph*{\textbf{Manifold $\mathrm{Sym}_{LC}^+$}}  
We parameterize $\mathrm{Sym}_{LC}^+$ by the chart 
$$
\phi: (u,v,w)\rightarrow 
\left(
\begin{aligned}
&e^{2u} & we^{u} \\
&we^{u} & w^{2}+e^{2v} 
\end{aligned}
\right)
$$
which induces the orthogonal frame formed by  $B_1(t)=\frac{\partial\phi}{\partial u}$,
$B_2(t)=\frac{\partial\phi}{\partial v}$ and $B_3(t)=\frac{\partial\phi}{\partial w}$.
The random process $X$ is set to be 
$$
\begin{aligned}
X(t)=&\Exp_{\mu(t)}\big(tZ_{1}B_{1}(t)+tZ_{2}B_{2}(t)+tZ_{3}B_{3}(t)\big)=\left(
\begin{aligned}
&e^{t+tZ_{1}} & 0 \\
&t+tZ_{3} & e^{t+tZ_{2}} 
\end{aligned}
\right)
\left(
\begin{aligned}
&e^{t+tZ_{1}} & t+tZ_{3} \\
&0 & e^{t+tZ_{2}} 
\end{aligned}
\right)
\end{aligned}
$$
with $Z_{1},Z_{2},Z_{3}\stackrel{i.i.d.}{\sim} \text{Uniform}(-0.1,0.1)$.
The mean curve $\mu$ is a geodesic with 
$$
\mu(t)=\phi(t,t,t)=
\left(
\begin{aligned}
&e^{2t} & te^{t} \\
&te^{t} & t^{2}+e^{2t} 
\end{aligned}
\right)
$$
and the covariance function is $\covarop(s,t)=\frac{st}{300}\mathbf{I}_{3}$ under the frame $(B_{1},B_{2},B_{3})$. With $\varepsilon(T_{ij})=\upsilon_{ij1}B_{1}(T_{ij})+\upsilon_{ij2}B_{2}(T_{ij})+\upsilon_{ij3}B_{3}(T_{ij})\in \tangentspace{\mu(T_{ij})}$, the contaminated observations are 
$$
\begin{aligned}
Y_{ij}
%=&\Exp_{\mu(T_{ij})}\{(T_{ij}Z_{1i}+\varepsilon_{ij1})B_1(T_{ij})+(T_{ij}Z_{2i}+\varepsilon_{ij2})B_{2}(T_{ij})+(T_{ij}Z_{3i}+\varepsilon_{ij3})B_{3}(T_{ij})\}\\
=&\left(
\begin{aligned}
&e^{T_{ij}+T_{ij}Z_{1i}+\upsilon_{ij1}} & 0 \\
&T_{ij}+T_{ij}Z_{3i}+\upsilon_{ij3} & e^{T_{ij}+T_{ij}Z_{2i}+\upsilon_{ij2}} 
\end{aligned}
\right)
\left(
\begin{aligned}
&e^{T_{ij}+T_{ij}Z_{1i}+\upsilon_{ij1}} & T_{ij}+T_{ij}Z_{3i}+\upsilon_{ij3} \\
&0 & e^{T_{ij}+T_{ij}Z_{2i}+\upsilon_{ij2}} 
\end{aligned}
\right),
\end{aligned}
$$
where $Z_{1i},Z_{2i},Z_{3i}\stackrel{i.i.d.}{\sim} \text{Uniform}(-0.1,0.1)$
and $\upsilon_{ij1},\upsilon_{ij2},\upsilon_{ij3}\stackrel{i.i.d.}{\sim} \text{Uniform}(-a,a)$ with $a>0$ set to satisfy $\text{SNR}=5$ defined in \eqref{equ:snr}.

\paragraph*{\textbf{Manifold $\mathrm{Sym}_{AF}^+$}}  
We parameterize $\mathrm{Sym}_{AF}^+$ by the chart 
$$
\phi: (u,v,w)\rightarrow 
\left(
\begin{aligned}
&e^{u} & w \\
&w & e^{v} 
\end{aligned}
\right)
$$
which gives rise to the frame formed by  $B_1(t)=\frac{\partial\phi}{\partial u}$,
$B_2(t)=\frac{\partial\phi}{\partial v}$ and $B_3(t)=\frac{\partial\phi}{\partial w}$.
The random process $X(t)$ is set to 
\begin{align*}
X(t)=&
\left(
\begin{aligned}
& \frac{1}{4}e^{t+tZ_{1}}+\frac{3}{4}e^{t+tZ_{2}} 
&\quad& 
\frac{\sqrt{3}}{4}e^{t+tZ_{1}}-\frac{\sqrt{3}}{4}e^{t+tZ_{2}}\\
&\frac{\sqrt{3}}{4}e^{t+tZ_{1}}-\frac{\sqrt{3}}{4}e^{t+tZ_{2}}
&\quad&
\frac{3}{4}e^{t+tZ_{1}}+\frac{1}{4}e^{t+tZ_{2}}\\
\end{aligned}
\right),
\end{align*}
for $Z_1,Z_2\stackrel{i.i.d.}{\sim} \text{Uniform}(-0.1,0.1)$. The mean function is $\mu(t)=e^{t}\mathbf{I}_{2}$ while the covariance function is $\covarop(s,t)=\text{diag}\{st/300,st/300,0\}$ under the frame $(B_{1},B_{2},B_{3})$, providing an illustration on covariance structure of non-full rank. With  $\varepsilon(T_{ij})=\upsilon_{ij1}B_{1}(T_{ij})+\upsilon_{ij2}B_{2}(T_{ij})\in \tangentspace{\mu(T_{ij})}$, the contaminated observations are 
\begin{align*}
Y_{ij}%=&\Exp_{\mu(T_{ij})}\{(T_{ij}Z_{1i}+\upsilon_{ij1})B_1(T_{ij})+(T_{ij}Z_{2i}+\upsilon_{ij2})B_{2}(T_{ij})\}\\
=&
\left(
\begin{aligned}
& \frac{1}{4}e^{T_{ij}+T_{ij}Z_{1i}+\upsilon_{ij1}}+\frac{3}{4}e^{T_{ij}+T_{ij}Z_{2i}+\upsilon_{ij2}} 
&\quad& 
\frac{\sqrt{3}}{4}e^{T_{ij}+T_{ij}Z_{1i}+\upsilon_{ij1}}-\frac{\sqrt{3}}{4}e^{T_{ij}+T_{ij}Z_{2i}+\upsilon_{ij2}}\\
&\frac{\sqrt{3}}{4}e^{T_{ij}+T_{ij}Z_{1i}+\upsilon_{ij1}}-\frac{\sqrt{3}}{4}e^{T_{ij}+T_{ij}Z_{2i}+\upsilon_{ij2}}
&\quad&
\frac{3}{4}e^{T_{ij}+T_{ij}Z_{1i}+\upsilon_{ij1}}+\frac{1}{4}e^{T_{ij}+T_{ij}Z_{2i}+\upsilon_{ij2}}\\
\end{aligned}
\right),
\end{align*}
where $Z_{1i},Z_{2i},Z_{3i}\stackrel{i.i.d.}{\sim} \text{Uniform}(-0.1,0.1)$
and $\upsilon_{ij1},\upsilon_{ij2}\stackrel{i.i.d.}{\sim} \text{Uniform}(-a,a)$ with $a>0$ set to satisfy $\text{SNR}=5$ defined in \eqref{equ:snr}.

%%%%%%%%%%%%%%%%%%%%%%%%%%

We consider different sample sizes and sampling rates, namely, $n=100,200,400$ and $m=5,10,20,30$. Each simulation is repeated independently 100 times. The kernel adopted is the tricube kernel defined by $K(u)=70(1-|u|^3)^3/81$, and the bandwidths $h_\mu$ and $h_{\covarop}$ are selected by the two-fold cross-validation procedure described in Section \ref{sec:estimation}. Estimation quality is measured by relative mean uniform integrated error (rMUIE) and relative root mean integrated squared error (rRMISE), defined by 
\begin{equation}
\label{equ:rmise}
\begin{aligned}
\text{rMUIE} & :=\frac{\expect\sup_{s,t\in\tdomain}\|\vpt_{(\hat{\mu}(s),\hat{\mu}(t))}^{({\mu}(s),{\mu}(t))}\hat{\covarop}(s,t)-\covarop(s,t)\|_G}{\sup_{s,t\in\tdomain}\|\covarop(s,t)\|_G},\\
\text{rRMISE} & :=\frac{\{\expect\int_{\tdomain^2}\|\vpt_{(\hat{\mu}(s),\hat{\mu}(t))}^{({\mu}(s),{\mu}(t))}\hat{\covarop}(s,t)-\covarop(s,t)\|^{2}_G\diffop s\diffop t\}^{1/2}}{\{\int_{\tdomain^2}\|\covarop(s,t)\|^{2}_G\diffop s\diffop t\}^{1/2}}.
\end{aligned}
\end{equation}

The results, summarized in Tables \ref{tab:simu-rMUIE} and \ref{tab:simu-rRMISE}, show that the estimation errors in terms of both rMUIE and rRMISE {in percentage} decrease as $n$ or $m$ increases, and thus demonstrate the effectiveness of the proposed estimation method. A phase transition phenomenon is also observed: When $m$ is increased from 5 to 10 or 20, the errors in terms of both rMUIE and rRMISE decrease substantially, while when $m$ is further increased to 30, the decrease in errors is marginal. This phenomenon, hinted by our theoretical analysis in Section \ref{sec:theory}, suggests that for a fixed sample size, when $m=5$ or $m=10$ the errors are primarily due to the low sampling rate $m$, while when $m=30$ or higher the errors are mainly contributed by the sample size.

To numerically verify that the proposed framework is invariant to parameterization, we also computed the estimates with a different parameterization of the manifolds in the above. Specifically, we considered the following additional parameterization called stereographic projection
\begin{equation}
\label{equ:simu:stereographic}
\varphi: (u,v)\in\mathbb{R}^{2}\rightarrow \left(\frac{2u}{u^{2}+v^{2}+1},\frac{2v}{u^{2}+v^{2}+1},\frac{u^{2}+v^{2}-1}{u^{2}+v^{2}+1}\right)\in \mathbb{S}^2
\end{equation}	
for the sphere $\mathbb S^2$,  parameterizing the matrices generated in the setting of $\mathrm{Sym}_{LC}^+$ by their lower triangular parts instead of their Cholesky factors, and parameterizing the matrices in the setting of $\mathrm{Sym}_{AF}^+$ by their Cholesky factors instead of their lower triangular parts. In addition, to verify that the results are invariant to frames, for each setting, we consider two sets of randomly selected  frames for computation. We then found that identical results were obtained under different choices of parameterization and/or frames. This numerically demonstrates that  the proposed framework and method are invariant to parameterization and the choice of frames. In addition, the manifold $\mathrm{Sym}_{AF}^+$ does not have a canonical embedding. As a matter of fact, we did not employ an embedding for any of the above manifolds in our studies, demonstrating the intrinsicality of the proposed framework.

\begin{table}
	\renewcommand{\arraystretch}{1.2}
	\caption{rMUIE and its Monte Carlo standard errors under different settings in percentage (\%)\label{tab:simu-rMUIE}}
	\begin{centering}
		\begin{tabular}{|c|c|c|c|c|c|}
			\hline 
			\multirow{2}{*}{manifold}	& \multirow{2}{*}{$n$}   & \multicolumn{4}{c|}{rMUIE} \tabularnewline
			\cline{3-6}
			&  & $m=5$ & $m=10$ & $m=20$ & $m=30$ \tabularnewline
			\hline 
			\multirow{3}{*}{$\mathbb{S}^{2}$} 
			& $100$ & 35.40$\,$(17.50) & 27.60$\,$(12.43) & 18.97$\,$(7.53) & 17.69$\,$(10.91)\tabularnewline
			\cline{2-6}
			& $200$ & 26.36$\,$(12.26) & 20.72$\,$(13.79) & 14.94$\,$(6.07)& 13.85$\,$(4.75) \tabularnewline
			\cline{2-6}
			& $400$ & 18.04$\,$(8.78) & 12.47$\,$(4.41) & 10.48$\,$(2.52) & 8.30$\,$(4.19) \tabularnewline
			\hline 
			\multirow{3}{*}{$\mathrm{Sym}_{LC}^{+}$} 
			& $100$ & 41.58$\,$(13.24) & 36.70$\,$(37.70) & 25.44$\,$(8.42) & 22.10$\,$(5.56) \tabularnewline
			\cline{2-6}
			& $200$ & 30.36$\,$(10.41) & 22.05$\,$(6.51) & 20.89$\,$(7.14)& 15.51$\,$(3.52) \tabularnewline
			\cline{2-6}
			& $400$ & 24.15$\,$(12.30) & 14.55$\,$(5.13) & 12.47$\,$(4.85)& 12.09$\,$(2.46) \tabularnewline
			\hline 
			\multirow{3}{*}{$\mathrm{Sym}_{AF}^{+}$} 
			& $100$ & 35.40$\,$(17.50) & 27.60$\,$(12.43) & 18.97$\,$(7.53) & 18.77$\,$(7.24) \tabularnewline
			\cline{2-6}
			& $200$ & 26.35$\,$(12.26) & 20.72$\,$(13.79) & 14.94$\,$(6.05) & 13.85$\,$(4.75) \tabularnewline
			\cline{2-6}
			& $400$ & 18.04$\,$(8.78) & 12.49$\,$(4.40) & 10.48$\,$(2.52) & 8.30$\,$(4.19) \tabularnewline
			\hline 	
		\end{tabular}
		\par\end{centering}	
\end{table}
\begin{table}			
	\renewcommand{\arraystretch}{1.2}
	\caption{rRMISE and its Monte Carlo standard errors under different settings in percentage (\%)\label{tab:simu-rRMISE}}
	\begin{centering}
		\begin{tabular}{|c|c|c|c|c|c|}
			\hline 
			\multirow{2}{*}{manifold}	& \multirow{2}{*}{$n$}   & \multicolumn{4}{c|}{rRMISE} \tabularnewline
			\cline{3-6}
			&  & $m=5$ & $m=10$ & $m=20$ & $m=30$ \tabularnewline
			\hline 
			\multirow{3}{*}{$\mathbb{S}^{2}$} 
			& $100$ & 24.22$\,$(8.70) & 20.63$\,$(7.97) & 16.10$\,$(6.11) & 15.68$\,$(6.73)\tabularnewline
			\cline{2-6}
			& $200$ & 17.16$\,$(5.73) & 14.00$\,$(6.14) & 12.10$\,$(4.48) & 11.89$\,$(4.55)\tabularnewline
			\cline{2-6}
			& $400$ & 11.99$\,$(4.45) & 9.29$\,$(3.03) & 8.81$\,$(1.81) & 6.90$\,$(3.49) \tabularnewline
			\hline 
			\multirow{3}{*}{$\mathrm{Sym}_{LC}^{+}$} 
			& $100$ & 29.52$\,$(7.20) & 25.98$\,$(12.15) & 21.66$\,$(6.63) & 19.01$\,$(2.98) \tabularnewline
			\cline{2-6}
			& $200$ & 21.13$\,$(5.19) & 16.27$\,$(3.81) & 18.99$\,$(6.05) & 13.95$\,$(2.94) \tabularnewline
			\cline{2-6}
			& $400$ & 16.29$\,$(4.33) & 11.04$\,$(2.54) & 10.99$\,$(4.33) & 10.08$\,$(2.37)\tabularnewline
			\hline 
			\multirow{3}{*}{$\mathrm{Sym}_{AF}^{+}$} 
			& $100$ & 24.22$\,$(8.70) & 20.63$\,$(7.97) & 16.10$\,$(6.11) & 15.37$\,$(5.75) \tabularnewline
			\cline{2-6}
			& $200$ & 17.16$\,$(5.73) & 14.00$\,$(6.14) & 13.21$\,$(4.96) & 11.89$\,$(4.55) \tabularnewline
			\cline{2-6}
			& $400$ & 11.99$\,$(4.45) & 10.55$\,$(3.55) & 9.81$\,$(1.81) & 6.90$\,$(3.49) \tabularnewline
			\hline 	
		\end{tabular}
		\par\end{centering}	
\end{table}

\section[Application]{Application to longitudinal diffusion tensors}\label{sec:application}
We apply the proposed framework to analyze longitudinal diffusion tensors from Alzheimer's Disease Neuroimaging Initiative (ADNI) database. The ADNI was launched in 2003 as a public-private partnership, led by Principal Investigator Michael W. Weiner, MD. The primary goal of ADNI has been to test whether serial magnetic resonance imaging (MRI), positron emission tomography (PET), other biological markers, and clinical and neuropsychological assessment can be combined to measure the progression of mild cognitive impairment (MCI) and early Alzheimer's disease (AD). For up-to-date information, see \url{www.adni-info.org}. 

Diffusion tensor imaging (DTI), a special kind of diffusion-weighted magnetic resonance imaging, has  been extensively adopted in brain science to investigate white matter tractography. In a DTI image, each brain voxel is associated with a $3\times 3$ symmetric positive-definite matrix, called diffusion tensor, that characterizes diffusion of water molecules in the voxel. As diffusion of water molecules carries rich information about axons,  diffusion tensor imaging has important applications in both clinical diagnostics and scientific research related to brain diseases. From a statistical perspective, diffusion tensors are modeled as random elements in $\dtispd$, and have been studied extensively, such as \cite{Fillard2005,arsigny2006,Lenglet2006,Pennec2006,Fletcher2007,Dryden2009,Zhu2009,Pennec2020}, among many others. In these works $\dtispd$ is endowed with a Riemannian metric or a non-Euclidean distance that aims to alleviate or completely eliminate swelling effect \citep{Arsigny2007}. However, none of them consider the longitudinal aspect of diffusion tensors.

We focus on the hippocampus, a brain region that plays an important role in memory and is central to Alzheimer's disease \citep{Lindberg2012}, and include in the study subjects with at least four properly recorded DTI images. This results in a sample of $n=177$ subjects with age ranging from 55.2 to 93.5. Among them, 42 subjects are cognitively normal (CN), while the others (AD) developed one of early mild cognitive impairment, mild cognitive impairment, late mild cognitive impairment and Alzheimer's disease. On average, there are $m=5.5$ DTI scans for each subject, which shows that the data are rather sparsely recorded. A standard procedure that includes denoising, eddy current and motion correction, skull stripping, bias correction and normalization is adopted to preprocess the raw images. Based on the preprocessed DTI images, diffusion tensors are derived. We endow $\dtispd$ with the Log-Cholesky metric \citep{Lin2019Riemannian} and turn it into a Riemannian manifold of nonpositive sectional curvature. Under the Log-Cholesky framework that avoids swelling effect and meanwhile enjoys computational efficiency, the Fr\'echet mean of the tensors inside hippocampus is calculated for each DTI scan, which represents a coarse-grain summary of hippocampal  diffusion tensors. As we shall see below, this averaged mean tensor is already capable of illuminating some differences of the diffusion dynamics  between the AD and CN groups.    %Here, the Fr\'echet means were calculated under the Log-Cholesky framework  which avoids swelling effect and meanwhile enjoys computational efficiency. Our goal is to investigate and compare the mean trajectories and first few intrinsic principal components of the averaged hippocampal diffusion tensors of the CN and AD groups.

The estimated Fr\'echet mean trajectories are depicted in Figure \ref{fig:dti-mean} with the bandwith 4.2 for the AD group and 5.7 for the CN group, where each tensor is visualized as an ellipsoid whose volume corresponds to the determinant of the tensor. They suggest that, overall the averaged hippocampal diffusion tensor remains rather stable for the CN group; the tensors at age 55.2 and 93.5 that markedly depart from the others could be due to boundary effect, i.e., there are relatively less data around the two boundary time points. In contrast, for the AD group,  the dynamic tensor varies more substantially, and  the diffusion (measured by the determinant of tensors and indicated by volume of ellipsoids) seems larger. Also, the mean trajectory of the AD group exhibits slightly lower fractional anisotropy at each time point. Fractional anisotropy, defined for each $3\times 3$ symmetric positive-definite matrix $A$ by \begin{equation*}\text{FA}=\sqrt{\frac{3}{2}\frac{(\rho_1-\bar \rho)^2+(\rho_2-\bar\rho)^2+(\rho_3-\bar\rho)^2}{\rho_1^2+\rho_2^2+\rho_3^2}}\end{equation*} where $\rho_1,\rho_2,\rho_3$ are eigenvalues of $A$ and $\bar\rho=(\rho_1+\rho_2+\rho_3)/3$, describes the degree of anisotropy of diffusion of water molecules. It is close to zero unless movement of the water molecules is constrained by structures such as white matter fibers. The below-normal fractional anisotropy might  suggest some damage on the hippocampal structure for the AD group.

For the covariance function, Figure \ref{fig:dti-pc} shows the first three intrinsic Riemannian functional principal components that are mapped on $\dtispd$ via the Riemannian exponential maps $\Exp_{\hat\mu(t)}$, where the bandwidth is 3.5 for the AD group and 4.5 for the CN group. They respectively account for 40.2\%,  22.2\% and 7.0\% of variance for the AD group, and 40.7\%, 19.4\% and 8.0\% of variance for the CN group. These components, compared side by side in Figure \ref{fig:dti-pc}, exhibit different patterns  between the two cohorts. For instance, the Riemannian functional principal components of the AD group show relatively larger diffusion and more dynamics over time. In addition, they exhibit relatively lower fractional anisotropy, which suggests that individual diffusion tensor trajectories in the AD group tend to deviate from their mean trajectory along the direction with below-normal fractional anisotropy.

%In summary, via the proposed framework we find that there are differences in longitudinal development of hippocampus between cognitively normal subjects and patients with Alzheimer's disease. 
We conclude this section by the following remarks. Note that,  
in the above analysis, the averaged hippocampal diffusion tensors do not capture the rich spatial information of all tensors within the hippocampus. To account for such information, all hippocampal diffusion tensors shall be taken into consideration by being modeled as an $\dtispd$-valued function defined on the hippocampal region which is a three-dimensional domain of $\real^3$. Along with the temporal dynamics, for each subject there are spatiotemporal Riemannian manifold-valued data, with the sparseness along the temporal direction.  Our framework can be extended to analyze such data, but the extension requires substantial development and is left for future study. 

In addition, each of the sparse trajectories is only observed in an individual-specific period shorter than the span ($93.6-55.2=38.4$ years) of the entire study. Functional data of this feature, called functional fragments  \citep{Delaigle2019,Descary2019recovering} or functional snippets \citep{Lin2021}, require special treatment on estimating the covariance structure. Particularly, local smoothing techniques can only estimate the diagonal region of the covariance function for such data and thus require the additional assumption that the covariance function is supported in the diagonal region, as we have  done implicitly in the above analysis. Extension of the estimation method proposed in this paper to functional fragments/snippets is nontrivial and thus also left for future study.

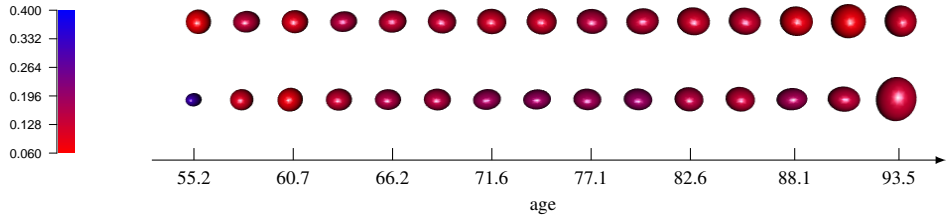
\begin{figure}[t]  
	\centering
	\begin{minipage}[c]{\linewidth}
		\centering
		\begin{tikzpicture}[scale=0.8, every node/.style={scale=0.8}]
		\newcommand\x{0.8}
		\newcommand\del{1.65}
		\newcommand\xstart{-5.8}
		
		\node at (0+\x,0) {\includegraphics[width=0.9\textwidth]{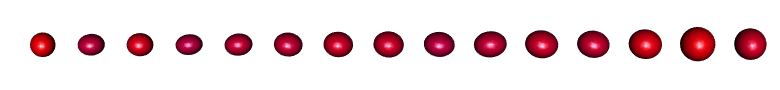}};
		\node at (0+\x,-1.3) {\includegraphics[width=0.9\textwidth]{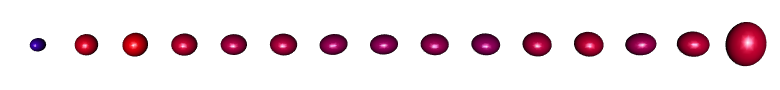}};
		
		\node at (-7.5,-1) {\includegraphics[scale=0.5]{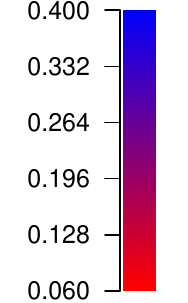}};
		
		\draw[-latex] (-6.5+\x,-2.3) -- (7.5,-2.3);
		%\node at (-8,-2.35) {age =};
		
		\draw (\xstart+\x,-2.1) -- (\xstart+\x,-2.3);
		\node at (\xstart+\x,-2.6) {55.2};
		
		\draw (\xstart+\del+\x,-2.1) -- (\xstart+\del+\x,-2.3);
		\node at (\xstart+\del+\x,-2.6) {60.7};
		
		\draw (\xstart+2*\del+\x,-2.1) -- (\xstart+2*\del+\x,-2.3);
		\node at (\xstart+2*\del+\x,-2.6) {66.2};
		
		\draw (\xstart+3*\del+\x,-2.1) -- (\xstart+3*\del+\x,-2.3);
		\node at (\xstart+3*\del+\x,-2.6) {71.6};
		
		\draw (\xstart+4*\del+\x,-2.1) -- (\xstart+4*\del+\x,-2.3);
		\node at (\xstart+4*\del+\x,-2.6) {77.1};
		
		\draw (\xstart+5*\del+\x,-2.1) -- (\xstart+5*\del+\x,-2.3);
		\node at (\xstart+5*\del+\x,-2.6) {82.6};
		
		\draw (\xstart+6.05*\del+\x,-2.1) -- (\xstart+6.05*\del+\x,-2.3);
		\node at (\xstart+6.05*\del+\x,-2.6) {88.1};
		
		\draw (\xstart+7.1*\del+\x,-2.1) -- (\xstart+7.1*\del+\x,-2.3);
		\node at (\xstart+7.1*\del+\x,-2.6) {93.5};
		
		\node at (0+\x,-3.1) {age};
		
		\end{tikzpicture} 
	\end{minipage}
	\caption{Mean functions. Top: AD group; bottom: CN group. The color encodes fractional anisotropy.\label{fig:dti-mean}}
\end{figure}

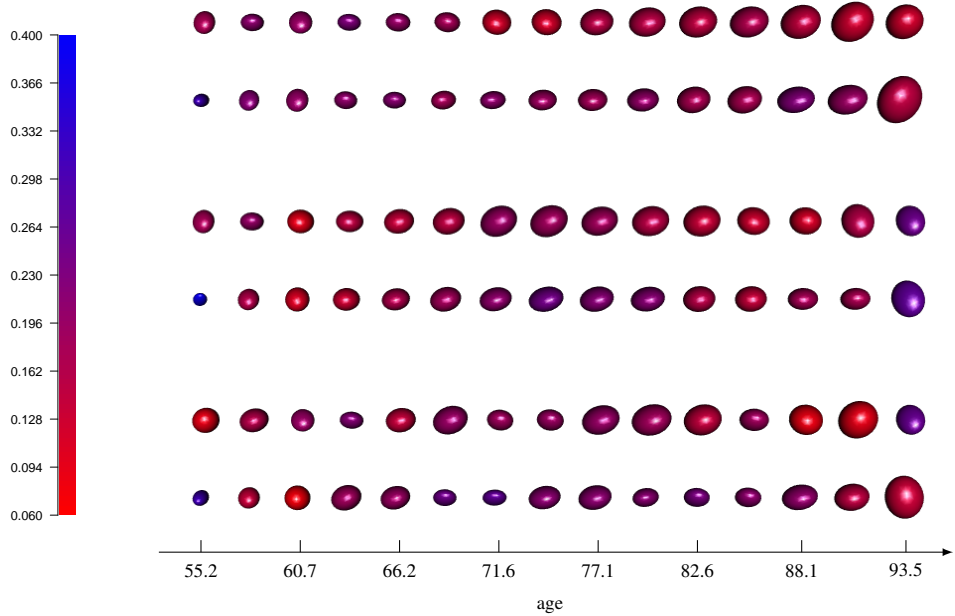
\begin{figure}[t]  
	\centering
	\begin{minipage}[c]{\linewidth}
		\centering
		\begin{tikzpicture}[scale=0.8, every node/.style={scale=0.8}]
		\newcommand\x{0.8}
		\newcommand\del{1.65}
		\newcommand\xstart{-5.8}
		
		\node at (0+\x,0) {\includegraphics[width=0.9\textwidth]{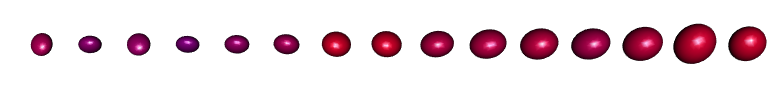}};
		\node at (0+\x,-1.3) {\includegraphics[width=0.9\textwidth]{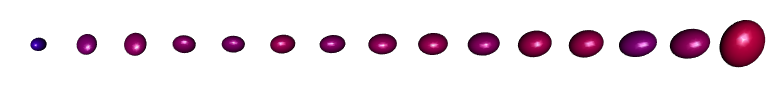}};
		
		\node at (0+\x,-3.3) {\includegraphics[width=0.9\textwidth]{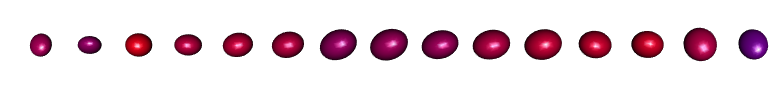}};
		\node at (0+\x,-4.6) {\includegraphics[width=0.9\textwidth]{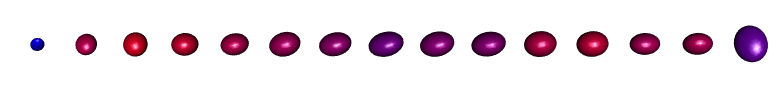}};
		
		\node at (0+\x,-6.6) {\includegraphics[width=0.9\textwidth]{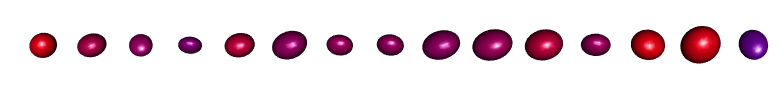}};
		\node at (0+\x,-7.9) {\includegraphics[width=0.9\textwidth]{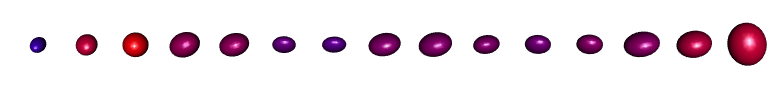}};
		
		\node at (-7.6,-4.2) {\includegraphics[scale=0.5]{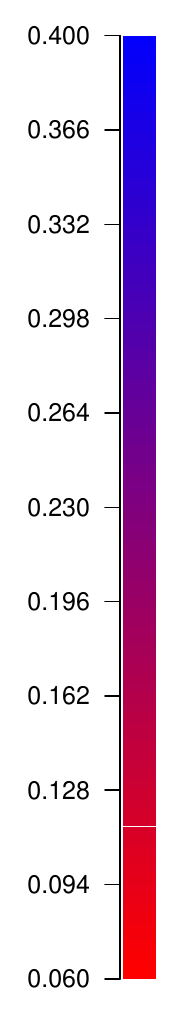}};
		
		\draw[-latex] (-6.5+\x,-8.8) -- (7.5,-8.8);

		\draw (\xstart+\x,-8.6) -- (\xstart+\x,-8.8);
		\node at (\xstart+\x,-9.1) {55.2};
		
		\draw (\xstart+\del+\x,-8.6) -- (\xstart+\del+\x,-8.8);
		\node at (\xstart+\del+\x,-9.1) {60.7};
		
		\draw (\xstart+2*\del+\x,-8.6) -- (\xstart+2*\del+\x,-8.8);
		\node at (\xstart+2*\del+\x,-9.1) {66.2};
		
		\draw (\xstart+3*\del+\x,-8.6) -- (\xstart+3*\del+\x,-8.8);
		\node at (\xstart+3*\del+\x,-9.1) {71.6};
		
		\draw (\xstart+4*\del+\x,-8.6) -- (\xstart+4*\del+\x,-8.8);
		\node at (\xstart+4*\del+\x,-9.1) {77.1};
		
		\draw (\xstart+5*\del+\x,-8.6) -- (\xstart+5*\del+\x,-8.8);
		\node at (\xstart+5*\del+\x,-9.1) {82.6};
		
		\draw (\xstart+6.05*\del+\x,-8.6) -- (\xstart+6.05*\del+\x,-8.8);
		\node at (\xstart+6.05*\del+\x,-9.1) {88.1};
		
		\draw (\xstart+7.1*\del+\x,-8.6) -- (\xstart+7.1*\del+\x,-8.8);
		\node at (\xstart+7.1*\del+\x,-9.1) {93.5};
		
		\node at (0+\x,-9.7) {age};
		
		\end{tikzpicture} 
	\end{minipage}
	\caption{The first principal component of AD group (Row 1) and CN group (Row 2), the second principal component of AD group (Row 3) and CN group (Row 4), and the third principal component of AD group (Row 5) and CN group (Row 6).  The color encodes fractional anisotropy.\label{fig:dti-pc}}
\end{figure}

%%%%%%%%%%%%%%%%%%%%%%%%%%%%%%%%%%%%%%%%%%%%%%
%% Support information (funding), if any,   %%
%% should be provided in the                %%
%% Acknowledgements section.                %%
%%%%%%%%%%%%%%%%%%%%%%%%%%%%%%%%%%%%%%%%%%%%%%
\section*{Acknowledgements}
Lingxuan Shao was a visiting student of Zhenhua Lin in National University of Singapore at the time of developing the paper. Lingxuan Shao and Zhenhua Lin are the joint first authors, and Fang Yao is the corresponding author. Zhenhua Lin's research is partially supported by NUS startup grant R-155-000-217-133. Fang Yao's research is partially supported by National Natural Science Foundation of China Grants 11931001 and 11871080, the National Key R\&D Program of China Grant 2020YFE0204200, the LMAM, and the Key Laboratory of Mathematical Economics and Quantitative Finance (Peking University), Ministry of Education. 
Data collection and sharing for this project was funded by the Alzheimer's Disease Neuroimaging Initiative (ADNI) (National Institutes of Health Grant U01 AG024904) and DOD ADNI (Department of Defense award number W81XWH-12-2-0012). ADNI is funded by the National Institute on Aging, the National Institute of Biomedical Imaging and Bioengineering, and through generous contributions from the following: AbbVie, Alzheimer's Association; Alzheimer's Drug Discovery Foundation; Araclon Biotech; BioClinica, Inc.; Biogen; Bristol-Myers Squibb Company; CereSpir, Inc.; Cogstate; Eisai Inc.; Elan Pharmaceuticals, Inc.; Eli Lilly and Company; EuroImmun; F. Hoffmann-La Roche Ltd and its affiliated company Genentech, Inc.; Fujirebio; GE Healthcare; IXICO Ltd.; Janssen Alzheimer Immunotherapy Research \& Development, LLC.; Johnson \& Johnson Pharmaceutical Research \& Development LLC.; Lumosity; Lundbeck; Merck \& Co., Inc.; Meso Scale Diagnostics, LLC.; NeuroRx Research; Neurotrack Technologies; Novartis Pharmaceuticals Corporation; Pfizer Inc.; Piramal Imaging; Servier; Takeda Pharmaceutical Company; and Transition Therapeutics. The Canadian Institutes of Health Research is providing funds to support ADNI clinical sites in Canada. Private sector contributions are facilitated by the Foundation for the National Institutes of Health (\url{www.fnih.org}). The grantee organization is the Northern California Institute for Research and Education, and the study is coordinated by the Alzheimer's Therapeutic Research Institute at the University of Southern California. ADNI data are disseminated by the Laboratory for Neuro Imaging at the University of Southern California.

%Data used in preparation of this article were obtained from the Alzheimer's Disease Neuroimaging Initiative (ADNI) database (\url{adni.loni.usc.edu}). As such, the investigators within the ADNI contributed to the design and implementation of ADNI and/or provided data but did not participate in analysis or writing of this report. A complete listing of ADNI investigators can be found at:\url{http://adni.loni.usc.edu/wp-content/uploads/how_to_apply/ADNI_Acknowledgement_List.pdf}.
% 
% The second author was supported in part by ...
 
%%%%%%%%%%%%%%%%%%%%%%%%%%%%%%%%%%%%%%%%%%%%%%
%% Supplementary Material, if any, should   %%
%% be provided in {supplement} environment  %%
%% with title inside \textbf{} and short    %%
%% description below.                       %%
%%%%%%%%%%%%%%%%%%%%%%%%%%%%%%%%%%%%%%%%%%%%%%
%\begin{supplement}
%This supplementary material contains proofs for Lemma~\ref{lem:c_tildec} and Proposition~\ref{prop:random:mu:rate:local}. 
%\end{supplement}

\begin{center}
	SUPPLEMENTARY MATERIAL
\end{center}

The supplementary material contains some preliminaries for Riemannian geometry, the asymptotic distribution of the proposed covariance estimator, proofs, theoretical results for the regular design, and further illustrations of the invariance property. The code and data are hosted at \url{https://github.com/linulysses/iRFDA-sparse}.

%%%%%%%%%%%%%%%%%%%%%%%%%%%%%%%%%%%%%%%%%%%%%%
%% Single Appendix:                         %%
%%%%%%%%%%%%%%%%%%%%%%%%%%%%%%%%%%%%%%%%%%%%%%

%%%%%%%%%%%%%%%%%%%%%%%%%%%%%%%%%%%%%%%%%%%%%%
%% Multiple Appendixes:                     %%
%%%%%%%%%%%%%%%%%%%%%%%%%%%%%%%%%%%%%%%%%%%%%%
%\begin{appendix}
%\section{???}
%
%\section{???}
%
%\end{appendix}

%%%%%%%%%%%%%%%%%%%%%%%%%%%%%%%%%%%%%%%%%%%%%%%%%%%%%%%%%%%%%
%%                  The Bibliography                       %%
%%                                                         %%
%%  imsart-???.bst  will be used to                        %%
%%  create a .BBL file for submission.                     %%
%%                                                         %%
%%  Note that the displayed Bibliography will not          %%
%%  necessarily be rendered by Latex exactly as specified  %%
%%  in the online Instructions for Authors.                %%
%%                                                         %%
%%  MR numbers will be added by VTeX.                      %%
%%                                                         %%
%%  Use \cite{...} to cite references in text.             %%
%%                                                         %%
%%%%%%%%%%%%%%%%%%%%%%%%%%%%%%%%%%%%%%%%%%%%%%%%%%%%%%%%%%%%%

%% if your bibliography is in bibtex format, uncomment commands:
\bibliographystyle{imsart-nameyear} % Style BST file (imsart-number.bst or imsart-nameyear.bst)
\bibliography{ref}       % Bibliography file (usually '*.bib')

%% or include bibliography directly:
% \begin{thebibliography}{}
% \bibitem{b1}
% \end{thebibliography}



\section*{Supplementary material}

# Supplementary Material to "Intrinsic Riemannian Functional Data Analysis for Sparse Longitudinal Observations"

Lingxuan Shao*, Zhenhua Lin†, and Fang Yao*

## S.1  Additional Background on Riemannian Manifold

Here we provide formal definitions related to Riemannian manifolds, starting with perhaps the most basic geometric space — the topological space.

**Definition S.1** (Topological space). *A topological space is a set and a class of subsets of the set, called the open sets, such that the class contains the empty set and is closed under the formation of arbitrary unions and finite intersections. Such a class is called a topology.*

The most commonly seen topological space is $\mathbb{R}$ together with the standard (but often not explicitly mentioned) topology that contains all open intervals, as well as the $d$-dimensional Euclidean space $\mathbb{R}^d$ together with the standard topology that contains all open balls. In real analysis, continuous functions play an important role. The concept of continuity can be generalized to functions defined on and/or taking values in general topological spaces.

**Definition S.2** (Continuity and homeomorphism). *A function $f : \mathcal{T}_1 \to \mathcal{T}_2$ between two topological spaces is continuous if to every open set $B$ of $\mathcal{T}_2$, the set $f^{-1}(B) := \{x \in \mathcal{T}_1 : f(x) \in B\}$ is an open set of $\mathcal{T}_1$. If $f$ is continuous and bijective and its inverse is also continuous, then we say $f$ is a homeomorphism between $\mathcal{T}_1$ and $\mathcal{T}_2$. Two topological spaces are homeomorphic to each other if there exists a homeomorphism between them.*

When both $\mathcal{T}_1$ and $\mathcal{T}_2$ are $\mathbb{R}$ with the standard topology, the continuity defined in the above coincides with the one via the $\epsilon$-$\delta$ definition. For instance, it is well know that $f(x) = x^3$ is a continuous function, and indeed, to every open interval $(a, b)$, $f^{-1}((a, b)) = (a^{1/3}, b^{1/3})$ is also an open interval, and this holds for all open sets of $\mathbb{R}$; see Example 1 of Section 18 in Munkres (2000). A property of a topological space is a topological property if it is preserved under a homeomorphism. Therefore, homeomorphic topological spaces share the same set of topological properties. A concrete example of topological properties is connectedness that is related to Assumption 4.4(a). A topological space $\mathcal{T}$ is **connected** if it is not the union of two disjoint non-empty open sets. For example, $\mathbb{R}$ is connected, but $\mathbb{R}\backslash\{0\}$ is not. If to every two points $p, q \in \mathcal{T}$ there is a continuous path connecting them, i.e., there exists a continuous function $\gamma : [0, 1] \to \mathcal{T}$ such that $\gamma(0) = p$ and $\gamma(1) = q$, then we say $\mathcal{T}$ is **path-connected**. A topological space $\mathcal{T}$ is **simply connected** if and only if it is path-connected, and for any two continuous paths $\gamma_1$ and $\gamma_2$ with the same start and end points, $\gamma_1$ can be continuously deformed into $\gamma_2$, i.e., there exists a continuous function $\gamma : [0, 1]^2 \to \mathcal{M}$ such that $\gamma(s, 0) = \gamma_1(s)$ and $\gamma(s, 1) = \gamma_2(s)$ for $s \in [0, 1]$. Intuitively speaking, a topological space is simply

---
*School of Mathematical Sciences, Center for Statistical Science, Peking University
†Department of Statistics and Data Science, National University of Singapore

connected if there are no holes passing through the space. For instance, $\mathbb{R}^d$ and spheres are simply connected, while the torus is not. Here, connectedness, path-connectedness and simple connectedness are all topological properties, i.e., preserved under homeomorphisms.

There are topological spaces that locally resemble a Euclidean space $\mathbb{R}^d$, but globally may not be homeomorphic to $\mathbb{R}^d$. Such spaces are called topological manifolds. An example of topological manifolds is the surface of our earth that may be roughly parameterized by the two-dimensional sphere $\mathbb{S}_r^2 = \{(x_1, x_2, x_3) \in \mathbb{R}^3 : x_1^2 + x_2^2 + x_3^2 = r^2\}$ of radius $r \approx 6371$km. Each small neighborhood of $\mathbb{S}_r^2$ looks like a (subset of) two-dimensional plane, and this is why our ancients perceived the flat earth model. However, there exists no homeomorphism between $\mathbb{S}_r^2$ and $\mathbb{R}^d$ for any $d$.

**Definition S.3** (Topological manifold)**.** *A topological space $\mathcal{T}$ is a topological manifold modeled on $\mathbb{R}^d$, if for each point $p \in \mathcal{T}$ there exists an open set that contains $p$ and is homeomorphic to $\mathbb{R}^d$. Here, $d$ is called the dimension of $\mathcal{T}$.*

In the above example of $\mathbb{S}_r^2$, let $U_1 = \mathbb{S}_r^2 \setminus \{(r, 0, 0)\}$ and $U_2 = \mathbb{S}_r^2 \setminus \{(-r, 0, 0)\}$. Then every point in $\mathbb{S}_r^2$ falls into one of these two open sets. Moreover, both $U_1$ and $U_2$ are homeomorphic to $\mathbb{R}^d$, with the following corresponding homeomorphisms

$$\phi_1(x_1, x_2, x_3) = \frac{r}{r - x_1}(x_2, x_3) \in \mathbb{R}^2$$

$$\phi_2(x_1, x_2, x_3) = \frac{r}{r + x_1}(x_2, x_3) \in \mathbb{R}^2.$$

This formally shows that $\mathbb{S}_r^2$ is a topological manifold of dimension 2.

Due to the local resemblance between $\mathbb{R}^d$ and a topological manifold, one might parameterize a local neighborhood of the topological manifold by using $\mathbb{R}^d$, as we did in the above for neighborhoods $U_1$ and $U_2$ of $\mathbb{S}_r^2$. Intuitively, the map $\phi_1$ (resp. $\phi_2$) assigns each point in $U_1$ (resp. $U_2$) a coordinate in $\mathbb{R}^2$. Since $U_1 \cup U_2 = \mathbb{S}_r^2$, each point gets a coordinate. However, for those points in $U_1 \cap U_2$, such as the points in the equator, are assigned two coordinates, one from $\phi_1$ and the other from $\phi_2$. In this case, one can obtain the coordinate under $\phi_1$ if we know its coordinate under $\phi_2$, and vice versa. For example, for $(x_1, x_2, x_3) \in U_1 \cap U_2$, if its coordinate under $\phi_1$ is $(y_1, z_1) \in \mathbb{R}^2$, then its coordinate under $\phi_2$ is $(y_2, z_2) = (\phi_2 \circ \phi_1^{-1})(y_1, z_1)$, where $\circ$ represents the composition of functions, i.e., $(f \circ g)(x) = f(g(x))$ for generic functions $f$, $g$ and argument $x$. Moreover, $\phi_2 \circ \phi_1^{-1}$, called a transition map, is a continuous function, since both $\phi_1$ and $\phi_2$ are homeomorphisms. Noting that $\phi_2 \circ \phi_1^{-1}$ is a function mapping $\mathbb{R}^2$ onto $\mathbb{R}^2$, we might impose regularity more than continuity on the transition maps, for instance, differentiability or smoothness. This consideration leads to the concept of differentiable and smooth manifolds.

**Definition S.4** (Differentiable and smooth manifolds)**.** *Let $\mathcal{M}$ be a d-dimensional topological manifold. For $k \geq 1$, a $C^k$-atlas on $\mathcal{M}$ is a collection of pairs $(U_\alpha, \phi_\alpha)$ that are indexed by an index set $J$ and satisfy the following axioms:*

- *Each $U_\alpha$ is an open subset of $\mathcal{M}$ and $\bigcup_{\alpha \in J} U_\alpha = \mathcal{M}$, i.e., the domains $U_\alpha$ together cover $\mathcal{M}$;*

- *Each $\phi_\alpha$ is a homeomorphism between $U_\alpha$ and the open set $\phi_\alpha(U_\alpha) = \{\phi_\alpha(x) \in \mathbb{R}^d : x \in U_\alpha\}$ of $\mathbb{R}^d$;*

- *For each pair of $\alpha, \beta \in J$, if $U_\alpha \cap U_\beta \neq \varnothing$, then the transition map $\phi_\alpha \circ \phi_\beta^{-1} : \phi_\beta(U_\alpha \cap U_\beta) \to \phi_\alpha(U_\alpha \cap U_\beta)$, illustrated in Figure S.1, is $k$ times differentiable; we say $\phi_\alpha$ and $\phi_\beta$ are compatible.*

*Two $C^k$-atlases are compatible if their union is again a $C^k$-atlas. An atlas is maximal if it contains any other atlas compatible with it. A $C^k$-manifold is a topological manifold together with a maximal $C^k$-atlas. A $C^\infty$-manifold is called a smooth manifold.*

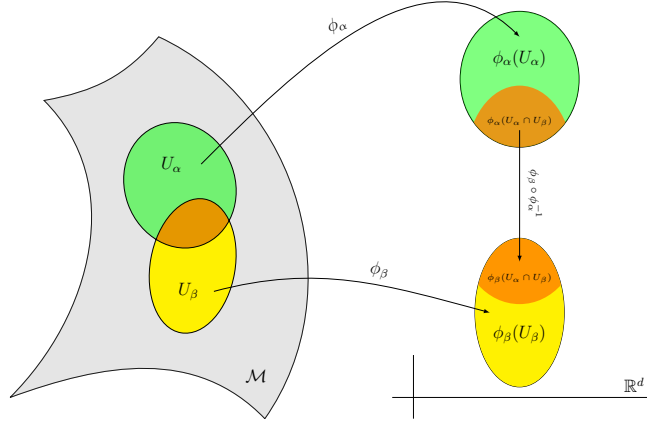

Figure S.1: Illustration of the chart and transition map.

The pair $(U_\alpha, \phi_\alpha)$ or sometimes $\phi_\alpha$ itself is called a **chart** (or **coordinate map**). One can check that compatibility defines an equivalence relation among atlases, and a maximal atlas is simply the union of atlases within the same equivalence class. Therefore, a differentiable manifold is essentially completely determined by an atlas, in the sense that any compatible atlas gives to the same differentiable manifold and incompatible atlases result in distinct differentiable manifolds. In light of this, in practice, we can describe a differentiable manifold simply by providing an atlas. For instance, in the example of $\mathbb{S}_r^2$, the collection $\mathscr{A} = \{(U_1, \phi_1), (U_2, \phi_2)\}$ forms a $C^\infty$-atlas that turns $\mathbb{S}_r^2$ into a smooth manifold.

Let $\mathcal{M}$ be a $d$-dimensional differentiable manifold in the sequel. One merit of the differentiable manifold is that we can discuss regularity of functions taking values in a differentiable manifold. For instance, for an interval $I \subset \mathbb{R}$, the function $\gamma : I \to \mathcal{M}$, which is also called a curve on $\mathcal{M}$, is differentiable at $t \in I$ if the function $\phi \circ \gamma : I \to \mathbb{R}^d$ is differentiable at $t$ for a chart $(U, \phi)$, and thus all charts, such that $\gamma(t) \in U$ in the maximal atlas associated with $\mathcal{M}$. The derivative, denoted by $\gamma'(t)$, measuring the velocity of the curve $\gamma$ at $t$, has different representations in different charts. However, once we know its representation in one chart, we can then obtain its representation in another chart via the transition maps. In this sense, the velocity of $\gamma$ at $t$ is essentially well defined, and is denoted by $\gamma'(t)$. Note that this sense of "well-definedness" applies generally to other concepts and quantities in differential geometry, that is, a manifold-related concept, such as the differentiability of a curve, is well defined if it holds for all relevant charts, and a manifold-related quantity, such the derivative of a differentiable curve at $t$, is well defined if its coordinate in one chart can be determined from the coordinate under another chart.

Consider all differentiable curves $\gamma : I \to \mathcal{M}$ with $\gamma(t) = p$. They may give rise to different derivatives $\gamma'(t)$. If we fix any chart $(U, \phi)$ with $p \in U$, then each $\gamma'(t)$ is represented by $(\phi \circ \gamma)'(t)$ which is a vector in $\mathbb{R}^d$ and all possible values of $(\phi \circ \gamma)'(t)$ form exactly the vector space $\mathbb{R}^d$. Therefore, with $\Gamma_p^t$ denoting the collection of differentiable curves $\gamma : I \to \mathcal{M}$ such that $\gamma(t) = p$, we can view the space $T_p\mathcal{M} := \{\gamma'(t) : \gamma \in \Gamma_p^t\}$ as a vector space with the vector addition $\gamma_1'(t) + \gamma_2'(t)$ defined to be $\gamma_3'(t) \in T_p\mathcal{M}$ such that $(\phi \circ \gamma_1)'(t) + (\phi \circ \gamma_2)'(t) = (\phi \circ \gamma_3)'(t)$, and the scalar multiplication $a\gamma'(t)$ defined to be $\eta'(t)$ such that $a(\phi \circ \gamma)'(t) = a(\phi \circ \eta)'(t)$, where $\gamma, \gamma_1, \gamma_2, \gamma_3, \eta \in \Gamma_p^t$ and $a \in \mathbb{R}$. Note that, the representation of both $\gamma_1'(t) + \gamma_2'(t)$ and $a\gamma'(t)$ in other charts is determined from its representation in $(U, \phi)$, and thus they are well defined manifold-related quantities. The space $T_p\mathcal{M}$ is called the **tangent space** at $p$ and

its elements are called **tangent vectors** at $p$. Such (informal) definition also shows that $T_p\mathcal{M}$ depends on $p$ through its dependence on $\Gamma_p^t$. This agrees well with the physical meaning of the tangent vector $\gamma'(t)$ that represents the direction and amount to move if one wants to get to $\gamma(t + \Delta t)$ from $\gamma(t)$ within an infinitesimal amount time $\Delta t$, i.e., $\gamma'(t)$ encodes both the velocity and the base point $p$. In the chart $(U, \phi)$, $\gamma'(t)$ can be represented by $(\phi(p), (\phi \circ \gamma)'(t))$, $\gamma_1'(t) + \gamma_2'(t)$ by $(\phi(p), (\phi \circ \gamma_1)'(t) + (\phi \circ \gamma_2)'(t))$, and $a\gamma'(t)$ by $(\phi(p), a(\phi \circ \gamma)'(t))$, respectively; here note that in the above coordinate representation of the vector addition and scalar multiplication operations, the coordinate $\phi(t)$ for the base point $p$ does not participate in the operations.

The formal definition below of the tangent space and tangent vectors, with the same essence of the above informal discussion, eliminates the dependence on $I$ and $\gamma(t) = p$. We say two differentiable curves $\gamma$ and $\eta$ on a differentiable manifold $\mathcal{M}$ is equivalent at $p \in \mathcal{M}$, denoted by $\gamma \sim_p \eta$, if $\gamma(s) = \eta(t) = p$ and $(\phi \circ \gamma)'(s) = (\phi \circ \eta)'(t)$ for some $s, t \in \mathbb{R}$ and for all $(U, \phi)$ with $p \in U$. One can check that for each fixed $p \in \mathcal{M}$, $\sim_p$ defines an equivalence relation among differentiable curves passing $p$.

**Definition S.5** (Tangent space and tangent vector). *Each class of equivalent differentiable curves at $p$ is a tangent vector at $p$, and the collection of tangent vectors at $p$ is the tangent space at $p$.*

The above formal definition seems abstract, but we can perceive tangent vectors and tangent spaces by using the aforementioned essentially equivalent informal definition, as well as the concrete coordinate representations. One thing we shall emphasize is that tangent vectors at distinct points are not directly comparable, since they have distinct base points. This is also evidenced by the abstract definition, as tangent vectors at distinct points are equivalence classes of distinct equivalence relation. For example, a tangent vector at $p$ is an equivalence class under the relation $\sim_p$, while a tangent vector at $q$ is an equivalence class under the relation $\sim_q$. In general, equivalence classes derived from different equivalence relations are not directly comparable.

For some manifolds, the tangent space might be visualized as a hyperplane that is tangent to the manifold. For instance, as depicted in Figure S.2, the tangent space at $p \in \mathbb{S}_r^2$ can be viewed as the two-dimensional affine plane that is tangent to $\mathbb{S}_r^2$ at $p$; such affine plane is regarded as a vector space with $p$ as the origin.

In elementary calculus, a function that maps a subset of a Euclidean space into potentially another Euclidean space is smooth if it is infinite times differentiable. Based on this elementary concept of smoothness, we can define smoothness for functions between manifolds.

**Definition S.6** (Smooth manifold map). *Let $\mathcal{M}$ be a $d_1$-dimensional smooth manifold and $\mathcal{N}$ a $d_2$-dimensional smooth manifold. A function $f : \mathcal{M} \to \mathcal{N}$ is smooth if to each chart $(U, \phi)$ of $\mathcal{M}$ and each chart $(V, \varphi)$ of $\mathcal{N}$ such that $V \cap f(U) \neq \varnothing$, the function $\varphi \circ f \circ \phi^{-1}$, that maps a subset of $\mathbb{R}^{d_1}$ into $\mathbb{R}^{d_2}$, is smooth.*

**Definition S.7** (Diffeomorphism). *A bijective map between two smooth manifolds is a diffeomorphism if it and its inverse are smooth. Two smooth manifolds are diffeomorphic to each other if there exists a diffeomorphism between them.*

Let $\mathcal{M}$ be a $d$-dimensional smooth manifold in the sequel. A **vector field** $V$ on a set $A$ of $\mathcal{M}$ is a function that maps points in $A$ into $T\mathcal{M} := \cup_{p \in \mathcal{M}} T_p\mathcal{M}$, such that $V(p) \in T_p\mathcal{M}$; here $T\mathcal{M}$ is called the **tangent bundle** of $\mathcal{M}$. Intuitively, at each point in $U$, the vector field $V$ assigns the point with a tangent vector at that point. Recall that for each tangent vector $u \in T_p\mathcal{M}$ and chart $(U, \phi)$ with $p \in U$, $u$ has a coordinate representation $(\phi(p), x_1, \ldots, x_d)$, e.g., $(x_1, \ldots, x_d) = (\phi \circ \gamma)'(t)$ for some differentiable curve $\gamma$ with $\gamma(t) = p$. This gives rise to a local coordinate representation of a smooth vector field, i.e., for a chart $(U, \phi)$ such that $A \cap U \neq \varnothing$, the coordinate of $V$ at $u$ is the coordinate of the tangent vector $V(u) \in T_p\mathcal{M}$. Let $\Phi_\phi(u)$ denote the coordinate of the tangent vector $u$ in the chart $(U, \phi)$. Then $\Phi \circ V$, restricted to the

domain $A \cap U$, is a function mapping $U$ into $\mathbb{R}^{2d}$. We say $V$ is a **smooth vector field** if to each chart $(U, \phi)$, the function $\Phi \circ V$ is a smooth manifold map. An example of (smooth) vector fields is the vector field for the movement of air on Earth that represents the wind speed and direction at each location.

**Definition S.8** (Riemannian manifold). *A Riemannian manifold is a smooth manifold $\mathcal{M}$ endowed with an inner product $\langle \cdot, \cdot \rangle_p$ on $T_p\mathcal{M}$ for each $p \in \mathcal{M}$ such that, for any smooth vector fields $V_1, V_2$ on $\mathcal{M}$, the function $p \mapsto \langle V_1(p), V_2(p) \rangle_p$ is a smooth manifold map defined on $\mathcal{M}$. The inner products $\langle \cdot, \cdot \rangle_p$ are collectively referred to as Riemannian metric or Riemannian metric tensor.*

For a Riemannian manifold, each of its tangent spaces is now an inner product space, and thus along a normed space with the induced norm $\|v\|_p = \langle v, v \rangle_p$ for $v \in T_p\mathcal{M}$. As a vector space, each tangent space is entitled to an independent basis. For Riemannian manifold, since each tangent space is an inner product space, it is also entitled to orthonormal basis. In this context, a **frame** refers to a map that assigns each point of the manifold an independent basis, and when all of such basis are orthonormal, we say the frame is an **orthonormal frame**.

The Riemannian metric also induces a (canonical) **distance** $d_\mathcal{M}$ on the Riemannian manifold, as follows. For a smooth curve $\gamma : I \to \mathcal{M}$, the restriction to an interval $[a, b] \subset I$ is referred to as a segment of $\gamma$ and is denoted by $\gamma([a, b])$. The length of the segment $\gamma([a, b])$ is defined by

$$\ell_\gamma(a, b) = \int_a^b \|\gamma'(t)\|_{\gamma(t)} \mathrm{d}t.$$

Such definition can be extended to regular curves that are formed by connecting a finite number of smooth segments, i.e., the length of a piecewisely smooth segment is defined to be the sum of the lengths of its smooth segments. For a connected Riemannian manifold $\mathcal{M}$, we can define the distance between two points $p, q$ by

$$d_\mathcal{M}(p, q) = \inf\{\ell_\gamma(a, b) : \gamma(a) = p, \gamma(b) = q, \text{ and } \gamma \text{ is a regular curve}\}.$$

The distance $d_\mathcal{M}$ then turns $\mathcal{M}$ also into a metric space. If such metric space is complete, then we say $\mathcal{M}$ is a complete Riemannian manifold. By Hopf–Rinow theorem, on a connected complete Riemannian manifold, two points can be connected by a minimizing geodesic of which the length is exactly the distance of the two points. By definition, a **geodesic** is a curve $\gamma : I \to \mathcal{M}$ such that for all sufficiently small $\epsilon > 0$ and all interior $t \in I$, $\gamma([t, t + \epsilon])$ is a shortest path that connects $\gamma(t)$ and $\gamma(t + \epsilon)$, and $\ell_\gamma(s, t) = a|s - t|$ for a constant $a$ and all $s, t \in I$. In words, geodesics are constant-speed curves that are locally shortest segments. Note that a geodesic may not be globally a shortest path connecting two points. For example, $\gamma(t) = (0, \cos(2\pi t), \sin(2\pi t))$ for $t \in [0, 1]$ is a geodesic, but it is not a shortest path from $\gamma(\epsilon)$ to $\gamma(1)$ for any $\epsilon \in [0, 1/2)$. A minimizing geodesic refers to those geodesics $\gamma : I \to \mathcal{M}$ such that $\gamma([s, t])$ is a shortest path connecting $\gamma(s)$ and $\gamma(t)$ for all $s, t \in I$.

**Remark S.1.** *In some textbooks, the term "curve" (and similarly, "geodesic") is sometimes used to denote the image of $\gamma : I \to \mathcal{M}$ on the manifold $\mathcal{M}$, rather than the map $\gamma$. In this paper, we do not adopt this practice. When we say a curve, we always refers to the map $\gamma$ itself. One shall note that, the length of a curve is invariant to parameterization, i.e., if $\eta : J \to \mathcal{M}$ is another curve such that $t = g(s)$ for a smooth function with $g'(s) \neq 0$ and $\eta(s) = \gamma(g(s))$ for all $s \in J$ and $I = g(J)$, then $\ell_\eta(g(a), g(b)) = \ell_\gamma(a, b)$ for all $a, b \in J$. This is because, $\ell_\gamma(g(a), g(b)) = \int_{g(a)}^{g(b)} \|\gamma'(t)\|_{\gamma(t)} \mathrm{d}t = \int_a^b \|\gamma'(g(s))\|_{\gamma(g(s))} \mathrm{d}g(s) = \int_a^b \|\eta'(s)\|_{\eta(s)} \mathrm{d}s = \ell_\eta(a, b).$*

Below we assume $\mathcal{M}$ is complete and connected. The Riemannian metric also induces the Riemannian exponential map for each point. For a unit tangent vector $u \in T_p\mathcal{M}$, let $\gamma_u(t)$ be the geodesic such that $\gamma_u(0) = p$ and $\gamma'_u(0) = u$. As a starting point specified by $p$ and an initial direction specified by the unit

tangent vector $u$ together uniquely determine a geodesic, the celebrated Hopf–Rinow theorem asserts that the following map $\mathrm{Exp}_p$ is well defined at each $p \in \mathcal{M}$ on the entire tangent space $T_p\mathcal{M}$.

**Definition S.9** (Exponential map)**.** *The map* $\mathrm{Exp}_p(u) = \gamma_{\frac{u}{\|u\|_p}}(\|u\|_p)$ *for* $u \in T_p\mathcal{M}$ *is called the exponential map at* $p$.

**Remark S.2.** *Exponential maps might be defined also for incomplete Riemannian manifolds, but potentially only locally, i.e.,* $\mathrm{Exp}_p u$ *may be only defined for tangent vectors* $u \in T_p\mathcal{M}$ *in a neighborhood of the zero tangent vector* $0 \in T_p\mathcal{M}$.

For $p \in \mathcal{M}$ and $u \in T_p\mathcal{M}$, the curve $\gamma_u(t) = \mathrm{Exp}_p(tu)$ is a geodesic with speed $\|u\|_p$. The **cut time** $\mathfrak{c}(p, u)$ is defined to be $t \in \mathbb{R}_+$ such that $\gamma_u([0, t - \epsilon])$ is a minimizing geodesic for any $\epsilon > 0$ but $\gamma_u([0, t])$ is not, i.e., $\mathfrak{c}(p, u) = \sup\{t \in \mathbb{R}_+ : \gamma_u([0, t]) \text{ is a minimizing geodesic with } \gamma_u(t) = \mathrm{Exp}_p(tu)\}$. Let $\mathcal{E}_p = \{tu : u \in T_p\mathcal{M}, \|u\|_p = 1, 0 \le t < \mathfrak{c}(p, u)\}$, which is a neighborhood of the zero tangent vector in $T_p\mathcal{M}$, and define $\mathcal{D}_p = \{\mathrm{Exp}_p u : u \in \mathcal{E}_p\}$. Then $\mathrm{Exp}_p$ is bijective between $\mathcal{E}_p$ and $\mathcal{D}_p$ and thus its inverse exists on $\mathcal{D}_p$.

**Definition S.10** (Riemannian logarithmic map)**.** *The inverse of* $\mathrm{Exp}_p$ *on* $\mathcal{D}_p$, *denoted by* $\mathrm{Log}_p$, *is called the Riemannian logarithmic map at* $p$.

**Remark S.3.** *Generally* $\mathrm{Log}_p$ *is defined in a neighborhood at* $p$. *When* $\mathfrak{c}(p, u) = \infty$ *for all* $u$, *then* $\mathcal{D}_p = \mathcal{M}$ *for a complete and connected manifold* $\mathcal{M}$, *and in this case,* $\mathrm{Log}_p$ *is defined on the entire* $\mathcal{M}$.

The Riemannian logarithmic map, illustrated in the left panel of Figure S.2, has important applications in statistical analysis of manifold-valued data. For example, a $\mathcal{M}$-valued random process $X(t)$ with the mean function $\mu(t)$ may be converted by the logarithmic map into $\mathrm{Log}_{\mu(t)}X(t)$, a process taking values in $T_{\mu(t)}\mathcal{M}$ for each $t$, for further analysis. As $t$ varies, $\mathrm{Log}_{\mu(t)}X(t)$ moves across different tangent spaces. This requires us to consider the space formed by the union of all tangent spaces. Such space is a special case of vector bundle in which each point of the manifold is associated with a vector space.

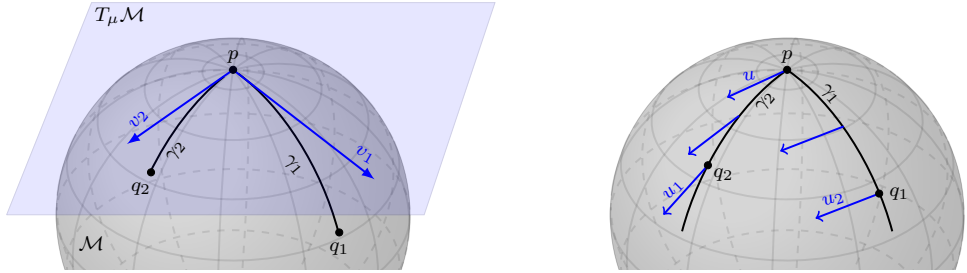

Figure S.2: Left: Illustration of the tangent space, geodesic, Riemannian exponential map and logarithmic map, where $\gamma_j(0) = p$, $\gamma_j(1) = q_j$, $q_j = \mathrm{Exp}_p v_j$ and $v_j = \mathrm{Log}_p q_j$, for $j = 1, 2$. Right: Illustration of parallel transport. The tangent vector $u$ at $p$ is parallelly transported along geodesics $\gamma_1$ and $\gamma_2$ to $q_1$ and $q_2$, resulting in tangent vectors $u_1 = \mathcal{P}_p^{q_1} u$ and $u_2 = \mathcal{P}_p^{q_2} u$, respectively.

**Definition S.11** (Vector bundle)**.** *A smooth vector bundle, denoted by* $\pi : \mathcal{E} \to \mathcal{M}$ *or simply* $\mathcal{E}$, *consists of a base smooth manifold* $\mathcal{M}$, *a smooth manifold* $\mathcal{E}$ *called total space, and a smooth bundle projection* $\pi$, *such that for every* $p \in \mathcal{M}$, *the fiber* $\pi^{-1}(p)$ *is a* $k$-*dimensional real vector space, and there is an open neighborhood* $U \subset \mathcal{M}$ *of* $p$ *and a diffeomorphism* $\Phi : \pi^{-1}(U) \to U \times \mathbb{R}^k$ *satisfying the property that for all* $z \in U$, $(\pi \circ \Phi^{-1})(z, v) = z$ *for all* $v \in \mathbb{R}^k$ *and the map* $v \mapsto \Phi^{-1}(z, v)$ *is a linear isomorphism between* $\mathbb{R}^k$ *and* $\pi^{-1}(z)$.

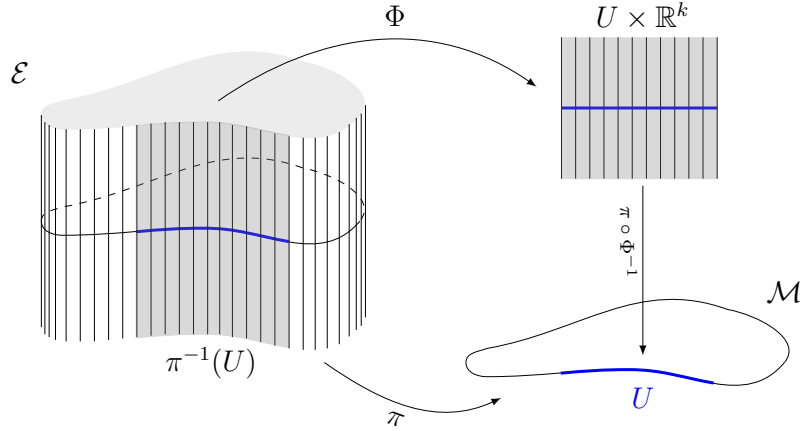

Figure S.3: Illustration of the vector bundle. The closed curve in the bottom represents the base manifold $\mathcal{M}$ and the figure on the left represents the total space $\mathcal{E}$, where each vertical line represents a fiber. The thickened segment $U \subset \mathcal{M}$ represents an open subset of the manifold $\mathcal{M}$, while $\Phi$ is a local trivialization defined on $\pi^{-1}(U)$ that is highlighted in gray in the total space.

The map $\Phi$ in the above definition is called a **local trivialization**. As graphically illustrated in Figure S.3, a vector bundle locally resembles the product space $U \times \mathbb{R}^k$ for some integer $k$. A function $V$ defined on $\mathcal{M}$ is called a **section** of the vector bundle if $V(p) \in \pi^{-1}(p)$ for all $p \in \mathcal{M}$. As previously mentioned, the union of all tangent spaces of a manifold, called the **tangent bundle** of the manifold, is a prominent example of vector bundle, where the tangent space at each point is a fiber. In particular, a section of a tangent bundle is also a vector field.

For a smooth function $f : \mathcal{M} \to \mathbb{R}$ and a tangent vector $v \in T_p\mathcal{M}$, the covariant derivative of $f$ at $p$ along the direction $v$, denoted by $\nabla_v f$, is defined by

$$(\nabla_v f)(p) := (f \circ \gamma)'(0) = \lim_{t \to 0} \frac{f(\gamma(t)) - f(p)}{t},$$

where $\gamma : [-1, 1] \to \mathcal{M}$ is a differentiable curve such that $\gamma(0) = p$ and $\gamma'(0) = v$. For a smooth vector field $U$, $\nabla_{V(p)} f$ is a real-valued function of $p$. Let $C^\infty(\mathcal{M})$ denote the collection of smooth real-valued functions defined on $\mathcal{M}$. In addition, for $f \in C^\infty(\mathcal{M})$ and a smooth vector field $U$, $fU$ denotes a smooth vector field defined by $(fU)(p) = f(p)U(p)$ for all $p \in \mathcal{M}$. Let $\Gamma(\mathcal{E})$ be the collection of smooth sections.

For a smooth curve in a Euclidean space, it is meaningful to discuss its acceleration which is represented by the second derivative of the curve. Note that the definition of second derivative involves differentiating the first derivative. To generalize the concept of acceleration to manifold-valued curves, we then need to differentiate the velocity — represented by tangent vectors — of the curve, and this involves the concept of connection.

**Definition S.12** (Connection). *A connection in a vector bundle $\mathcal{E}$ is a map $\nabla : \Gamma(\mathcal{E}) \times \Gamma(\mathcal{E}) \to \Gamma(\mathcal{E})$, with $(V, U) \mapsto \nabla_V U$, that satisfies the following properties:*

- *$\nabla_V U$ is linear over $C^\infty(\mathcal{M})$ in $V$, i.e., $\nabla_{fV_1 + gV_2} U = f\nabla_{V_1} U + g\nabla_{V_2} U$ for $f, g \in C^\infty(\mathcal{M})$ and $V_1, V_2 \in \Gamma(\mathcal{E})$;*

- *$\nabla_V U$ is linear over $\mathbb{R}$ in $U$, i.e., $\nabla_V(a_1 U_1 + a_2 U_2) = a_1\nabla_V U_1 + a_2\nabla_V U_2$ for $a_1, a_2 \in \mathbb{R}$ and $U_1, U_2 \in \Gamma(\mathcal{E})$;*

- $\nabla_V(fU) = f\nabla_V U + (\nabla_V f)U$ *for* $f \in C^\infty(\mathcal{M})$.

In the above, the value of $\nabla_V U$ at $p$ depends on $V$ only through its value at $p$ (Proposition 4.5, Lee, 2018). This observation leads to the definition of covariant derivative of a vector field at $p$ along $v \in T_p\mathcal{M}$. Consequently, the expression $\nabla_v U$ is sensible for $v \in T_p\mathcal{M}$, and is called the **covariant derivative** of $U$ at $p$ along the tangent vector $v$.

Let $f_{U,V}(p) = \langle U(p), V(p) \rangle_p$. For a connection $\nabla$ on the tangent bundle of $\mathcal{M}$, we say $\nabla$ is compatible with the metric on $\mathcal{M}$ if $\nabla_v f_{U,V} = \langle \nabla_v U, V \rangle_p + \langle U, \nabla_v V \rangle_p$ for all smooth vector fields $U$ and $V$, each $p \in \mathcal{M}$ and each tangent vector $v \in T_p\mathcal{M}$. For vector fields $U$ and $V$, we use $[U, V]$ to denote a new vector field such that $\nabla_{[U,V]}f = \nabla_U \nabla_V f - \nabla_V \nabla_U f$ for all $f \in C^\infty(\mathcal{M})$. Similarly, for $u, v \in T_p\mathcal{M}$, $[u, v]$ denotes the tangent at $p$ such that $\nabla_{[u,v]}f = \nabla_u \nabla_v f - \nabla_v \nabla_u f$ for all $f \in C^\infty(\mathcal{M})$. A connection is torsion-free if $\nabla_U V - \nabla_V U = [U, V]$ for all smooth vector fields $U, V$. For a Riemannian manifold, there exists a unique connection that is both torsion-free and compatible with the Riemannian metric. Such connection is called the **Levi–Civita** connection and deemed the canonical connection in the tangent bundle.

To identify different fibers, one can introduce a parallel transport $\mathscr{P}$ on a vector bundle along a curve $\gamma$ on the base manifold. Such parallel transport must satisfy the following axioms: 1) $\mathscr{P}_p^p$ is the identity map on $\pi^{-1}(p)$ for all $p \in \mathcal{M}$, 2) $\mathscr{P}_{\gamma(u)}^{\gamma(t)} \circ \mathscr{P}_{\gamma(s)}^{\gamma(u)} = \mathscr{P}_{\gamma(s)}^{\gamma(t)}$, and 3) the dependence of $\mathscr{P}$ on $\gamma$, $s$ and $t$ are smooth. An example is the vector bundle and the parallel transport constructed in Section 2.4. For a tangent bundle, such parallel transport can be induced by a connection.

**Definition S.13** (Parallel transport). *Let $\nabla$ be a connection in the tangent bundle of a Riemannian manifold $\mathcal{M}$. A smooth vector field $U$ is parallel along $\gamma : I \to \mathcal{M}$ (with respect to $\nabla$) if $\nabla_{\gamma'(t)}U = 0$ for all $t \in I$. The parallel transport of $v \in T_p\mathcal{M}$ along $\gamma$ with $p = \gamma(0)$ is $U(\gamma(t))$ for the unique smooth vector field $U$ along $\gamma$ such that $U$ is parallel along $\gamma$ and $U(0) = v$.*

Unlike Euclidean spaces, manifolds are often not flat and exhibit curvature that measures the degree of deviation from being flat. For smooth vector fields $U, V, W$, we define the map $R(U, V, W) = \nabla_U \nabla_V W - \nabla_V \nabla_U W - \nabla_{[U,V]}W$. It turns out that the value of $R(U, V, W)$ at $p$ depends only on the values of $U, V, W$ at $p$, and therefore it is sensible to write $R(u, v, w)$ for tangent vectors at the same point.

**Definition S.14** (Sectional curvature). *The sectional curvature at $p$ is a real-valued function on $T_p\mathcal{M} \times T_p\mathcal{M}$ defined for $u, v \in T_p\mathcal{M}$ by $\mathfrak{K}(u, v) = \langle R(u, v, v), u \rangle_p / (\langle u, u \rangle_p \langle v, v \rangle_p - \langle u, v \rangle_p^2)$.*

Note that sectional curvature is invariant to the length of tangent vectors $u$ and $v$. We say the sectional curvature of a Riemannian manifold $\mathcal{M}$ is upper (lower, resp.) bounded by $\kappa$ if $\mathfrak{K}(u, v) \leq \kappa$ ($\mathfrak{K}(u, v) \geq \kappa$, resp.) for all $p \in \mathcal{M}$ and $u, v \in T_p\mathcal{M}$.

## S.2    Asymptotic distribution of the covariance estimator

In this section, we provide a weak convergence result for the estimated covariance under the assumption $\hat{\mu} = \mu$ that is also adopted in Zhang and Wang (2016) for simplification. With the consistency property of $\hat{\mu}$, the derived asymptotic normality under the assumption $\hat{\mu} = \mu$ may approximate the reality when sample size is sufficiently large. To drop this assumption, a detailed analysis on the asymptotic normality of $\hat{\mu}$ seems needed. However, this turns out to be very challenging in the context of Riemannian data due to the curvature effect. Since the focus of this paper is the construction of the covariance vector bundle framework rather than the mean estimation by local linear smoothing, we decide to leave it for future study.

For $v_s \in T_{\mu(s)}\mathcal{M}$ and $v_t \in T_{\mu(t)}\mathcal{M}$, let $\gamma_{v_s,v_t} = \langle \mathcal{C}v_s, v_t \rangle_{\mu(t)}$ and $\hat{\gamma}_{v_s,v_t} = \langle \mathscr{P}_{(\hat{\mu}(s),\hat{\mu}(t))}^{(\mu(s),\mu(t))} \hat{\mathcal{C}}(s,t)v_s, v_t \rangle_{\mu(t)}$. It is seen that each pair $(v_s, v_t)$ defines a linear functional on $\mathbb{L}(\mu(s), \mu(t))$. Then $\mathscr{P}_{(\hat{\mu}(s),\hat{\mu}(t))}^{(\mu(s),\mu(t))} \hat{\mathcal{C}}(s,t)$ is weakly

convergent to $\mathcal{C}(s,t)$ if we can show that $\hat{\gamma}_{v_s,v_t}$ weakly converges to $\gamma_{v_s,v_t}$ for all $(v_s,v_t)$ according to the Cramér–Wold device. Below we demonstrate this for the random design; similar result can be proved for the hybrid design, and for deterministic design by utilizing the concept of design densities (Sacks and Ylvisaker, 1970).

To state the result, let $f$ be the probability density of $T_{11}$, and $\|K\|^2 = \int K^2(u)\mathrm{d}u$. Define

$$V_1(s,t,v_s,v_t) = \mathrm{Var}\{\langle \mathrm{Log}_{\mu(T_1)}X(T_1), v_s\rangle\langle \mathrm{Log}_{\mu(T_2)}X(T_2), v_t\rangle \mid T_1 = s, T_2 = t\},$$

$$V_2(s,t,v_s,v_t) = \mathrm{Cov}\{\langle \mathrm{Log}_{\mu(T_1)}X(T_1), v_s\rangle\langle \mathrm{Log}_{\mu(T_2)}X(T_2), v_t\rangle, \langle \mathrm{Log}_{\mu(T_1)}X(T_1), v_s\rangle\langle \mathrm{Log}_{\mu(T_3)}X(T_3), v_t\rangle \mid$$
$$T_1 = s, T_2 = t, T_3 = t\},$$

$$V_3(s,t,v_s,v_t) = \mathrm{Cov}\{\langle \mathrm{Log}_{\mu(T_1)}X(T_1), v_s\rangle\langle \mathrm{Log}_{\mu(T_2)}X(T_2), v_t\rangle, \langle \mathrm{Log}_{\mu(T_3)}X(T_3), v_s\rangle\langle \mathrm{Log}_{\mu(T_4)}X(T_4), v_t\rangle \mid$$
$$T_1 = s, T_2 = t, T_3 = s, T_4 = t\}.$$

**Theorem S.2.1.** *Suppose that Assumptions 2.1, 2.2, 3.1, 4.1, 4.4, 4.5 and 4.6 hold. In addition, assume $h_{\mathcal{C}} \to 0$, $nm^2 h_{\mathcal{C}}^2 \to \infty$, $m^3 h_{\mathcal{C}}^3 \ll n$, and $mh_{\mathcal{C}} \to c$ for some constant $c \in [0, \infty]$. Then for $s, t \in \mathcal{T}$ and $v_s \in T_{\mu(s)}\mathcal{M}$ and $v_t \in T_{\mu(t)}\mathcal{M}$,*

$$\Sigma_{\mathcal{C}}^{-1/2}\Big(\hat{\gamma}_{v_s,v_t} - \gamma_{v_s,v_t} - b(h_{\mathcal{C}}) + o_P(h_{\mathcal{C}}^2)\Big) \xrightarrow{D} N(0,1), \tag{S.1}$$

*where $b(h) = \frac{1}{2}h^2\Big(\int u^2 K(u)\mathrm{d}u\Big)\Big(\frac{\partial^2 \gamma}{\partial s^2}(s,t) + \frac{\partial^2 \gamma}{\partial t^2}(s,t)\Big)$ and*

$$\Sigma_{\mathcal{C}} = \{1 + 1_{s=t}\}\left\{\frac{\|K\|^4}{nm(m-1)h_{\mathcal{C}}^2}\frac{V_1(s,t,v_s,v_t)}{f(s)f(t)} + \frac{\|K\|^2}{n(m-1)h_{\mathcal{C}}}\frac{f(s)V_2(t,s,v_t,v_s)+f(t)V_2(s,t,v_s,v_t)}{f(s)f(t)}\right\} + \frac{(m-2)(m-3)V_3(s,t,v_s,v_t)}{nm(m-1)}.$$

The proof for the above theorem follows from Zhang and Wang (2016) once one realizes that $\hat{\gamma}_{v_s,v_t}$ is the estimated covariance based on the raw covariance $\langle \mathrm{Log}_{\mu(T_{ij})}X(T_{ij}), v_s\rangle\langle \mathrm{Log}_{\mu(T_{ik})}X(T_{ik}), v_t\rangle$. From the theorem we then observe the same phase transition as that in Zhang and Wang (2016) in the following corollary.

**Corollary S.2.1.** *Assume the conditions of Theorem S.2.1.*

(a) *When $m \gg n^{1/4}$, with $h_{\mathcal{C}} \ll n^{-1/4}$ and $mh_{\mathcal{C}} \to \infty$, one has*

$$\sqrt{n}\Big(\hat{\gamma}_{v_s,v_t} - \gamma_{v_s,v_t}\Big) \xrightarrow{D} N\big(0, V_1(s,t,v_s,v_t)\big).$$

(b) *When $m/n^{1/4} \to c_\circ$, with $h_{\mathcal{C}} = c_* n^{-1/4}$, one has*

$$\sqrt{n}\Big(\hat{\gamma}_{v_s,v_t} - \gamma_{v_s,v_t} - b(h)\Big) \xrightarrow{D} N(0, \Sigma_*)$$

*where*

$$\Sigma_* = \{1 + 1_{s=t}\}\left\{\frac{\|K\|^4}{c_\circ^2 c_*^2}\frac{V_1(s,t,v_s,v_t)}{f(s)f(t)} + \frac{\|K\|^2}{c_\circ c_*}\frac{f(s)V_2(t,s,v_t,v_s)+f(t)V_2(s,t,v_s,v_t)}{f(s)f(t)}\right\} + V_3(s,t,v_s,v_t).$$

(c) *When $m \ll n^{1/4}$, with $h_{\mathcal{C}} \asymp n^{-1/6}m^{-1/3}$, one has*

$$n^{1/3}m^{2/3}\Big(\hat{\gamma}_{v_s,v_t} - \gamma_{v_s,v_t} - b(h)\Big) \xrightarrow{D} N\Big(0, \{1 + 1_{s=t}\}\|K\|^4\frac{V_1(s,t,v_s,v_t)}{f(s)f(t)}\Big).$$

# S.3   Proofs of Main Results

*Proof of Lemma 2.1.* We prove this result for any fixed $t \in \mathcal{T}$. First, we show $\mathbf{E}\{\mathrm{Log}_{\mu(t)}X(t)\} = 0$. Suppose that $\gamma$ is a geodesic emanating from $\mu(t)$ with velocity $v \in T_{\mu(t)}\mathcal{M}$ and $\|v\|_{\mu(t)} = 1$. According to Proposition 2.10 of Oller and Corcuera (1995), one has

$$
\begin{aligned}
\frac{d}{ds}F(\gamma(s), t)\big|_{s=0} &= \mathbf{E}\{2d_{\mathcal{M}}(X(t), \mu(t))\cos\langle v, \mathrm{Log}_{\mu(t)}X(t)\rangle\} \\
&= 2\mathbf{E}\{\|\mathrm{Log}_{\mu(t)}X(t)\|\cos\langle v, \mathrm{Log}_{\mu(t)}X(t)\rangle\}.
\end{aligned}
$$

Since $F(p, t)$ reaches the minimum at $p = \mu(t)$, we have $\frac{d}{ds}F(\gamma(s), t)\big|_{s=0} = 0$ for any $v \in T_{\mu(t)}\mathcal{M}$. As $\|\mathrm{Log}_{\mu(t)}X(t)\|\cos\langle v, \mathrm{Log}_{\mu(t)}X(t)\rangle$ is the projection of $\mathrm{Log}_{\mu(t)}X(t)$ onto $v$, $\mathbf{E}\{\|\mathrm{Log}_{\mu(t)}X(t)\|\cos\langle v, \mathrm{Log}_{\mu(t)}X(t)\rangle\} = 0$ for all $v \in T_{\mu(t)}\mathcal{M}$ then implies $\mathbf{E}\{\mathrm{Log}_{\mu(t)}X(t)\} = 0$. According to the definition of $Y(t)$, it holds that $\mathbf{E}\{\mathrm{Log}_{\mu(t)}Y(t)\} = 0$. Similarly, it can be shown that the derivative of $F^*(\cdot, t) = \mathbf{E}\{d_{\mathcal{M}}^2(Y(t), \cdot)\}$ vanishes at $\mu(t)$. With Assumption 2.1, this implies that $\mu(t)$ is the unique minimum of $F^*(\cdot, t)$ and thus is the Fréchet mean of $Y(t)$. □

*Proof of Theorem 2.1.* The mean continuity and joint measurability ensure that $\mathrm{Log}_\mu X$ is a random element in $\mathscr{T}(\mu)$. According to the definition of $\mathcal{C}(s, t)$, for any $u, v \in \mathscr{T}(\mu)$,

$$
\begin{aligned}
\langle\!\langle \int_{\mathcal{T}} \mathcal{C}(s, \cdot)u(s)\mathrm{d}s, v\rangle\!\rangle_\mu &= \int_{\mathcal{T}}\int_{\mathcal{T}} \langle \mathcal{C}(s, t)u(s), v(t)\rangle_{\mu(t)}\mathrm{d}s\mathrm{d}t \\
&= \int_{\mathcal{T}}\int_{\mathcal{T}} \mathbf{E}\{\langle\mathrm{Log}_{\mu(s)}X(s), u(s)\rangle_{\mu(s)}\langle\mathrm{Log}_{\mu(t)}X(t), v(t)\rangle_{\mu(t)}\}\mathrm{d}s\mathrm{d}t \\
&= \mathbf{E}\left\{\int_{\mathcal{T}} \langle\mathrm{Log}_{\mu(s)}X(s), u(s)\rangle_{\mu(s)}\mathrm{d}s \int_{\mathcal{T}}\langle\mathrm{Log}_{\mu(t)}X(t), v(t)\rangle_{\mu(t)}\mathrm{d}t\right\} \\
&= \mathbf{E}(\langle\!\langle\mathrm{Log}_\mu X, u\rangle\!\rangle_\mu\langle\!\langle\mathrm{Log}_\mu X, v\rangle\!\rangle_\mu) = \langle\!\langle\mathbf{C}u, v\rangle\!\rangle_\mu,
\end{aligned}
$$

which implies that $(\mathbf{C}u)(t) = \int_{\mathcal{T}}\mathcal{C}(s, t)u(s)\mathrm{d}s$. □

*Proof of Theorem 2.2.* To see that $\{(\pi^{-1}(U_\alpha \times U_\beta), \varphi_{\alpha,\beta}) : (\alpha, \beta) \in J^2\}$ is a smooth atlas, it is sufficient to check the transition maps. Suppose that $(p, q, \sum_{j,k=1}^d v_{jk}B_{\alpha,j}(p) \otimes B_{\beta,k}(q)) \in \pi^{-1}(U_\alpha \times U_\beta) \cap \pi^{-1}(U_{\tilde{\alpha}} \times U_{\tilde{\beta}})$ is also represented by $(p, q, \sum_{j,k=1}^d \tilde{v}_{jk}B_{\tilde{\alpha},j}(p) \otimes B_{\tilde{\beta},k}(q))$. The transformation from the coefficient vector $v = (v_{11}, v_{12}, \ldots, v_{dd})$ to $\tilde{v} = (\tilde{v}_{11}, \tilde{v}_{12}, \ldots, \tilde{v}_{dd})$ is smooth, since $\tilde{v} = \{\mathcal{J}_{\tilde{\alpha}}^\top(p) \otimes \mathcal{J}_{\tilde{\beta}}^\top(q)\}v$ and $\mathcal{J}_{\tilde{\alpha}}^\top(p), \mathcal{J}_{\tilde{\beta}}^\top(q)$ and their Kronecker product $\mathcal{J}_{\tilde{\alpha}}^\top(p) \otimes \mathcal{J}_{\tilde{\beta}}^\top(q)$ are respectively smooth in $p$, $q$ and $(p, q)$, where $\mathcal{J}_\alpha(\cdot)$ denotes the Jacobian matrix that transforms the basis $\{B_{\alpha,1}(\cdot), \ldots, B_{\alpha,d}(\cdot)\}$ into $\{B_{\tilde{\alpha},1}(\cdot), \ldots, B_{\tilde{\alpha},d}(\cdot)\}$.

According to the vector bundle construction lemma (Lemma 5.5, Lee, 2002), it is sufficient to check that when $U := (U_\alpha \times U_\beta) \cap (U_{\tilde{\alpha}} \times U_{\tilde{\beta}}) \neq \varnothing$ for some indices $\alpha, \beta, \tilde{\alpha}, \tilde{\beta}$, the composite map $\Phi_{\alpha,\beta} \circ \Phi_{\tilde{\alpha},\tilde{\beta}}^{-1}$ from $U \times \mathbb{R}^{d^2}$ to itself has the form $\Phi_{\alpha,\beta} \circ \Phi_{\tilde{\alpha},\tilde{\beta}}^{-1} = (p, q, \mathcal{J}(p, q)v)$ for a smooth map $\mathcal{J} : U \to \mathrm{GL}(d^2, \mathbb{R})$, where $\mathrm{GL}(d^2, \mathbb{R})$ is the collection of invertible real $d^2 \times d^2$ matrices. From above discussion, we have $\mathcal{J}(p, q) = \mathcal{J}_\alpha^\top(p) \otimes \mathcal{J}_\beta^\top(q)$ is smooth in $(p, q)$. In addition, $\mathcal{J}(p, q) \in \mathrm{GL}(d^2, \mathbb{R})$ since both $\mathcal{J}_\alpha^\top(p)$ and $\mathcal{J}_\alpha^\top(q)$ are invertible and so is their Kronecker product. Note that the vector bundle construction lemma also asserts that any compatible atlas for $\mathcal{M}$ gives rise to the same smooth structure on $\mathbb{L}$. □

*Proof of Theorem 2.3.* One can show that the parallel transport defined in (5) is a genuine parallel transport satisfying the property of Definition A.54 of Rodrigues and Capelas de Oliveira (2007) on the vector bundle. Then the conclusion directly follows from Definitions A.55 and A.57 of Rodrigues and Capelas de Oliveira (2007) and the remarks right below them. □

*Proof of Theorem 2.4.* We first show that the definition (7) is invariant to the choice of orthonormal bases. To this end, fix an orthonormal basis in $T_q\mathcal{M}$, and suppose that $\{\tilde{e}_1, \ldots, \tilde{e}_d\}$ is another orthonormal basis in $T_p\mathcal{M}$ and is related to $\{e_1, \ldots, e_d\}$ by a $d \times d$ unitary matrix $\mathbf{O}$. Let $A_1$ and $A_2$ be the respective matrix representation of $L_1$ and $L_2$ under the basis $\{e_1, \ldots, e_d\}$. Then their matrix representation under the basis $\{\tilde{e}_1, \ldots, \tilde{e}_d\}$ is $\tilde{A}_1 = \mathbf{O}A_1$ and $\tilde{A}_2 = \mathbf{O}A_2$, respectively. The inner product $G_{p,q}(L_1, L_2)$ is then calculated by $\mathrm{tr}(\tilde{A}_1^\top \tilde{A}_2) = \mathrm{tr}(A_1^\top \mathbf{O}^\top \mathbf{O} A_2) = \mathrm{tr}(A_1^\top A_2)$, which shows that $G_{p,q}(L_1, L_2)$ is invariant to the choice of bases in $T_p\mathcal{M}$. Its invariance to the choice of bases in $T_q\mathcal{M}$ can be proved in a similar fashion.

The smoothness of $G$ can be established by an argument similar to the one leading to Theorem 2.2 in conjunction with smoothness of the trace of matrices. To see that the parallel transport (5) preserves the bundle metric and thus defines isometries among fibers of $\mathbb{L}$, i.e., for any $L_1, L_2 \in \mathbb{L}(p_1, q_1)$,

$$G_{(p_1,q_1)}(L_1, L_2) = G_{(p_2,q_2)}(\mathscr{P}^{(p_2,q_2)}_{(p_1,q_1)}L_1, \mathscr{P}^{(p_2,q_2)}_{(p_1,q_1)}L_2),$$

suppose that $\{e_1, \ldots, e_d\}$ is an orthogonal basis of $T_{p_1}\mathcal{M}$. Then $\{\mathcal{P}^{p_2}_{p_1}e_1, \ldots, \mathcal{P}^{p_2}_{p_1}e_d\}$ is an orthogonal basis of $T_{p_2}\mathcal{M}$. This further implies that

$$\begin{aligned}
G_{(p_2,q_2)}(\mathscr{P}^{(p_2,q_2)}_{(p_1,q_1)}L_1, \mathscr{P}^{(p_2,q_2)}_{(p_1,q_1)}L_2) &= \sum_{k=1}^{d} \langle (\mathscr{P}^{(p_2,q_2)}_{(p_1,q_1)}L_1)(\mathcal{P}^{p_2}_{p_1}e_k), (\mathscr{P}^{(p_2,q_2)}_{(p_1,q_1)}L_2)(\mathcal{P}^{p_2}_{p_1}e_k) \rangle_{q_2} \\
&= \sum_{k=1}^{d} \langle \mathcal{P}^{q_2}_{q_1}[L_1(e_k)], \mathcal{P}^{q_2}_{q_1}[L_2(e_k)] \rangle_{q_2} \\
&= \sum_{k=1}^{d} \langle L_1(e_k), L_2(e_k) \rangle_{q_1} = G_{(p_1,q_1)}(L_1, L_2),
\end{aligned}$$

which completes the proof. $\qquad\square$

*Proof of Proposition 3.1.* Suppose that $(\tilde{B}_{ij,1}, \ldots, \tilde{B}_{ij,d})$ is another orthonormal basis for $T_{\hat{\mu}(T_{ij})}\mathcal{M}$, and $\mathbf{O}_{ij}$ is the unitary matrix relating $(B_{ij,1}, \ldots, B_{ij,d})$ to $(\tilde{B}_{ij,1}, \ldots, \tilde{B}_{ij,d})$. Then the coefficient vectors $\tilde{z}_{ij}$ and $\tilde{g}_{k,ij}$ of $\mathrm{Log}_{\hat{\mu}(T_{ij})}Y_{ij}$ and $\hat{\psi}_k(T_{ij})$ under the basis $(\tilde{B}_{ij,1}, \ldots, \tilde{B}_{ij,d})$ are linked to $z_{ij}$ and $g_{k,ij}$ by $\tilde{z}_{ij} = \mathbf{O}_{ij}z_{ij}$ and $\tilde{g}_{k,ij} = \mathbf{O}_{ij}g_{k,ij}$, respectively. Similarly, $\tilde{C}_{i,jl}$ is linked to $C_{i,jl}$ by $\tilde{C}_{i,jl} = \mathbf{O}_{ij}C_{i,jl}\mathbf{O}_{il}^\top$. More concisely, if we put

$$\mathbf{O}_i = \begin{pmatrix} \mathbf{O}_{i1} & & & \\ & \mathbf{O}_{i2} & & \\ & & \ddots & \\ & & & \mathbf{O}_{im_i} \end{pmatrix},$$

then $\tilde{z}_i = \mathbf{O}_i z_i$, $\tilde{g}_{k,i} = \mathbf{O}_i g_{k,i}$ and $\tilde{\Sigma}_i = \mathbf{O}_i \Sigma_i \mathbf{O}_i^\top$, which are the counterpart of $z_i$, $g_{k,i}$ and $\Sigma_i$ under the bases $(\tilde{B}_{ij,1}, \ldots, \tilde{B}_{ij,d})$, respectively. Note that $\tilde{\Sigma}_i^{-1} = (\mathbf{O}_i \Sigma_i \mathbf{O}_i^\top)^{-1} = \mathbf{O}_i^{-\top} \Sigma_i^{-1} \mathbf{O}_i^{-1} = \mathbf{O}_i \Sigma_i^{-1} \mathbf{O}_i^\top$ since $\mathbf{O}_{ij}$ are unitary matrices and thus $\mathbf{O}_i^{-1} = \mathbf{O}_i^\top$. Now we see that $\tilde{g}_{k,i}^\top \tilde{\Sigma}_i^{-1} \tilde{z}_i = g_{k,i}^\top \mathbf{O}_i^\top \mathbf{O}_i \Sigma_i^{-1} \mathbf{O}_i^\top \mathbf{O}_i z_i = g_{k,i}^\top \Sigma_i^{-1} z_i$, which clearly implies that the scores $\hat{\xi}_{ik}$ calculated under the bases $(\tilde{B}_{ij,1}, \ldots, \tilde{B}_{ij,d})$ is identical to the one computed under the bases $(B_{ij,1}, \ldots, B_{ij,d})$. $\qquad\square$

*Proof of Lemma 4.1.* Notice that

$$\begin{aligned}
& \|\mathcal{P}^{p_1}_{q_1} \mathcal{P}^{q_1}_{q_2} \mathrm{Log}_{q_2}y - \mathcal{P}^{p_1}_{p_2} \mathrm{Log}_{p_2}y\|_{p_1} \\
=\, & \|\mathcal{P}^{p_1}_{p_2} \mathcal{P}^{p_2}_{q_1} \mathcal{P}^{q_1}_{q_2} \mathrm{Log}_{q_2}y - \mathcal{P}^{p_1}_{p_2} \mathrm{Log}_{p_2}y\|_{q_2} \\
\leq\, & \|\mathcal{P}^{q_2}_{p_2} \mathcal{P}^{p_2}_{p_1} \mathcal{P}^{p_1}_{q_1} \mathcal{P}^{q_1}_{q_2} \mathrm{Log}_{q_2}y - \mathrm{Log}_{q_2}y\|_{q_2} + \|\mathrm{Log}_{q_2}y - \mathcal{P}^{q_2}_{p_2} \mathrm{Log}_{p_2}y\|_{q_2},
\end{aligned}$$

where the equality follows from the fact that parallel transport preserves the inner product. Note that the operator $\mathcal{P}_{p_2}^{q_2}\mathcal{P}_{p_1}^{p_2}\mathcal{P}_{q_1}^{p_1}\mathcal{P}_{q_2}^{q_1}$ moves a tangent vector parallelly along a geodesic quadrilateral defined by the the points $p_1, p_2, q_1, q_2$. The holonomy theory (Eq (6), Nichols et al., 2016) and the compactness of $\mathcal{G}$ suggests that there exists a constant $c_1 > 0$ depending only on $\mathcal{G}$, such that for any $v \in T_{q_2}\mathcal{M}$ with $\|v\|_{q_2} \leq \operatorname{diam}(\mathcal{G})$,

$$\|\mathcal{P}_{p_2}^{q_2}\mathcal{P}_{p_1}^{p_2}\mathcal{P}_{q_2}^{p_1}v - v\|_{q_2} \leq c_1\|\operatorname{Log}_{q_2}p_2\|_{q_2} = c_1 d_{\mathcal{M}}(p_2, q_2),$$
$$\|\mathcal{P}_{p_1}^{q_2}\mathcal{P}_{q_1}^{p_1}\mathcal{P}_{q_2}^{q_1}v - v\|_{q_2} \leq c_1\|\operatorname{Log}_{q_1}p_1\|_{q_2} = c_1 d_{\mathcal{M}}(p_1, q_1),$$

which further imply that

$$\begin{aligned}
\|\mathcal{P}_{p_2}^{q_2}\mathcal{P}_{p_1}^{p_2}\mathcal{P}_{q_1}^{p_1}\mathcal{P}_{q_2}^{q_1}v - v\|_{q_2} &= \|(\mathcal{P}_{p_2}^{q_2}\mathcal{P}_{p_1}^{p_2}\mathcal{P}_{q_2}^{p_1})(\mathcal{P}_{p_1}^{q_2}\mathcal{P}_{q_1}^{p_1}\mathcal{P}_{q_2}^{q_1})v - v\|_{q_2} \\
&\leq \|(\mathcal{P}_{p_2}^{q_2}\mathcal{P}_{p_1}^{p_2}\mathcal{P}_{q_2}^{p_1})(\mathcal{P}_{p_1}^{q_2}\mathcal{P}_{q_1}^{p_1}\mathcal{P}_{q_2}^{q_1})v - (\mathcal{P}_{p_1}^{q_2}\mathcal{P}_{q_1}^{p_1}\mathcal{P}_{q_2}^{q_1})v\|_{q_2} + \|(\mathcal{P}_{p_1}^{q_2}\mathcal{P}_{q_1}^{p_1}\mathcal{P}_{q_2}^{q_1})v - v\|_{q_2} \\
&\leq c_1(d_{\mathcal{M}}(p_2, q_2) + d_{\mathcal{M}}(p_1, q_1)).
\end{aligned}$$

According to Theorem 3 in Pennec (2019), we have

$$\|\operatorname{Log}_{q_2}y - \mathcal{P}_{p_2}^{q_2}\operatorname{Log}_{p_2}y\|_{q_2} \leq c_2\|\operatorname{Log}_{q_2}p_2\|_{q_2} \leq c_2 d_{\mathcal{M}}(p_2, q_2)$$

for some constant $c_2 > 0$ depending only on $\mathcal{G}$. The proof is then completed by taking $c = c_1 + c_2$.  □

*Proof of Propositions 4.1 and 4.2.* Simple computation shows that

$$\begin{aligned}
&\big(\hat{Q}_n(y, \tau) - F^*(y, \tau)\big) \\
={}& \frac{\hat{u}_2(\tau)}{\hat{\sigma}_0^2(\tau)}\frac{1}{nm}\sum_{ij}K_{h_\mu}(T_{ij} - \tau)\Big(d_{\mathcal{M}}^2(Y_{ij}, y) - F^*(y, \tau)\Big) \\
&- \frac{\hat{u}_1(\tau)}{\hat{\sigma}_0^2(\tau)}\frac{1}{nm}\sum_{ij}K_{h_\mu}(T_{ij} - \tau)(T_{ij} - \tau)\Big(d_{\mathcal{M}}^2(Y_{ij}, y) - F^*(y, \tau)\Big) \\
={}& \frac{\hat{u}_2(\tau)}{\hat{\sigma}_0^2(\tau)}\frac{1}{nm}\sum_{ij}K_{h_\mu}(T_{ij} - \tau)\Big(d_{\mathcal{M}}^2(Y_{ij}, y) - F^*(y, \tau) - \partial_\tau F^*(y, \tau)(T_{ij} - \tau)\Big) \\
&- \frac{\hat{u}_1(\tau)}{\hat{\sigma}_0^2(\tau)}\frac{1}{nm}\sum_{ij}K_{h_\mu}(T_{ij} - \tau)(T_{ij} - \tau)\Big(d_{\mathcal{M}}^2(Y_{ij}, y) - F^*(y, \tau) - \partial_\tau F^*(y, \tau)(T_{ij} - \tau)\Big).
\end{aligned}$$

Below we focus on the fist term, noting that the second term can be analyzed in a similar way.

Define

$$U := \frac{1}{nm}\sum_{ij}K_{h_\mu}(T_{ij} - \tau)\Big(d_{\mathcal{M}}^2(Y_{ij}, y) - F^*(y, \tau) - \partial_\tau F^*(y, \tau)(T_{ij} - \tau)\Big).$$

Then, according to either Lemma S.4.1 or Lemma S.4.3, the rate of the first term depends on the rate of $U$. By Taylor expansion of $F^*(y, T_{ij})$ at $\tau$ and Assumption 4.5(a), we have

$$\begin{aligned}
\sup_{\tau \in B(t;h)}|\mathbf{E}U| &= \sup_{\tau \in B(t;h)}\left|\mathbf{E}\bigg(\frac{1}{nm}\sum_{ij}K_{h_\mu}(T_{ij} - \tau)\big(F^*(y, T_{ij}) - F^*(y, \tau) - \partial_\tau F^*(y, \tau)(T_{ij} - \tau)\big)\bigg)\right| \\
&= \sup_{\tau \in B(t;h)}\left|\mathbf{E}\bigg(\frac{1}{nm}\sum_{ij}K_{h_\mu}(T_{ij} - \tau) \times O(h_\mu^2)\bigg)\right| = O(h_\mu^2).
\end{aligned}$$

For the random and hybrid designs, define the envelop function

$$H := \frac{2\text{diam}(\mathcal{K})^2}{m} \sum_{j=1}^{m} \sup_{\tau \in B(t;h)} K_{h_\mu}(T_{1j} - \tau).$$

According to Lemma S.4.1(a), we have $\mathbf{E}(H^2) = O(1 + \frac{1}{mh_\mu})$ and thus

$$\sup_{\tau \in B(t;h_\mu)} |U - \mathbf{E}U| = O_p\left(\sqrt{\frac{1}{n} + \frac{1}{nmh_\mu}}\right)$$

according to Theorems 2.7.11 and 2.14.2 of van der Vaart and Wellner (1996). Lemma S.4.6 asserts that the last equation also holds for a deterministic design. With Lemma S.4.1 we deduce that

$$\sup_{\tau \in B(t;h)} \left|\hat{Q}_n(y,\tau) - F^*(y,\tau)\right| = O_p\left(h_\mu^2 + \sqrt{\frac{1}{n} + \frac{1}{nmh_\mu}}\right).$$

A similar argument leads to

$$\sup_{\substack{d_{\mathcal{M}}(y_1,y_2)<\delta \\ \tau \in B(t;h)}} \left|\left(\hat{Q}_n(y_1,\tau) - \hat{Q}_n(y_2,\tau)\right) - \left(F^*(y_1,\tau) - F^*(y_2,\tau)\right)\right| = O_p\left(\delta h_\mu^2 + \delta\sqrt{\frac{1}{n} + \frac{1}{nmh_\mu}}\right). \tag{S.2}$$

for any $y_1, y_2 \in \mathcal{K}$ and $\delta > 0$. Following from the argument in the proof of Lemma 2 in Petersen and Müller (2019), one can verify that for any $\kappa > 0$

$$\lim_{\delta \to 0} \limsup_{n \to \infty} \text{Pr}\left\{\sup_{d_{\mathcal{M}}(y_1,y_2)<\delta, \tau \in B(t;h)} \left|\left(\hat{Q}_n(y_1,\tau) - \hat{Q}_n(y_2,\tau)\right) - \left(F^*(y_1,\tau) - F^*(y_2,\tau)\right)\right| > \kappa\right\} = 0,$$

and further

$$\sup_{\tau \in B(t;h)} d_{\mathcal{M}}(\mu(\tau), \hat{\mu}(\tau)) = o_p(1) \tag{S.3}$$

given Assumption 4.5(b).

To derive the rate we apply (S.2) with $y_1 = y$ and $y_2 = \mu(\tau)$ to obtain

$$\sup_{\substack{d_{\mathcal{M}}(y,\mu(\tau))<\delta \\ \tau \in B(t;h)}} \left|\left(\hat{Q}_n(y,\tau) - \hat{Q}_n(\mu(\tau),\tau)\right) - \left(F^*(y,\tau) - F^*(\mu(\tau),\tau)\right)\right| = O_p\left(\delta h_\mu^2 + \delta\sqrt{\frac{1}{n} + \frac{1}{nmh_\mu}}\right). \tag{S.4}$$

By (S.3), the the event $\{d_{\mathcal{M}}(\hat{\mu}(\tau), \mu(\tau)) < \eta_1\}$ occurs with probability tending to one. On this event, according to Assumption 4.5(c), we have

$$F^*(\hat{\mu}(\tau), \tau) - F^*(\mu(\tau), \tau) - C_1 d_{\mathcal{M}}(\hat{\mu}(\tau), \mu(\tau))^2 \geq 0.$$

Since $\hat{\mu}(\tau)$ is the minimizer of $\hat{Q}_n(y, \tau)$, we have $\hat{Q}_n(\mu(\tau), \tau) - \hat{Q}_n(\hat{\mu}(\tau), \tau) \geq 0$ and the following inequality on the event $\{d_{\mathcal{M}}(\hat{\mu}(\tau), \mu(\tau)) < \eta_1\}$,

$$\left(F^*(\hat{\mu}(\tau), \tau) - F^*(\mu(\tau), \tau)\right) - \left(\hat{Q}_n(\hat{\mu}(\tau), \tau) - \hat{Q}_n(\mu(\tau), \tau)\right) \geq C_1 d_{\mathcal{M}}(\hat{\mu}(\tau), \mu(\tau))^2. \tag{S.5}$$

Let $a_n = h_\mu^2 + \sqrt{n^{-1} + (nmh_\mu)^{-1}}$. Below we fix an arbitrary $\epsilon > 0$, and find $M > 0$ accordingly to satisfy

$$\Pr\left\{ \sup_{\tau \in B(t;h)} d_{\mathcal{M}}(\hat{\mu}(\tau), \mu(\tau)) > M a_n \right\} \le \epsilon.$$

To this end, for $R > 0$ to be determined later, let

$$B_R(\delta) = \left\{ \sup_{\substack{d_{\mathcal{M}}(y,\mu(\tau)) < \delta \\ \tau \in B(t;h)}} \left| \left(\hat{Q}_n(y,\tau) - \hat{Q}_n(\mu(\tau),\tau)\right) - \left(F^*(y,\tau) - F^*(\mu(\tau),\tau)\right) \right| \le R\delta\left(h_\mu^2 + \sqrt{\frac{1}{n} + \frac{1}{nmh_\mu}}\right) \right\},$$

$$B_j = \{2^j M a_n \le \sup_{\tau \in B(t;h)} d_{\mathcal{M}}(\hat{\mu}(\tau), \mu(\tau)) \le 2^{j+1} M a_n\},$$

$$B_C = \{ \sup_{\tau \in B(t;h)} d_{\mathcal{M}}(\hat{\mu}(\tau), \mu(\tau)) > \frac{1}{2}\eta_1 \}.$$

Let $j_0 \ge 0$ be an integer satisfying $\frac{1}{2}\eta_1 < 2^{j_0+1} M a_n \le \eta_1$. We then have

$$\Pr\left\{ \sup_{\tau \in B(t;h)} d_{\mathcal{M}}(\hat{\mu}(\tau), \mu(\tau)) > M a_n \right\} \le \sum_{j=0}^{j_0} \Pr\{B_j \cap B_R(2\eta_1)\} + \Pr\{B_C \cap B_R(2\eta_1)\} + \Pr\{\Omega \backslash B_R(2\eta_1)\}$$

$$\le \sum_{j=0}^{j_0} \Pr\left( B_j \cap B_R(2\eta_1) \cap \left\{ \sup_{\tau \in B(t;h)} \left| \left(F^*(\hat{\mu}(\tau),\tau) - F^*(\mu(\tau),\tau)\right) - \left(\hat{Q}_n(\hat{\mu}(\tau),\tau) - \hat{Q}_n(\mu(\tau),\tau)\right) \right| \ge C_1 (2^j M a_n)^2 \right\} \right)$$

$$+ \Pr(B_C) + \Pr\{\Omega \backslash B_R(2\eta_1)\}$$

$$\le \sum_{j=0}^{j_0} \Pr\left( 1_{B_R(2^{j+1} M a_n)} \sup_{\substack{d_{\mathcal{M}}(y,\mu(\tau)) \le 2^{j+1} M a_n \\ \tau \in B(t;h)}} \left| \left(F^*(y,\tau) - F^*(\mu(\tau),\tau)\right) - \left(\hat{Q}_n(y,\tau) - \hat{Q}_n(\mu(\tau),\tau)\right) \right| > C_1 (2^j M a_n)^2 \right)$$

$$+ \Pr(B_C) + \Pr\{\Omega \backslash B_R(2\eta_1)\}$$

$$\le R \sum_{j=0}^{j_0} \frac{2^{j+1} M a_n^2}{C_1 (2^j M a_n)^2} + \Pr(B_C) + \Pr\{\Omega \backslash B_R(2\eta_1)\} \le \frac{4R}{C_1 M} + \Pr(B_C) + \Pr\{\Omega \backslash B_R(2\eta_1)\}.$$

Since $\lim_{n\to\infty} \Pr\{B_C\} = 0$ according to $\sup_{\tau \in B(t;h)} d_{\mathcal{M}}(\mu(\tau), \hat{\mu}(\tau)) = o_p(1)$ and $\lim_{R\to\infty} \liminf_{n\to\infty} \Pr\{\Omega \backslash B_R(2\eta_1)\} = 0$ according to (S.4), there exist $n_0$ and $R_0$ such that for any $n > n_0$,

$$\Pr\{\Omega \backslash B_{R_0}(2\eta_1)\} < \frac{\epsilon}{3} \quad \text{and} \quad \Pr\{B_C\} < \frac{\epsilon}{3}.$$

Taking $M = \frac{12 R_0}{\epsilon C_1}$, we have

$$\Pr\left\{ \sup_{\tau \in B(t;h)} d_{\mathcal{M}}(\hat{\mu}(\tau), \mu(\tau)) > M a_n \right\} \le \frac{4 R_0}{C_1 M} + \Pr(B_C) + \Pr\{\Omega \backslash B_{R_0}(2\eta_1)\} \le \epsilon.$$

Therefore,

$$\lim_{M\to\infty} \limsup_{n\to\infty} \Pr\left\{ \sup_{\tau \in B(t;h)} d_{\mathcal{M}}(\hat{\mu}(\tau), \mu(\tau)) > M a_n \right\} = 0.$$

This yields the desired convergence rate of Proposition 4.2.

To establish the other proposition, by using the techniques in the proof of Lemma 5 of Zhang and Wang (2016) for the random and hybrid design, and Lemma S.4.6 combined with the above argument for the

deterministic design, we can show that

$$\sup_{\tau \in \mathcal{T}, y \in \mathcal{K}} \left| \hat{Q}_n(y, \tau) - F^*(y, \tau) \right| = O_p\left( h_\mu^2 + \sqrt{\frac{1}{n} + \frac{1}{nmh_\mu}} \right) (\log n)^{1/2}.$$

Together with (S.5), this proves Proposition 4.1.                    □

*Proof of Theorems 4.1.* We divide the proof into three steps. In the first step, we identify two key terms that determine the convergence rate. In the second step and third step, we address the terms separately. Note that $v_1 = \cdots = v_n = 1/\{nm(m-1)\}$.

*Step 1.* Define $\gamma := \hat{\mu}$ to simplify notation in the sequel. Since the parallel transport $\mathscr{P}^{(\mu(s),\mu(t))}_{(\gamma(s),\gamma(t))}$ preserves the fiber metric according to Theorem 2.4,

$$\mathscr{P}^{(\mu(s),\mu(t))}_{(\gamma(s),\gamma(t))}\hat{\mathcal{C}}(s,t) = \arg_{\beta_0} \min_{\beta_0, \beta_1, \beta_2 \in \mathbb{L}(\mu(s),\mu(t))} \left\{ \sum_i \nu_i \sum_{j \neq k} \| \mathscr{P}^{(\mu(s),\mu(t))}_{(\gamma(s),\gamma(t))} \mathscr{P}^{(\gamma(s),\gamma(t))}_{(\gamma(T_{ij}),\gamma(T_{ik}))} \hat{\mathcal{C}}_{i,jk} \right.$$
$$\left. - \beta_0 - \beta_1(T_{ij}-s) - \beta_2(T_{ik}-t) \|^2_{G_{(\mu(s),\mu(t))}} K_{h_{\mathcal{C}}}(s-T_{ij}) K_{h_{\mathcal{C}}}(t-T_{ik}) \right\}.$$

Simple computation shows that

$$\mathscr{P}^{(\mu(s),\mu(t))}_{(\gamma(s),\gamma(t))}\hat{\mathcal{C}}(s,t) - \mathcal{C}(s,t)$$
$$= \frac{(S_{20}S_{02} - S_{11}^2)R_{00} - (S_{10}S_{02} - S_{01}S_{11})R_{10} + (S_{10}S_{11} - S_{01}S_{20})R_{01}}{(S_{20}S_{02} - S_{11}^2)S_{00} - (S_{10}S_{02} - S_{01}S_{11})S_{10} + (S_{10}S_{11} - S_{01}S_{20})S_{01}}, \tag{S.6}$$

where $R_{ab}$ and $S_{ab}$ are defined as

$$R_{ab} := R_{ab}(s,t) := \sum_i \nu_i \sum_{j \neq k} \varpi(T_{ij}, T_{ik}) \left(\frac{T_{ij}-s}{h_{\mathcal{C}}}\right)^a \left(\frac{T_{ik}-t}{h_{\mathcal{C}}}\right)^b \left\{ \mathscr{P}^{(\mu(s),\mu(t))}_{(\gamma(s),\gamma(t))} \mathscr{P}^{(\gamma(s),\gamma(t))}_{(\gamma(T_{ij}),\gamma(T_{ik}))} \hat{\mathcal{C}}_{i,jk} - \mathcal{C}(s,t) \right\},$$

$$S_{ab} := S_{ab}(s,t) := \sum_i \nu_i \sum_{j \neq k} \varpi(T_{ij}, T_{ik}) \left(\frac{T_{ij}-s}{h_{\mathcal{C}}}\right)^a \left(\frac{T_{ik}-t}{h_{\mathcal{C}}}\right)^b, \tag{S.7}$$

with $\varpi(s',t') = K_{h_{\mathcal{C}}}(s-s')K_{h_{\mathcal{C}}}(t-t')$ for $s', t' \in \mathcal{T}$ to further simplify notations.

Define $\tilde{\mathcal{C}}_{i,jk} := \mathscr{P}^{(\mu(s),\mu(t))}_{(\mu(T_{ij}),\mu(T_{ik}))} (\mathrm{Log}_{\mu(T_{ij})} Y_{ij} \otimes \mathrm{Log}_{\mu(T_{ik})} Y_{ik}) \in \mathbb{L}(\mu(s),\mu(t))$. We consider the decomposition

$$\mathscr{P}^{(\mu(s),\mu(t))}_{(\gamma(s),\gamma(t))} \mathscr{P}^{(\gamma(s),\gamma(t))}_{(\gamma(T_{ij}),\gamma(T_{ik}))} \hat{\mathcal{C}}_{i,jk} - \mathcal{C}(s,t)$$
$$= \left\{ \mathscr{P}^{(\mu(s),\mu(t))}_{(\gamma(s),\gamma(t))} \mathscr{P}^{(\gamma(s),\gamma(t))}_{(\gamma(T_{ij}),\gamma(T_{ik}))} \hat{\mathcal{C}}_{i,jk} - \tilde{\mathcal{C}}_{i,jk} \right\} + \left\{ \tilde{\mathcal{C}}_{i,jk} - \mathcal{C}(s,t) \right\}, \tag{S.8}$$

according to which we split $R_{ab}$ into $R_{ab} = R_{ab,1} + R_{ab,2}$ with

$$R_{ab,1} := \sum_i \nu_i \sum_{j \neq k} \varpi(T_{ij}, T_{ik}) \left(\frac{T_{ij}-s}{h_{\mathcal{C}}}\right)^a \left(\frac{T_{ik}-t}{h_{\mathcal{C}}}\right)^b \left\{ \mathscr{P}^{(\mu(s),\mu(t))}_{(\gamma(s),\gamma(t))} \mathscr{P}^{(\gamma(s),\gamma(t))}_{(\gamma(T_{ij}),\gamma(T_{ik}))} \hat{\mathcal{C}}_{i,jk} - \tilde{\mathcal{C}}_{i,jk} \right\},$$

$$R_{ab,2} := \sum_i \nu_i \sum_{j \neq k} \varpi(T_{ij}, T_{ik}) \left(\frac{T_{ij}-s}{h_{\mathcal{C}}}\right)^a \left(\frac{T_{ik}-t}{h_{\mathcal{C}}}\right)^b \left\{ \tilde{\mathcal{C}}_{i,jk} - \mathcal{C}(s,t) \right\}.$$

Combining (S.6) and (S.8), we deduce that

$$
\begin{aligned}
&\mathscr{P}_{(\gamma(s),\gamma(t))}^{(\mu(s),\mu(t))}\hat{\mathcal{C}}(s,t) - \mathcal{C}(s,t)\\
&= \frac{(S_{20}S_{02}-S_{11}^2)\big[R_{00,1}+R_{00,2}-\partial_s\mathcal{C}(s,t)h_{\mathcal{C}}S_{10}-\partial_t\mathcal{C}(s,t)h_{\mathcal{C}}S_{01}\big]}{(S_{20}S_{02}-S_{11}^2)S_{00}-(S_{10}S_{02}-S_{01}S_{11})S_{10}+(S_{10}S_{11}-S_{01}S_{20})S_{01}}\\
&\quad-\frac{(S_{10}S_{02}-S_{01}S_{11})\big[R_{10,1}+R_{10,2}-\partial_s\mathcal{C}(s,t)h_{\mathcal{C}}S_{20}-\partial_t\mathcal{C}(s,t)h_{\mathcal{C}}S_{11}\big]}{(S_{20}S_{02}-S_{11}^2)S_{00}-(S_{10}S_{02}-S_{01}S_{11})S_{10}+(S_{10}S_{11}-S_{01}S_{20})S_{01}}\\
&\quad+\frac{(S_{10}S_{11}-S_{01}S_{20})\big[R_{01,1}+R_{01,2}-\partial_s\mathcal{C}(s,t)h_{\mathcal{C}}S_{11}-\partial_t\mathcal{C}(s,t)h_{\mathcal{C}}S_{02}\big]}{(S_{20}S_{02}-S_{11}^2)S_{00}-(S_{10}S_{02}-S_{01}S_{11})S_{10}+(S_{10}S_{11}-S_{01}S_{20})S_{01}}.
\end{aligned}
\tag{S.9}
$$

In light of Lemmas S.4.2 and S.4.3, the convergence rate of (S.6) depends on $R_{ab,1}$ and $R_{ab,2}-\partial_s\mathcal{C}(s,t)h_{\mathcal{C}}S_{a+1,b}-\partial_t\mathcal{C}(s,t)h_{\mathcal{C}}S_{a,b+1}$.

*Step 2.* In this step we address $R_{ab,1}$. According to the definition of $\mathscr{P}$ in (5), the first part in Equation (S.8) is

$$
\big(\mathcal{P}_{\gamma(s)}^{\mu(s)}\mathcal{P}_{\gamma(T_{ij})}^{\gamma(s)}\mathrm{Log}_{\gamma(T_{ij})}Y_{ij}\big)\otimes\big(\mathcal{P}_{\gamma(t)}^{\mu(t)}\mathcal{P}_{\gamma(T_{ik})}^{\gamma(t)}\mathrm{Log}_{\gamma(T_{ik})}Y_{ik}\big)-\big(\mathcal{P}_{\mu(T_{ij})}^{\mu(s)}\mathrm{Log}_{\mu(T_{ij})}Y_{ij}\big)\otimes\big(\mathcal{P}_{\mu(T_{ik})}^{\mu(t)}\mathrm{Log}_{\mu(T_{ik})}Y_{ik}\big).
$$

Then according to Assumption 4.4(b) and Lemma 4.1, its rate is

$$
\Big\|\mathscr{P}_{(\gamma(s),\gamma(t))}^{(\mu(s),\mu(t))}\mathscr{P}_{(\gamma(T_{ij}),\gamma(T_{ik}))}^{(\gamma(s),\gamma(t))}\hat{\mathcal{C}}_{i,jk}-\tilde{\mathcal{C}}_{i,jk}\Big\|_G=O\Big(\sup_{\tau:|\tau-s|<h_{\mathcal{C}}\text{ or }|\tau-t|<h_{\mathcal{C}}}d_{\mathcal{M}}(\gamma(\tau),\mu(\tau))\Big).
$$

By Proposition 4.2, we conclude that

$$
R_{ab,1}=O_p\Big(h_\mu^2+\sqrt{\tfrac{1}{n}+\tfrac{1}{nmh_\mu}}\Big).
$$

*Step 3.* In this step, we first analyze the term $R_{00,2}-\partial_s\mathcal{C}(s,t)h_{\mathcal{C}}S_{10}-\partial_t\mathcal{C}(s,t)h_{\mathcal{C}}S_{01}$ in (S.9), which equals to

$$
U:=\sum_i\nu_i\sum_{j\neq k}\varpi(T_{ij},T_{ik})\big\{\tilde{\mathcal{C}}_{i,jk}-\mathcal{C}(s,t)-\partial_s\mathcal{C}(s,t)(T_{ij}-s)-\partial_t\mathcal{C}(s,t)(T_{ik}-t)\big\}.
$$

We start with bounding its mean. Let $\mathbb{T}=\{T_{ij}:i=1,\ldots,n,\ j=1,\ldots,m_i\}$ and observe that

$$
\mathbf{E}(\tilde{\mathcal{C}}_{i,jk}\mid\mathbb{T})=\mathscr{P}_{(\mu(T_{ij}),\mu(T_{ik}))}^{(\mu(s),\mu(t))}\mathcal{C}(T_{ij},T_{ik}).
$$

In addition, since $\mathcal{C}$ is twice differentiable and the parallel transport $\mathscr{P}$ is depicted by a partial differential equation, we have the following Taylor expansion at $(s,t)$,

$$
\mathscr{P}_{(\mu(T_{ij}),\mu(T_{ik}))}^{(\mu(s),\mu(t))}\mathcal{C}(T_{ij},T_{ik})=\mathcal{C}(s,t)+\partial_s\mathcal{C}(s,t)(T_{ij}-s)+\partial_t\mathcal{C}(s,t)(T_{ik}-t)+O(h_{\mathcal{C}}^2)
\tag{S.10}
$$

for all $T_{ij},T_{ik}$ such that $|T_{ij}-s|<h_{\mathcal{C}}$ and $|T_{ik}-t|<h_{\mathcal{C}}$, where $O(h_{\mathcal{C}}^2)$ is uniform over all $T_{ij}$ and $T_{ik}$ due to Assumption 4.6 and the compactness of $\mathcal{K}$. Then we further deduce that

$$
\begin{aligned}
&\mathbf{E}(U)\\
&=\mathbf{E}\Big\{\mathbf{E}\Big[\sum_i\nu_i\sum_{j\neq k}\varpi(T_{ij},T_{ik})\big\{\tilde{\mathcal{C}}_{i,jk}-\mathcal{C}(s,t)-\partial_s\mathcal{C}(s,t)(T_{ij}-s)-\partial_t\mathcal{C}(s,t)(T_{ik}-t)\big\}\Big|\mathbb{T}\Big]\Big\}\\
&=\mathbf{E}\Big\{\sum_i\nu_i\sum_{j\neq k}\varpi(T_{ij},T_{ik})\times O(h_{\mathcal{C}}^2)\Big\}=O(h_{\mathcal{C}}^2).
\end{aligned}
$$

For the random and hybrid designs, the i.i.d assumption on trajectories and Lemma S.4.2 imply that

$$\mathbf{E}\|U - \mathbf{E}U\|_G^2 \le \frac{1}{n^2} \sum_{i=1}^n \mathbf{E} \left\| \frac{1}{m(m-1)} \sum_{j \ne k} \varpi(T_{ij}, T_{ik}) \left\{ \hat{\mathcal{C}}_{i,jk} - \mathscr{P}^{(\mu(s),\mu(t))}_{(\mu(T_{ij}),\mu(T_{ik}))} \mathcal{C}(T_{ij}, T_{ik}) \right\} \right\|_G^2$$

$$\le \frac{4\text{diam}(\mathcal{M})^4}{n} \mathbf{E} \left| \frac{1}{m(m-1)} \sum_{j \ne k} \varpi(T_{ij}, T_{ik}) \right|^2 = O\left( \frac{1}{n} + \frac{1}{nm^2 h_{\mathcal{C}}^2} \right).$$

Lemma S.4.7 asserts that this also holds for the deterministic design. Combining this with $\mathbf{E}U = O(h_{\mathcal{C}}^2)$, we deduce that $\mathbf{E}\|U - \mathbf{E}U\|_G^2 = O(n^{-1} + n^{-1}m^{-2}h_{\mathcal{C}}^{-2})$, and with Markov inequality, further conclude that $R_{00,2} - \partial_s \mathcal{C}(s,t) h_{\mathcal{C}} S_{10} - \partial_t \mathcal{C}(s,t) h_{\mathcal{C}} S_{01} = O_p(h_{\mathcal{C}}^2 + n^{-1/2} + n^{-1/2}m^{-1}h_{\mathcal{C}}^{-1})$.

Similar arguments can show that the terms $R_{10,2} - \partial_s \mathcal{C}(s,t) h_{\mathcal{C}} S_{20} - \partial_s \mathcal{C}(s,t) h_{\mathcal{C}} S_{11}$ and $R_{01,2} - \partial_s \mathcal{C}(s,t) h_{\mathcal{C}} S_{11} - \partial_t \mathcal{C}(s,t) h_{\mathcal{C}} S_{02}$ in (S.9) are of the same order. The equation (10) is then obtained by inserting the results in Steps 2 and 3 into Step 1. $\square$

*Proof of Theorem 4.2.* Similar to the proof of Theorem 4.1, we only need to consider the uniform rate of the term $R_{00,1} + R_{00,2} - \partial_s \mathcal{C}(s,t) h_{\mathcal{C}} S_{10} - \partial_t \mathcal{C}(s,t) h_{\mathcal{C}} S_{01}$ in (S.9).

Due to boundedness of $\mathcal{K}$ and Lemma 4.1, we have

$$\sup_{s,t} \|\mathcal{P}^{(\mu(s),\mu(t))}_{(\gamma(s),\gamma(t))} \mathcal{P}^{(\gamma(s),\gamma(t))}_{(\gamma(T_{ij}),\gamma(T_{ik}))} \hat{\mathcal{C}}_{i,jk} - \bar{\mathcal{C}}_{i,jk}\|_G \le C \sup_\tau d_{\mathcal{M}}(\gamma(\tau), \mu(\tau)).$$

Therefore, according to Lemma S.4.2 and Proposition 4.1, we deduce that

$$\sup_{s,t} \|R_{00,1}\|_G \le c \left( \sup_{s,t} |S_{00}| \right) \left( \sup_\tau d_{\mathcal{M}}(\gamma(\tau), \mu(\tau)) \right) = O_p \left( h_\mu^2 + \sqrt{\frac{\log n}{nmh_\mu} + \frac{\log n}{n}} \right)$$

for a universal constant $c > 0$ depending only on $\mathcal{K}$.

The uniform convergence rates of $R_{00,2} - \partial_s \mathcal{C}(s,t) h_{\mathcal{C}} S_{10} - \partial_t \mathcal{C}(s,t) h_{\mathcal{C}} S_{01}$ and other similar terms are obtained by arguments similar to those in Theorem 5.2 of Zhang and Wang (2016), except that no truncation argument is needed due to Assumption 4.4(b), and moments of some random quantities are calculated by using the techniques in Lemma S.4.2 for the hybrid design and by using Lemma S.4.7 for the deterministic design. $\square$

## S.4   Technical Lemmas

The following lemma is used to establish the convergence rate of the mean estimator under the random or hybrid design.

**Lemma S.4.1** (mean, random). *Suppose that Assumptions 2.1, 2.2, 3.1, 4.4, 4.5. Under either Assumption 4.1 or Assumption 4.3, if $h_\mu \to 0$ and $nmh_\mu \to \infty$, then for any $t$ and $h = O(h_\mu)$, we have*

(a) $\mathbf{E} \left| \frac{1}{m} \sum_j \sup_{\tau \in B(t;h)} K_{h_\mu}(T_{1j} - \tau) \right|^2 = O\left( 1 + \frac{1}{mh_\mu} \right)$;

(b) $\sup_{\tau \in \mathcal{T}} |\hat{u}_k(\tau)| = O_p(h_\mu^k)$ *for* $k = 0, 1, 2$;

(c) $\inf_{\tau \in B(t;h)} |\hat{\sigma}_0^2(\tau)| \asymp h_\mu^2 (1 + o_P(1))$.

*Proof of Lemma S.4.1.* Under either Assumption 4.1 or Assumption 4.3, $T_{i1}, \ldots, T_{im}$ are identically distributed (since they are exchangeable), and we deduce that

$$\sup_{\tau \in B(t;h)} |\mathbf{E}\hat{u}_k| = \sup_{\tau \in B(t;h)} \left| \frac{1}{nm} \sum_{ij} \mathbf{E}[K_{h_\mu}(T_{ij} - \tau)(T_{ij} - \tau)^k] \right| = \sup_{\tau \in B(t;h)} |\mathbf{E}K_{h_\mu}(T_{11} - \tau)(T_{11} - \tau)^k| = O(h_\mu^k).$$

Define an envelop function

$$H_k := \frac{1}{m} \sum_j \sup_{\tau \in B(t;h)} \left| K_{h_\mu}(T_{1j} - \tau)(T_{1j} - \tau)^k \right|$$

for $\hat{u}_k$. Under Assumption 4.1, the second moment of $H_k$ is

$$\begin{aligned}
\mathbf{E}(H_k^2) &= \frac{1}{m^2} \sum_{j_1, j_2} \left\{ \mathbf{E} \sup_{\tau \in B(t;h)} \left| K_{h_\mu}(T_{1j_1} - \tau)(T_{1j_1} - \tau)^k \right| \times \sup_{\tau \in B(t;h)} \left| K_{h_\mu}(T_{1j_2} - \tau)(T_{1j_2} - \tau)^k \right| \right\} \\
&= \frac{1}{m} \mathbf{E} \sup_{\tau \in B(t;h)} |K_{h_\mu}(T_{1j_1} - \tau)(T_{1j_1} - \tau)^k|^2 \\
&\quad + \frac{m-1}{m} \mathbf{E}\left\{ \sup_{\tau \in B(t;h)} \left| K_{h_\mu}(T_{11} - \tau)(T_{11} - \tau)^k \right| \right\} \times \mathbf{E}\left\{ \sup_{\tau \in B(t;h)} \left| K_{h_\mu}(T_{12} - \tau)(T_{12} - \tau)^k \right| \right\} \\
&= O\left( h_\mu^{2k}(1 + \frac{1}{mh_\mu}) \right).
\end{aligned}$$

Under Assumption 4.3,

$$\begin{aligned}
\mathbf{E}(H_k^2) &= \frac{1}{m^2} \sum_{j_1, j_2} \left\{ \mathbf{E} \sup_{\tau \in B(t;h)} \left| K_{h_\mu}(T_{1j_1} - \tau)(T_{1j_1} - \tau)^k \right| \times \sup_{\tau \in B(t;h)} \left| K_{h_\mu}(T_{1j_2} - \tau)(T_{1j_2} - \tau)^k \right| \right\} \\
&= \frac{1}{m} \mathbf{E} \sup_{\tau \in B(t;h)} |K_{h_\mu}(T_{1j_1} - \tau)(T_{1j_1} - \tau)^k|^2 \\
&\quad + \frac{m-1}{m} \mathbf{E}\left\{ \sup_{\tau \in B(t;h)} \left| K_{h_\mu}(T_{11} - \tau)(T_{11} - \tau)^k \right| \times \sup_{\tau \in B(t;h)} \left| K_{h_\mu}(T_{12} - \tau)(T_{12} - \tau)^k \right| \right\} \\
&\leq O\left( \frac{h_\mu^{2k-1}}{m} \right) + \mathbf{E}\left[ \mathbf{E}\left\{ \sup_{\tau \in B(t;h)} \left| K_{h_\mu}(T_{11} - \tau)(T_{11} - \tau)^k \right| \mid S_{11}, S_{12} \right\} \right. \\
&\quad \left. \times \mathbf{E}\left\{ \sup_{\tau \in B(t;h)} \left| K_{h_\mu}(T_{12} - \tau)(T_{12} - \tau)^k \right| \mid S_{11}, S_{12} \right\} \right] \\
&\leq O\left( \frac{h_\mu^{2k-1}}{m} \right) + O(h_\mu^{2k-2}) \mathbf{E}\left[ \mathbf{E}\left\{ \mathbf{1}_{t - S_{11} - O(h_\mu) \leq \zeta_{11} \leq t - S_{11} + O(h_\mu)} \mid S_{11} \right\} \mathbf{E}\left\{ \mathbf{1}_{t - S_{12} - O(h_\mu) \leq \zeta_{12} \leq t - S_{12} + O(h_\mu)} \mid S_{12} \right\} \right].
\end{aligned}$$

When $h_\mu \lesssim L^{-1}$, $\mathbf{E}\left\{ \mathbf{1}_{t - S_{11} - O(h_\mu) \leq \zeta_{11} \leq t - S_{11} + O(h_\mu)} \mid S_{11} \right\}$ is of order $O(hL)$ when $|S_{11} - t| = O(L^{-1})$ and zero otherwise, an similar observation applies to $\mathbf{E}\left\{ \mathbf{1}_{t - S_{12} - O(h_\mu) \leq \zeta_{12} \leq t - S_{12} + O(h_\mu)} \mid S_{12} \right\}$. Together, they imply that

$$\begin{aligned}
&\mathbf{E}\left[ \mathbf{E}\left\{ \mathbf{1}_{t - S_{11} - O(h_\mu) \leq \zeta_{11} \leq t - S_{11} + O(h_\mu)} \mid S_{11} \right\} \mathbf{E}\left\{ \mathbf{1}_{t - S_{11} - O(h_\mu) \leq \zeta_{11} \leq t - S_{11} + O(h_\mu)} \mid S_{12} \right\} \right] \\
&= O(h_\mu^2 L^2) \mathbf{E}\left\{ \mathbf{1}_{|S_{11} - t| = O(L^{-1})} \mathbf{1}_{|S_{12} - t| = O(L^{-1})} \right\} = O(h_\mu^2 L^2) O(L^{-2}) = O(h_\mu^2).
\end{aligned}$$

When $h_\mu \gtrsim L^{-1}$, $\mathbf{E}\left\{ \mathbf{1}_{t - S_{11} - O(h_\mu) \leq \zeta_{11} \leq t - S_{11} + O(h_\mu)} \mid S_{11} \right\}$ is of order $O(1)$ when $|S_{11} - t| = O(h_\mu)$ and zero otherwise, an similar observation applies to $\mathbf{E}\left\{ \mathbf{1}_{t - S_{12} - O(h_\mu) \leq \zeta_{12} \leq t - S_{12} + O(h_\mu)} \mid S_{12} \right\}$. Together, they imply that

$$\begin{aligned}
&\mathbf{E}\left[ \mathbf{E}\left\{ \mathbf{1}_{t - S_{11} \leq \zeta_{11} \leq t - S_{11} + 2h} \mid S_{11} \right\} \mathbf{E}\left\{ \mathbf{1}_{t - S_{12} \leq \zeta_{12} \leq t - S_{12} + 2h} \mid S_{12} \right\} \right] \\
&= O(1) \mathbf{E}\left\{ \mathbf{1}_{|S_{11} - t| = O(h_\mu)} \mathbf{1}_{|S_{12} - t| = O(L^{-1})} \right\} = O(1) O(h_\mu^2) = O(h_\mu^2).
\end{aligned}$$

In summary, we still have $\mathbf{E}H_k^2 = O\left(h_\mu^{2k}(1 + \frac{1}{mh_\mu})\right)$ under Assumption 4.3. Part (a) is then verified by taking $k = 0$ in the above.

Part (b) can be proved by an argument analogous to the proof for Lemma 4 of Zhang and Wang (2016). For part (c), it is seen that $\hat{\sigma}_0(\tau) \asymp \{\mathbf{E}\hat{u}_0\mathbf{E}\hat{u}_2 - (\mathbf{E}\hat{u}_1)^2\}(1 + o_P(1))$, where the $o_P(1)$ component is uniform over $\tau$. Define $V := K_{h_\mu}(T_{11} - \tau)$ and $W := \mathbf{E}V \asymp 1$. Simple calculation shows that

$$\mathbf{E}\hat{u}_0\mathbf{E}\hat{u}_2 - (\mathbf{E}\hat{u}_1)^2 = W\mathbf{E}\big(V\big[(T_{11} - \tau) - W^{-1}\mathbf{E}\{V(T_{11} - \tau)\}\big]^2\big) \asymp h_\mu^2$$

uniformly over all $\tau \in \mathcal{T}$.          □

The following lemma is used to establish Theorems 4.1 and 4.2 under the random or hybrid design. Its proof is similar to that for Lemma S.4.1 and thus is omitted.

**Lemma S.4.2** (covariance, random). *Suppose that Assumptions 2.1, 2.2, 3.1, 4.4 and 4.5. Under either of additional Assumptions 4.1 and 4.3, if $h_{\mathcal{C}} \to 0$ and $nm^2h_{\mathcal{C}}^2 \to \infty$, we have*

$$\sup_{s,t\in\mathcal{T}} \mathbf{E}\{S_{ab}(s,t)\} = O(1),$$

$$\sup_{s,t\in\mathcal{T}} |S_{ab}(s,t) - \mathbf{E}\{S_{ab}(s,t)\}| = o_P(1),$$

$$\inf_{s,t\in\mathcal{T}}\{(S_{20}S_{02} - S_{11}^2)S_{00} - (S_{10}S_{02} - S_{01}S_{11})S_{10} + (S_{10}S_{11} - S_{01}S_{20})S_{01}\} \asymp 1 + o_P(1),$$

$$\sup_{s,t\in\mathcal{T}} \mathbf{E}\left|\frac{1}{m(m-1)}\sum_{j\neq k}K_{h_{\mathcal{C}}}(s - T_{ij})K_{h_{\mathcal{C}}}(t - T_{ik})\right|^2 = O\left(1 + \frac{1}{m^2h_{\mathcal{C}}^2}\right).$$

The next lemma is used to prove the convergence rates of the mean and covariance estimators under the deterministic design.

**Lemma S.4.3** (mean and covariance, deterministic). *Suppose that Assumptions 4.4(c)(d) and 4.2 hold, and $K$ is decreasing on $[0,1]$. If $nmh_\mu \to \infty$ and $h \asymp h_\mu$, then*

(a) $\sup_{\tau\in\mathcal{T}} |\hat{u}_k(\tau)| = O(h_\mu^k)$ *for $k = 0, 1, 2$;*

(b) $\sup_{\tau\in\mathcal{T}} |\hat{u}_k(\tau)| \asymp h_\mu^k$ *for $k = 0, 2$;*

(c) $\inf_{\tau\in\mathcal{T}} |\hat{\sigma}_0^2(\tau)| \asymp h_\mu^2$.

*If $nm^2h_{\mathcal{C}}^2 \to \infty$, then*

(d) $\sup_{s,t\in\mathcal{T}} S_{ab}(s,t) = O(1)$;

(e) $\inf_{s,t\in\mathcal{T}}\{(S_{20}S_{02} - S_{11}^2)S_{00} - (S_{10}S_{02} - S_{01}S_{11})S_{10} + (S_{10}S_{11} - S_{01}S_{20})S_{01}\} \asymp 1$;

(f) $\sup_{s,t\in\mathcal{T}} \sum_{i,j\neq k} \nu_i K_{h_{\mathcal{C}}}(s - T_{ij})K_{h_{\mathcal{C}}}(t - T_{ik})\left|\frac{T_{ij}-s}{h_{\mathcal{C}}}\right|^a \left|\frac{T_{ik}-t}{h_{\mathcal{C}}}\right|^b \asymp 1$,

*where $S_{ab}$ is defined in (S.7).*

*Proof.* Part (a) can be verified by simple calculation. For part (b), we fix $\tau \in \mathcal{T}$ and let $W = \sum_{ij} K\left(\frac{T_{ij}-\tau}{h_\mu}\right)$. The assumptions on the kernel function imply that $K(u) \geq c_0$ on $[-3/4, 3/4]$ for some constant $c_0 > 0$ depending only on $K$. In the sequel, we assume $n$ is sufficiently large so that $nmh_\mu \gg 1$. Assumption 4.2 implies that there are at least $c_1 nmh_\mu/2$ points within the interval $[\tau - 3h_\mu/4, \tau + 3h_\mu/4]$, from which we

deduce that $W \geq c_0 c_1 n m h_\mu / 2 > 0$ regardless of the location of $\tau$. Let $w_{ij} = K\left(\frac{T_{ij} - \tau}{h_\mu}\right)/W$, which is well defined. Observe that

$$\hat{u}_k(\tau) = \frac{W}{nmh_\mu} \sum_{ij} w_{ij} (T_{ij} - \tau)^k.$$

According the assumptions on the kernel function and Assumption 4.2, at least $c_1 nm h_\mu / 4$ of the pairs $(i, j)$ satisfy $|T_{ij} - \tau| \geq h_\mu / 8$ and $w_{ij} \geq c_0/W$. Thus,

$$\hat{u}_k(\tau) \geq \frac{W}{nmh_\mu} \frac{c_1 nm h_\mu}{4} \frac{c_0}{W} \frac{h_\mu^k}{8^k} \geq \frac{c_0 c_1}{2^{3k+2}} h_\mu^k$$

regardless of the value of $\tau$. Combining this with the first statement we prove the second statement. The last statement can be established in a similar fashion.

To establish part (c), let $E = \sum_{ij} w_{ij} T_{ij}$. As $w_{ij}$ is nonzero if and only if $T_{ij} \in (\tau - h_\mu, \tau + h_\mu)$ and $\sum_{ij} w_{ij} = 1$, we have $E \in [\tau - h_\mu, \tau + h_\mu]$. We then observe that

$$
\begin{aligned}
\hat{\sigma}_0^2(\tau) &= \frac{W^2}{(nmh_\mu)^2} \left\{ \sum_{ij} w_{ij} (T_{ij} - \tau)^2 \right\} - \frac{W^2}{(nmh_\mu)^2} \left\{ \sum_{ij} w_{ij} (T_{ij} - \tau) \right\}^2 \\
&= \frac{W^2}{(nmh_\mu)^2} \left[ \sum_{ij} w_{ij} \left\{ (T_{ij} - \tau) - \sum_{i'j'} w_{i'j'} (T_{i'j'} - \tau) \right\}^2 \right] \\
&= \frac{W^2}{(nmh_\mu)^2} \left\{ \sum_{ij} w_{ij} (T_{ij} - E)^2 \right\}.
\end{aligned}
$$

According to Assumption 4.2, there are at least $c_1 nm h_\mu / 4$ of $T_{ij}$ such that $|T_{ij} - E| \geq h_\mu / 8$ and $w_{ij} \geq c_0/W$. This implies that

$$\hat{\sigma}_0^2(\tau) \geq \frac{c_0 c_1 h_\mu^2}{256} \frac{W}{nmh_\mu} \geq \frac{c_0^2 c_1^2}{512} h_\mu^2$$

regardless of the value of $\tau$, where the last inequality is due to $W \geq c_0 c_1 nm h_\mu / 2$ that we have deduced previously.

The other statements can be established by similar arguments.  □

An $\epsilon$-cover of a subset $S$ of a pseudo-metric space $(\Omega, d)$ is a subset $A \subset S$ such that for each $p \in S$ there exists a $q \in A$ such that $d(p, q) \leq \epsilon$. We define $N(\epsilon, S, d) = \min\{|A| : A$ is an $\epsilon$ cover of $S\}$ to be the $\epsilon$-covering number of $S$, where $|A|$ denotes the cardinality of the set $A$. An $\epsilon$-packing of $S$ is a subset $A \subset S$ such that $d(p, q) > \epsilon$ for $p, q \in A$. The $\epsilon$-packing number of $S$ is defined by $M(\epsilon, S, d) = \max\{|A| : A$ is an $\epsilon$ packing of $S\}$. A standard relation between $\epsilon$-covering number and $\epsilon$-packing number is $M(2\epsilon, S, d) \leq N(\epsilon, S, d) \leq M(\epsilon, S, d)$ for all $\epsilon > 0$.

**Lemma S.4.4.** *Let $(S_1, d_1)$ and $(S_2, d_2)$ be two pseudo-metric spaces and $(S_1 \times S_2, d_1 \times d_2)$ the product pseudo-metric space with the pseudo-metric $(d_1 \times d_2)(p_1 \times p_2, q_1 \times q_2) = \{d_1^2(p_1, q_1) + d_2^2(p_2, q_2)\}^{1/2}$ for $p_1 \times p_2, q_1 \times q_2 \in S_1 \times S_2$. Then $N(\epsilon, S_1 \times S_2, d_1 \times d_2) \leq N(\epsilon/\sqrt{2}, S_1, d_1) N(\epsilon/\sqrt{2}, S_2, d_2)$.*

*Proof of Lemma S.4.4.* Let $A_1$ and $A_2$ be an $\epsilon/\sqrt{2}$-cover of $A_1$ and $A_2$, respectively. For each $k = 1, 2$, for every $p_k \in S_k$ there exists $p'_k \in A_k$ such that $d_k(p_k, p'_k) \leq \epsilon/\sqrt{2}$. Then for each $p_1 \times p_2 \in S_1 \times S_2$, we have $(d_1 \times d_2)(p_1 \times p_2, p'_1 \times p'_2) = \{d_1^2(p_1, p'_1) + d_2^2(p_2, p'_2)\}^{1/2} \leq \epsilon$. This shows that $A = \{p'_1 \times p'_2 : p'_1 \in A_1, p'_2 \in A_2\}$ is an $\epsilon$-cover. The conclusion of the lemma then follows from the observation $|A| = N(\epsilon/\sqrt{2}, S_1, d_1) N(\epsilon/\sqrt{2}, S_2, d_2)$.  □

**Lemma S.4.5.** *Let $d_h(y \times t, z \times s) := \{h^{-2}|s-t|^2 + d_{\mathcal{M}}^2(y,z)\}^{1/2}$ be a distance on the product space $\mathcal{M} \times \mathcal{T}$, and $c > 0$ a constant. Then we have $\sup_t \operatorname{diam}(\mathcal{K} \times B(t;ch)) \le \sqrt{4c^2 + \operatorname{diam}^2(\mathcal{K})}$. In addition, for all sufficiently small $\epsilon > 0$, $\sup_t N(\epsilon, \mathcal{K} \times B(t;ch), d_h \times d_{\mathcal{M}}) \le c_0 \epsilon^{-d-1}$, where $d$ is the dimension of $\mathcal{M}$ and $c_0$ is a constant depending on $c$ and $\mathcal{K}$.*

*Proof.* Given Lemma S.4.4, it is sufficient to show that $N(\epsilon, \mathcal{K}, d_{\mathcal{M}}) \le c_1 \epsilon^{-d}$ and $N(\epsilon, B(t;ch), d_h) \le c_2 \epsilon^{-1}$ for some constants $c_1, c_2 > 0$. The compactness implies that $\mathcal{K}$ has bounded sectional curvature. Bishop–Günther inequality (Gray, 2012) implies that $M(\epsilon, \mathcal{K}, d_{\mathcal{M}}) \le c_1 \epsilon^{-d}$. Note that the space $(B(t;ch), d_B)$, with $d_B(s_1, s_2) = h^{-1}|s_1 - s_2|$ for all $s_1, s_2 \in B(t;ch)$, is isometric to the interval $[-c,c]$ endowed with the standard distance $d_E(s_1, s_2) = |s_1 - s_2|$, which implies that $N(\epsilon, B(t;ch), d_B) = N(\epsilon, [-c,c], d_E) \le c_2 \epsilon^{-1}$. $\qquad\square$

The following lemma is used to establish the convergence rate of the mean estimator under the deterministic design.

**Lemma S.4.6** (mean, deterministic)**.** *Suppose that Assumptions 2.1, 2.2, 3.1, 4.4, 4.5 and 4.2 hold. Let*

$$U(y,\tau) = \frac{1}{nm} \sum_{ij} K_{h_\mu}(T_{ij} - \tau) \left( d_{\mathcal{M}}^2(Y_{ij}, y) - F^*(y,\tau) - \partial_\tau F^*(y,\tau)(T_{ij} - \tau) \right).$$

*Then, if $h_\mu \to 0$ and $nmh_\mu \to \infty$, then for any deterministic $t \in \mathcal{T}$ and $h = O(h_\mu)$, for all sufficiently small $h_\mu$,*

$$\mathbf{E} \left\{ \sup_{\substack{y \in \mathcal{K} \\ \tau \in B(t;h)}} \left| U(y,\tau) - \mathbf{E} U(y,\tau) \right| \right\} = O\left( n^{-1/2} + (nmh_\mu)^{-1/2} \right), \tag{S.11}$$

$$\mathbf{E} \left\{ \sup_{\substack{d(y_1,y_2)<\delta \\ \tau \in B(t;h)}} \left| \{U(y_1,\tau) - \mathbf{E} U(y_1,\tau)\} - \{U(y_2,\tau) - \mathbf{E} U(y_2,\tau)\} \right| \right\} = O\left( \delta n^{-1/2} + \delta (nmh_\mu)^{-1/2} \right), \tag{S.12}$$

*where $\delta > 0$ is a constant. In addition, if $h_\mu \to 0$, $nh_\mu \gtrsim 1$ and $nmh_\mu / \log n \to \infty$, then*

$$\mathbf{E} \left\{ \sup_{\substack{y \in \mathcal{K} \\ \tau \in \mathcal{T}}} \left| U(y,\tau) - \mathbf{E} U(y,\tau) \right| \right\} = O\left( n^{-1/2} + (nmh_\mu)^{-1/2} \right) (\log n)^{1/2}, \tag{S.13}$$

$$\mathbf{E} \left\{ \sup_{\substack{d(y_1,y_2)<\delta \\ \tau \in \mathcal{T}}} \left| \{U(y_1,\tau) - \mathbf{E} U(y_1,\tau)\} - \{U(y_2,\tau) - \mathbf{E} U(y_2,\tau)\} \right| \right\} = O\left( \delta n^{-1/2} + \delta (nmh_\mu)^{-1/2} \right) (\log n)^{1/2}. \tag{S.14}$$

*Proof.* To simplify notation, the symbol $c$ below, which denotes a constant not depending on $n, m, h_\mu, \tau, t$ but maybe depending on other constants such as $\operatorname{diam}(\mathcal{K})$, $\sup_{u \in [-1,1]} K(u)$, Lipschitz constant of $K$, etc, will often be re-used potentially with different values at each occurrence. Below we prove (S.11) and (S.13); the proofs for (S.12) and (S.14) are similar and thus omitted.

We first consider the case $mh_\mu \gtrsim 1$. Let

$$V_i(y,\tau) = \frac{1}{mh_\mu} \sum_{j=1}^{m} K\left( \frac{T_{ij} - \tau}{h_\mu} \right) \left( d_{\mathcal{M}}^2(Y_{ij}, y) - F^*(y, T_{ij}) \right) \tag{S.15}$$

and

$$Z_n(y, \tau) = \frac{1}{\sqrt{n}} \sum_{i=1}^{n} V_i(y, \tau). \tag{S.16}$$

Then $\mathbf{E} V_i(y, \tau) = 0$ and $U(y, \tau) - \mathbf{E} U(y, \tau) = n^{-1/2} Z_n(y, \tau)$. Now we observe that

$$
\begin{aligned}
|V_i(y, \tau_1) - V_i(z, \tau_2)| &\leq \frac{1}{mh_\mu} \left| \sum_{j=1}^{m} \left\{ K\left(\frac{T_{ij} - \tau_1}{h_\mu}\right) - K\left(\frac{T_{ij} - \tau_2}{h_\mu}\right) \right\} \left( d_{\mathcal{M}}^2(Y_{ij}, y) - F^*(y, T_{ij}) \right) \right| \\
&\quad + \frac{1}{mh_\mu} \left| \sum_{j=1}^{m} K\left(\frac{T_{ij} - \tau_2}{h_\mu}\right) \left( d_{\mathcal{M}}^2(Y_{ij}, y) - F^*(y, T_{ij}) - d_{\mathcal{M}}^2(Y_{ij}, z) + F^*(z, T_{ij}) \right) \right| \\
&\leq \frac{c}{mh_\mu} \left( \frac{|\tau_2 - \tau_1|}{h_\mu} + d(y, z) \right) \sum_{j=1}^{m} (1_{\tau_1 - h_\mu \leq T_{ij} \leq \tau_1 + h_\mu} + 1_{\tau_2 - h_\mu \leq T_{ij} \leq \tau_2 + h_\mu}) \\
&\leq c \frac{\max(c_2 mh_\mu, 1)}{mh_\mu} d_h(y \times \tau_1, z \times \tau_2) \\
&\leq c d_h(y \times \tau_1, z \times \tau_2)
\end{aligned}
$$

where $d_h(y \times \tau_1, z \times \tau_2) := \{h_\mu^{-2} |\tau_2 - \tau_1|^2 + d_{\mathcal{M}}^2(y, z)\}^{1/2}$ defines a distance on the product space $\mathcal{K} \times \mathcal{T}$. With the entropy bound in Lemma S.4.5, by Theorem 3.3 of van de Geer (1990) we deduce that

$$\Pr\left\{ \sup_{y \in \mathcal{K}, \tau \in B(t; h)} |Z_n(y, \tau)| \geq x \right\} \leq \exp(-cx^2), \tag{S.17}$$

which directly implies that $\mathbf{E}\{ \sup_{y \in \mathcal{K}, \tau \in B(t; h)} |Z_n(y, \tau)| \} = O(1)$ and further $\mathbf{E}\{ \sup_{y \in \mathcal{K}, \tau \in B(t; h)} |U(y, \tau) - \mathbf{E} U(y, \tau)| \} = O(n^{-1/2})$.

Next we consider the case $mh_\mu \to 0$. Let

$$V_i(y, \tau) = \sum_{j=1}^{m} K\left(\frac{T_{ij} - \tau}{h_\mu}\right) \left( d_{\mathcal{M}}^2(Y_{ij}, y) - F^*(y, T_{ij}) \right) \tag{S.18}$$

and

$$Z_n(y, \tau) = \frac{1}{\sqrt{nmh_\mu}} \sum_{i=1}^{n} V_i(y, \tau). \tag{S.19}$$

Then $\mathbf{E} V_i(y, \tau) = 0$ and $U(y, \tau) - \mathbf{E} U(y, \tau) = (nmh_\mu)^{-1/2} Z_n(y, \tau)$. Observe that

$$
\begin{aligned}
|V_i(y, \tau_1) - V_i(z, \tau_2)| &\leq c \left( \frac{|\tau_2 - \tau_1|}{h_\mu} + d(y, z) \right) \sum_{j=1}^{m} (1_{\tau_1 - h_\mu \leq T_{ij} \leq \tau_1 + h_\mu} + 1_{\tau_2 - h_\mu \leq T_{ij} \leq \tau_2 + h_\mu}) \\
&\leq c d_h(y \times \tau_1, z \times \tau_2),
\end{aligned}
$$

where we use the fact that $\sum_{j=1}^{m} (1_{\tau_1 - h_\mu \leq T_{ij} \leq \tau_1 + h_\mu} + 1_{\tau_2 - h_\mu \leq T_{ij} \leq \tau_2 + h_\mu}) \leq c$ due to the assumption $mh_\mu \to 0$ and Assumption 4.2. Note that for all sufficiently small $h_\mu$, there is at most one non-zero item in (S.18) and thus $Z_n(y, \tau)$ in (S.19) is sum of independent random variables. In addition, there are only at most $cnmh_\mu$ non-zero terms in (S.19). Based on Theorem 3.3 of van de Geer (1990) again we see that (S.17) holds, which implies that $\mathbf{E}\{ \sup_{y \in \mathcal{K}, \tau \in B(t; h)} |Z_n(y, \tau)| \} = O(1)$ and further $\mathbf{E}\{ \sup_{y \in \mathcal{K}, \tau \in B(t; h)} |U(y, \tau) - \mathbf{E} U(y, \tau)| \} = O\left(\frac{1}{\sqrt{nmh_\mu}}\right)$.

To establish (S.13), let $R = \lceil h_\mu^{-1} |\mathcal{T}| \rceil = O(h_\mu^{-1})$ and $A_1, \ldots, A_R$ a partition of $\mathcal{T}$ with $|A_r| \leq h_\mu$. According

to (S.17), we observe that, in either case of $mh_\mu \gtrsim 1$ and $mh_\mu \to 0$,

$$
\begin{aligned}
\Pr\left\{\sup_{y\in\mathcal{K},\tau\in\mathcal{T}}|Z_n(y,\tau)|\ge x\sqrt{\log n}\right\} &\le \sum_{r=1}^{R}\Pr\left\{\sup_{y\in\mathcal{K},\tau\in A_r}|Z_n(y,\tau)|\ge x\sqrt{\log n}\right\}\\
&= O(h_\mu^{-1})\exp(-cx\log n)\le O(n^{-1}h_\mu^{-1})n^{1-x}\\
&= O(1)n^{1-x},
\end{aligned}
$$

which then implies (S.13). $\qquad\square$

The following lemma is used to establish Theorems 4.1 and 4.2 under the deterministic design. Its proof is similar to that of Lemma S.4.6 and thus is omitted.

**Lemma S.4.7** (covariance, deterministic). *Suppose that Assumptions 2.1, 2.2, 3.1, 4.4, 4.5, 4.6 and 4.2 hold. Let*

$$
U(s,t) \coloneqq \sum_i \nu_i \sum_{j\ne k} \varpi(T_{ij},T_{ik})\left\{\tilde{\mathcal{C}}_{i,jk} - \mathcal{C}(s,t) - \partial_s\mathcal{C}(s,t)(T_{ij}-s) - \partial_t\mathcal{C}(s,t)(T_{ik}-t)\right\},
$$

*where $\varpi(s',t') = K_{h_\mathcal{C}}(s-s')K_{h_\mathcal{C}}(t-t')$ for $s',t'\in\mathcal{T}$. If $h_\mathcal{C}\to 0$ and $nm^2h_\mathcal{C}^2\to\infty$, then for all sufficient small $h_\mathcal{C}$,*

$$
\sup_{s,t\in\mathcal{T}}\mathbf{E}\left\{|U(s,t)-\mathbf{E}U(s,t)|\right\} = O\left(n^{-1/2}+(nm^2h_\mathcal{C}^2)^{-1/2}\right). \tag{S.20}
$$

*If $h_\mathcal{C}\to 0$, $nh_\mathcal{C}^2\gtrsim 1$ and $nm^2h_\mathcal{C}^2/\log n\to\infty$, then*

$$
\mathbf{E}\left\{\sup_{s,t\in\mathcal{T}}|U(s,t)-\mathbf{E}U(s,t)|\right\} = O\left(n^{-1/2}+(nm^2h_\mathcal{C}^2)^{-1/2}\right)(\log n)^{1/2}. \tag{S.21}
$$

## S.5  Theoretical Results for Regular Design

In a regular design, each sample path is observed on a common set of time points $\{T_j\}_{1\le j\le m}$. From a theoretical perspective, this design is fundamentally different from the designs discussed in Section 4, as under such design, $m\to\infty$ is required for the estimators $\hat\mu$ and $\hat{\mathcal{C}}$ to be consistent. Below we consider both random and deterministic regular design which includes the often-encountered equally-spaced design as a special case.

**Assumption S.5.1** (Regular Random Design). *The design points $\{T_j\}_{1\le j\le m}$, independent of other random quantities, are i.i.d. sampled from a distribution on $\mathcal{T}$ with a probability density that is bounded away from zero and infinity.*

**Assumption S.5.2** (Regular Deterministic Design). *The design points $\{T_j\}_{1\le j\le m}$ are nonrandom, and there exist constants $c_2\ge c_1>0$, such that for any intervals $A,B\subset\mathcal{T}$,*

(a) $c_1m|A|-1 \le \sum_{j=1}^{m}1_{T_j\in A} \le \max\{c_2m|A|,1\}$,

(b) $c_1m^2|A||B|-1 \le \sum_{j,k}1_{T_j\in A}1_{T_k\in B} \le \max\{c_2m^2|A||B|,1\}$,

*where $|A|$ denotes the length of $A$.*

For any fixed $t\in\mathcal{T}$, under either of Assumptions S.5.1 or S.5.2, the number of distinct observed time points in the interval of length $h_\mu$ is $O(mh_\mu)$, and thus the condition of $mh_\mu\gtrsim 1$ is necessary for consistency

of the mean and covariance estimators. Proposition S.5.1 presents the local and global uniform convergence rates for the mean estimation under regular design. The optimal bandwidth $h_\mu \asymp \frac{1}{m}$ leads to the same convergence rate as that from Cai and Yuan (2011) in the Euclidean case.

**Proposition S.5.1.** *Suppose that Assumptions 2.1, 2.2, 3.1, 4.4 and 4.5 hold. Under either of Assumptions S.5.1 or S.5.2, if $h_\mu \to 0$ and $mh_\mu > 1/c_1$, then*

$$\sup_{t \in \mathcal{T}} d_{\mathcal{M}}^2(\mu(t), \hat{\mu}(t)) = O_p\left(h_\mu^4 + \frac{\log n}{n}\right),$$

*and for any fixed $t \in \mathcal{T}$ and $h = O(h_\mu)$,*

$$\sup_{\tau : |\tau - t| \le h} d_{\mathcal{M}}^2(\mu(\tau), \hat{\mu}(\tau)) = O_p\left(h_\mu^4 + \frac{1}{n}\right).$$

In the above, the condition $mh_\mu > 1/c_1$ ensures that there is at least one observation of the time point for the interval $[t - h_\mu, t + h_\mu]$ for each $t$, according to Assumption S.5.2 for the regular deterministic design. The proof of Proposition S.5.1 is similar to that of Propositions 4.1 and 4.2, where Lemma S.4.1 is replaced with Lemma S.5.1 below for the regular random design and Lemmas S.4.3 and S.4.6 are replaced with Lemma S.5.2 for the regular deterministic design.

**Lemma S.5.1** (mean, regular random). *Suppose that Assumptions 2.1, 2.2, 3.1, 4.4, 4.5. Define*

$$U := \frac{1}{nm} \sum_{ij} K_{h_\mu}(T_{ij} - \tau) \left( d_{\mathcal{M}}^2(Y_{ij}, y) - F^*(y, \tau) - \partial_\tau F^*(y, \tau)(T_{ij} - \tau) \right).$$

*Under Assumption S.5.1, if $h_\mu \to 0$ and $mh_\mu \gtrsim 1$, then for any fixed $t \in \mathcal{T}$ and $h = O(h_\mu)$,*

(a) $\sup_{\tau \in B(t;h)} |U - \mathbf{E}U| = O_p\left(\sqrt{\frac{1}{n}}\right)$;

(b) $\sup_{\tau \in \mathcal{T}} |\hat{u}_k(\tau)| = O_p(h_\mu^k)$ for $k = 0, 1, 2$;

(c) $\inf_{\tau \in B(t;h)} |\hat{\sigma}_0^2(\tau)| \asymp h_\mu^2(1 + o_P(1))$.

*Proof of Lemma S.5.1.* Define the envelop function

$$H := \frac{2\mathrm{diam}(\mathcal{K})^2}{m} \sum_{j=1}^{m} \sup_{\tau \in B(t;h)} K_{h_\mu}(T_j - \tau).$$

Since $mh_\mu \gtrsim 1$, simple computation leads to $\mathbf{E}(H^2) = O(1)$ and thus

$$\sup_{\tau \in B(t;h)} |U - \mathbf{E}U| = O_p\left(\sqrt{\frac{1}{n}}\right)$$

according to Theorems 2.7.11 and 2.14.2 of van der Vaart and Wellner (1996). Combining above results together, we deduce part (a). Similar technique leads to part (b). For part (c), it is seen that $\hat{\sigma}_0(\tau) \asymp \{\mathbf{E}\hat{u}_0\mathbf{E}\hat{u}_2 - (\mathbf{E}\hat{u}_1)^2\}(1 + o_P(1))$, where the $o_P(1)$ component is uniform over $\tau$. Define $V := K_{h_\mu}(T_{11} - \tau)$ and $W := \mathbf{E}V \asymp 1$. Simple calculation shows that

$$\mathbf{E}\hat{u}_0\mathbf{E}\hat{u}_2 - (\mathbf{E}\hat{u}_1)^2 = W\mathbf{E}\left(V[(T_{11} - \tau) - W^{-1}\mathbf{E}\{V(T_{11} - \tau)\}]^2\right) \asymp h_\mu^2$$

uniformly over all $\tau \in \mathcal{T}$. $\qquad\square$

**Lemma S.5.2** (mean, regular deterministic). *Suppose that Assumptions 2.1, 2.2, 3.1, 4.4, 4.5. Define*

$$U := \frac{1}{nm} \sum_{ij} K_{h_\mu}(T_j - \tau) \left( d_{\mathcal{M}}^2(Y_{ij}, y) - F^*(y, \tau) - \partial_\tau F^*(y, \tau)(T_j - \tau) \right).$$

*Under Assumption S.5.2, if $h_\mu \to 0$ and $mh_\mu > 1/c_1$, then for any fixed $t \in \mathcal{T}$ and $h = O(h_\mu)$,*

(a) $\sup_{\tau \in B(t;h)} |U - \mathbf{E}U| = O_p\left(\sqrt{\frac{1}{n}}\right)$;

(b) $\sup_{\tau \in \mathcal{T}} |\hat{u}_k(\tau)| = O_p(h_\mu^k)$ for $k = 0, 1, 2$;

(c) $\inf_{\tau \in B(t;h)} |\hat{\sigma}_0^2(\tau)| \asymp h_\mu^2$.

*Proof of Lemma S.5.2.* Part (a) can be established by an argument similar to that of Lemma S.4.6 under the condition $mh_\mu \gtrsim 1$. Part (b) can be verified by simple calculation. For part (c), we fix $\tau \in \mathcal{T}$ and let $W := \sum_j K\left(\frac{T_j - \tau}{h_\mu}\right)$. The assumptions on the kernel function imply that $K(u) \geq c_0$ on $[-3/4, 3/4]$ for some constant $c_0 > 0$ depending only on $K$. Assumption S.5.2 implies that there are at least $3c_1 mh_\mu/2$ points within the interval $[\tau - 3h_\mu/4, \tau + 3h_\mu/4]$, from which we deduce that $W \geq 3c_0 c_1 mh_\mu/2 > 0$ regardless of the location of $\tau$. Let $w_j = K\left(\frac{T_j - \tau}{h_\mu}\right)/W$, which is well defined. Observe that $\hat{u}_k(\tau) = \frac{W}{mh_\mu} \sum_j w_j (T_j - \tau)^k$. According to the assumptions on the kernel function and Assumption S.5.2, at least $c_1 mh_\mu/4$ of sampling points $T_j$ satisfy $|T_j - \tau| \geq h_\mu/8$ and $w_j \geq c_0/W$. Thus, for $k = 0, 2$,

$$\hat{u}_k(\tau) \geq \frac{W}{mh_\mu} \frac{c_1 mh_\mu}{4} \frac{c_0}{W} \frac{h_\mu^k}{8^k} \geq \frac{c_0 c_1}{2^{3k+2}} h_\mu^k$$

regardless of the value of $\tau$. Combining this with part (b) we prove part (c).  □

The following theorem provides the pointwise and uniform convergence rates of the covariance estimator under a regular design, where the optimal bandwidth $h_\mu \asymp h_{\mathcal{C}} \asymp \frac{1}{m}$ leads to the same convergence rate as that from Cai and Yuan (2011) in the Euclidean case.

**Theorem S.5.1.** *Suppose that Assumptions 2.1, 2.2, 3.1, 4.4, 4.5 and 4.6 hold. Under either of Assumptions S.5.1 or S.5.2, if $h_\mu \to 0$, $h_{\mathcal{C}} = O(h_\mu)$, $mh_\mu > 1/c_1$, $m^2 h_{\mathcal{C}}^2 > 1/c_1$ and $nh_{\mathcal{C}} \gtrsim 1$, then*

$$\sup_{(s,t) \in \mathcal{T}^2} \left\| \mathscr{P}_{(\hat{\mu}(s),\hat{\mu}(t))}^{(\mu(s),\mu(t))} \hat{\mathcal{C}}(s,t) - \mathcal{C}(s,t) \right\|_{G_{(\mu(s),\mu(t))}}^2 = O_p\left( h_\mu^4 + h_{\mathcal{C}}^4 + \frac{\log n}{n} \right),$$

*and for any fixed $s, t \in \mathcal{T}$,*

$$\left\| \mathscr{P}_{(\hat{\mu}(s),\hat{\mu}(t))}^{(\mu(s),\mu(t))} \hat{\mathcal{C}}(s,t) - \mathcal{C}(s,t) \right\|_{G(\mu(s),\mu(t))}^2 = O_p\left( h_\mu^4 + h_{\mathcal{C}}^4 + \frac{1}{n} \right).$$

The proof of Theorem S.5.1 is similar to that of Theorems 4.1 and 4.2, where Lemmas S.4.2, S.4.3 and S.4.7 are replaced by Lemma S.5.3 below to analyze the parts $S_{ab}$ and $U$. The proof of Lemma S.5.3 is similar to that of Lemmas S.5.1 and S.5.2 and thus omitted.

**Lemma S.5.3** (covariance, regular). *Suppose that Assumptions 2.1, 2.2, 3.1, 4.4, 4.5 and 4.6 hold. Define*

$$U(s,t) := \sum_i \nu_i \sum_{j \neq k} \varpi(T_{ij}, T_{ik}) \left\{ \tilde{\mathcal{C}}_{i,jk} - \mathcal{C}(s,t) - \partial_s \mathcal{C}(s,t)(T_{ij} - s) - \partial_t \mathcal{C}(s,t)(T_{ik} - t) \right\},$$

*where $\varpi(s', t') = K_{h_{\mathcal{C}}}(s - s') K_{h_{\mathcal{C}}}(t - t')$ for $s', t' \in \mathcal{T}$. Under either of Assumptions S.5.1 or S.5.2, if $h_\mu \to 0$, $h_{\mathcal{C}} = O(h_\mu)$, $mh_\mu > 1/c_1$, $m^2 h_{\mathcal{C}}^2 > 1/c_1$, and $nh_{\mathcal{C}} \gtrsim 1$, then*

(a) $\sup_{s,t\in\mathcal{T}}\mathbf{E}\left\{S_{ab}(s,t)\right\}=O(1)$;

(b) $\sup_{s,t\in\mathcal{T}}|S_{ab}(s,t)-\mathbf{E}\left\{S_{ab}(s,t)\right\}|=o_P(1)$;

(c) $\inf_{s,t\in\mathcal{T}}\{(S_{20}S_{02}-S_{11}^2)S_{00}-(S_{10}S_{02}-S_{01}S_{11})S_{10}+(S_{10}S_{11}-S_{01}S_{20})S_{01}\}\asymp 1+o_P(1)$;

(d) $\sup_{s,t\in\mathcal{T}}|U(s,t)-\mathbf{E}U(s,t)|=O_p\left(\sqrt{\frac{1}{n}}\right)$.

## S.6   Additional Illustration of Invariance

The covariance function and its estimator proposed in our paper are invariant to the manifold parameterization, choice of frame and embedding. This important invariance property is a consequence of the intrinsic perspective we take, and below we demonstrate that it is not shared by non-intrinsic statistical methods.

**A method non-invariant to parameterization and frame selection.** An "obvious" estimator for $\mathcal{C}$ might be obtained by utilizing a frame along $\hat{\mu}(\cdot)$ and the coefficient process of Lin and Yao (2019). Specifically, fix a frame along $\hat{\mu}$ which determines an orthonormal basis of $T_{\hat{\mu}(t)}\mathcal{M}$ for each $t\in\mathcal{T}$. Then $\mathrm{Log}_{\hat{\mu}(T_{ij})}Y_{ij}$ can be represented by its coefficient vector $\hat{c}_{ij}$ with respect to the frame, and $\hat{\mathcal{C}}_{i,jk}$ is also represented by the observed coefficient matrix $\hat{c}_{ij}\hat{c}_{ik}^\top$. Local linear smoothing (Yao et al., 2005) or other smoothing methods can be applied on these matrices to yield an estimated coefficient matrix at any pair $(s,t)$ of time points, and the corresponding estimate $\hat{\mathcal{C}}(s,t)$ is recovered from the estimated coefficient matrix and the frame. However, this estimate is not invariant to the frame, i.e., different frames give rise to different estimates $\hat{\mathcal{C}}(s,t)$. As a simple example, consider two frames that coincide on all $T_{\hat{\mu}(t)}\mathcal{M}$ but not on $T_{\hat{\mu}(s)}\mathcal{M}$ and $T_{\hat{\mu}(t)}\mathcal{M}$, and assume that $s,t\notin\{T_{ij}:i=1,\ldots,n,\ j=1,\ldots,m_i\}$. Then the coefficient matrices $\hat{c}_{ij}\hat{c}_{ik}^\top$ with respect of the two frames are identical and thus this "obvious" estimator will produce identical estimated coefficient matrix at the pair $(s,t)$. However, since the two frames differ at $s$ and $t$, the estimates $\hat{\mathcal{C}}(s,t)$ recovered from the estimated coefficient matrix under the two frames are different. In addition, smoothing methods optimize certain objective function of the observations which are the frame-dependent coefficient matrices $\hat{c}_{ij}\hat{c}_{ik}^\top$ in this context, while most objective functions, like sum of squared errors, are not invariant to the frame, and consequently the corresponding estimate is frame-dependent.

We now numerically demonstrate that the above method based on Yao et al. (2005) is not invariant to parameterization and frame selection. For this purpose, we generate data from the two-dimensional sphere $\mathbb{S}^2=\{(x,y,z)\in\mathbb{R}^3:x^2+y^2+z^2=1\}$ with the same setting in Section 5 with sample size $n=100$ and sampling rate $m=10$. Consider the following three frames:

- The frame $(B_1(t)=\frac{\partial\phi}{\partial u},B_2(t)=\frac{\partial\phi}{\partial v})$ derived from the polar parameterization in Equation (12);

- The frame $(B_1^{4\pi}(t),B_2^{4\pi}(t))$ constructed by

$$B_1^{4\pi}(t)=\cos(4\pi t)B_1(t)+\sin(k\pi t)B_2(t),\quad B_2^{4\pi}(t)=\sin(4\pi t)B_1(t)+\cos(4\pi t)B_2(t),$$

  which is a rotated version of $(B_1(t),B_2(t))$;

- The frame $(\tilde{B}_1(t)=\frac{\partial\varphi}{\partial u},\tilde{B}_2(t)=\frac{\partial\varphi}{\partial v})$ derived from the parameterization in Equation (15).

For each of these frames, we apply the method described above to estimate $\mathcal{C}$ under the identical conditions, e.g., with the same logarithmically equidistant grid of bandwidths $h_\mathcal{C}=0.20,0.28,0.40,0.56,0.80$ and known

mean function. If the method were invariant to frames, then we would expect to observe *identical* relative root mean integral square error (rMISE) quantified by

$$\text{rMISE} := \frac{\{\mathbf{E} \int_{\mathcal{T}^2} \|\hat{\mathcal{C}}(s,t) - \mathcal{C}(s,t)\|_G^2 \mathrm{d}s \mathrm{d}t\}^{1/2}}{\{\int_{\mathcal{T}^2} \|\mathcal{C}(s,t)\|_G^2 \mathrm{d}s \mathrm{d}t\}^{1/2}} \tag{S.22}$$

for any fixed bandwidth. The results, presented in Table S.1 and based on 100 independent Monte Carlo replicates, however, show that different frames lead to distinct rMISE for a fixed bandwidth and distinct minimum rMISE over a grid of bandwidths, and thus clearly show that the above method based on Yao et al. (2005) is not invariant to frames.

Table S.1: rMISE under different frames and bandwidths

| rMISE | $h_{\mathcal{C}} = 0.20$ | $h_{\mathcal{C}} = 0.28$ | $h_{\mathcal{C}} = 0.40$ | $h_{\mathcal{C}} = 0.56$ | $h_{\mathcal{C}} = 0.80$ |
|---|---|---|---|---|---|
| $(B_1(t), B_2(t))$ | 25.48% (20.20%) | 22.71% (19.38%) | 20.69% (18.86%) | **19.44% (18.29%)** | 19.94% (17.36%) |
| $(B_1^{4\pi}(t), B_2^{4\pi}(t))$ | 116.32% (34.57%) | 109.13% (30.08%) | 98.24% (23.31%) | **91.38% (24.19%)** | 93.07% (23.65%) |
| $(\tilde{B}_1(t), \tilde{B}_2(t))$ | 49.19% (21.77%) | 46.63% (20.59%) | 43.52% (19.38%) | 39.78% (18.04%) | **36.67% (16.83%)** |

**A method non-invariant to embedding.** We demonstrate that different embeddings for the method of Dai et al. (2020) yield distinct estimates of the covariance function. Consider a plane $\mathcal{M} = (0,1) \times [0,1]$ with the metric inherited from $\mathbb{R}^2$. The underlying population $X$ on $\mathcal{M}$ is $X(t) = (0.25 + 0.5t + Z_1, 0.5 + Z_2)$ where $Z_1, Z_2 \sim \text{Uniform}(-0.1, 0.1)$ and $\mu(t) = (0.25 + 0.5t, 0.5)$. We generate $n = 100$ paths from $X$ and for each path we randomly sample $\text{Poisson}(10) + 2$ observations, where $\text{Poisson}(10)$ is the Poisson distribution with mean parameter 10. Consider the following three isometric embeddings of $\mathcal{M}$ into $\mathbb{R}^3$:

$$\iota_1 : (x,y) \to (x,y,0) \qquad\qquad\qquad\qquad\qquad \text{plane};$$
$$\iota_2 : (x,y) \to (\frac{1}{\pi}\sin(\pi x), \frac{1}{\pi}\cos(\pi x), y) \qquad \text{half cylindrical surface};$$
$$\iota_3 : (x,y) \to (\frac{1}{2\pi}\sin(2\pi x), \frac{1}{2\pi}\cos(2\pi x), y) \qquad \text{cylindrical surface}.$$

For each of these embeddings, we apply the method of Dai et al. (2020) to produce an estimate of $\mathcal{C}$ and calculate rMISE (S.22) of these estimates, where for illustration, we consider logarithmically equidistant grid of bandwidths $h_{\mathcal{C}} = 0.10, 0.14, 0.22, 0.33, 0.50$ and the known true mean function $\mu(t)$. If the method of Dai et al. (2020) were invariant to embeddings, then we would expect to observe identical rMISE for these embeddings for each bandwidth. Table S.2 with the rMISE results based on 100 Monte Carlo simulation replicates, suggesting the opposite, clearly shows that the method of Dai et al. (2020) is not invariant to choices of the frame.

Table S.2: rMISE under different embeddings and bandwidths

| rMISE | $h_{\mathcal{C}} = 0.10$ | $h_{\mathcal{C}} = 0.14$ | $h_{\mathcal{C}} = 0.22$ | $h_{\mathcal{C}} = 0.33$ | $h_{\mathcal{C}} = 0.50$ |
|---|---|---|---|---|---|
| $\iota_1$ | 26.91% (17.21%) | 22.38% (15.72%) | 19.50% (14.76%) | **17.85% (14.49%)** | 19.25% (31.40%) |
| $\iota_2$ | 51.52% (21.16%) | 49.22% (19.54%) | 47.88% (18.48%) | **47.10% (17.92%)** | 54.57% (90.06%) |
| $\iota_3$ | 81.83% (29.78%) | 79.97% (28.08%) | 78.49% (26.78%) | **77.15% (25.40%)** | 90.01% (150.56%) |

**A real data example.** We now demonstrate that different choices of the frame for the extrinsic method based on Yao et al. (2005) lead to distinct statistical results for the real data analyzed in Section 6. Consider

the frame $\{(B_k)\}_{1\le k\le 6}$ induced from the parameterization

$$\phi : (x_1, x_2, x_3, x_4, x_5, x_6) \in \mathbb{R}^6 \to \begin{pmatrix} e^{x_1} & 0 & 0 \\ x_4 & e^{x_2} & 0 \\ x_5 & x_6 & e^{x_3} \end{pmatrix} \begin{pmatrix} e^{x_1} & x_4 & x_5 \\ 0 & e^{x_2} & x_6 \\ 0 & 0 & e^{x_3} \end{pmatrix} \in \mathrm{Sym}_{LC}^+$$

and the frame $\{(\bar{B}_k)\}_{1\le k\le 6}$ derived from a rotation of $\{(B_k)\}_{1\le k\le 6}$ on each tangent space $T_{\hat{\mu}(T_{ij})}\mathrm{Sym}_{LC}^+$ by

$$\{(\bar{B}_k)\}_{1\le k\le 6} = \{(B_k)\}_{1\le k\le 6} \times$$
$$\mathrm{diag}\left\{ \begin{pmatrix} \cos(4\pi T_{ij}) & -\sin(4\pi T_{ij}) \\ \sin(4\pi T_{ij}) & \cos(4\pi T_{ij}) \end{pmatrix}, \begin{pmatrix} \cos(4\pi T_{ij}) & -\sin(4\pi T_{ij}) \\ \sin(4\pi T_{ij}) & \cos(4\pi T_{ij}) \end{pmatrix}, \begin{pmatrix} \cos(4\pi T_{ij}) & -\sin(4\pi T_{ij}) \\ \sin(4\pi T_{ij}) & \cos(4\pi T_{ij}) \end{pmatrix} \right\},$$

where $\mathrm{diag}(M_1, M_2, M_3)$ for matrices $M_1, M_2, M_3$ denotes the block diagonal matrix formed by $M_1, M_2, M_3$. For each of these two frames, we apply the extrinsic method based on Yao et al. (2005) to estimate the covariance function and its eigenfunctions under identical conditions, e.g., with the same estimated mean function (with $h_\mu = 10$) and the same choice of bandwidth $h_{\mathcal{C}} = 20$. Figure S.4, depicting the first three functional principal components obtained from the two frames, clearly shows that the two frames yield distinct estimates. This demonstrates that the above method based on Yao et al. (2005) is not invariant to frames.

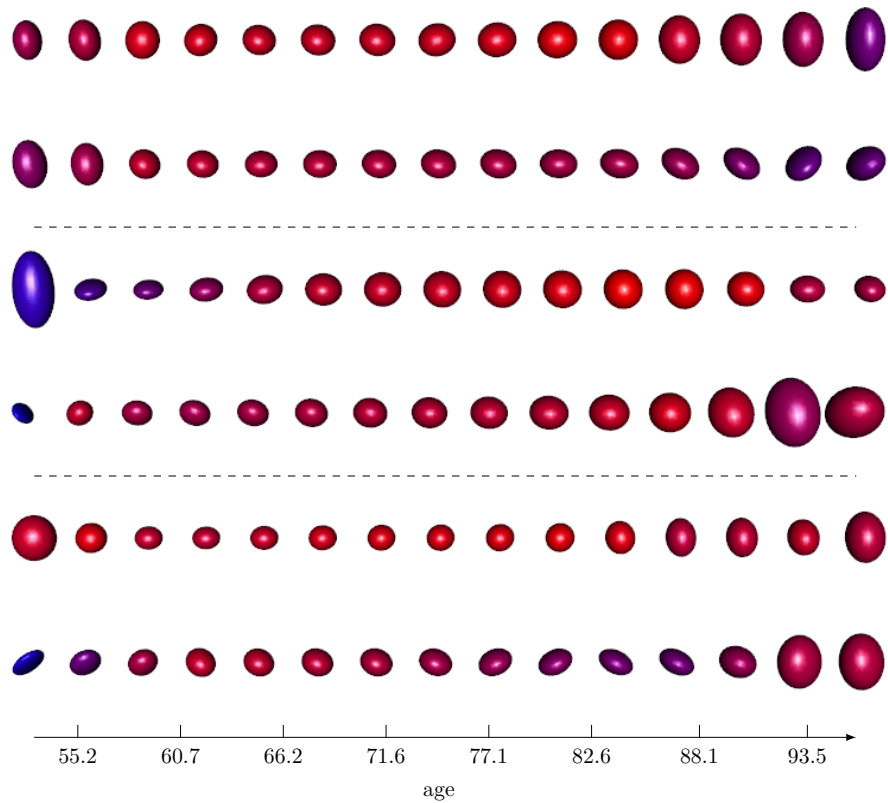

Figure S.4: The first principal component resulting from the frame $\{(B_k)\}_{1 \leq k \leq 6}$ (Row 1) and the frame $\{(\tilde{B}_k)\}_{1 \leq k \leq 6}$ (Row 2), the second principal component resulting from the frame $\{(B_k)\}_{1 \leq k \leq 6}$ (Row 3) and the frame $\{(\tilde{B}_k)\}_{1 \leq k \leq 6}$ (Row 4), and the third principal component resulting from the frame $\{(B_k)\}_{1 \leq k \leq 6}$ (Row 5) and the frame $\{(\tilde{B}_k)\}_{1 \leq k \leq 6}$ (Row 6). The color encodes fractional anisotropy.



\end{document}